%% file: main.tex
\documentclass[11pt]{amsart}
\usepackage[text={6.3in,9in}, centering]{geometry}
\usepackage{cancel}
\allowdisplaybreaks[1]

\usepackage{array}   %
\newcolumntype{L}{>{$}l<{$}} %

\usepackage{tensor}
\usepackage{slashed}
\usepackage{xcolor}

\usepackage{wrapfig}
\usepackage[pdftex]{graphicx}
\usepackage{pdfpages}

\usepackage[utopia,cal=cmcal,greekuppercase=upright]{mathdesign}
 \usepackage[T1]{fontenc}
\usepackage[utf8]{inputenc}

\usepackage{esint}
\usepackage{epstopdf}
\usepackage{amscd}
\usepackage{enumerate}
\usepackage{mymathsty}
\usepackage{mymathsybs}
\usepackage{fancyhdr} 
\usepackage{microtype}
\usepackage[pdfencoding=auto,psdextra]{hyperref}
\usepackage[all]{xy}
\xyoption{poly}
\usepackage[capitalise]{cleveref}
\crefname{lem}{Lemma}{Lemmas}
\Crefname{lem}{Lemma}{Lemmas}
\crefname{thm}{Theorem}{Theorems}
\Crefname{thm}{Theorem}{Theorems}
\crefname{prop}{Proposition}{Propositions}
\Crefname{prop}{Proposition}{Propositions}
\crefname{defn}{Definition}{Definitions}
\Crefname{defn}{Definition}{Definitions}

\usepackage[final]{showlabels}

 \numberwithin{equation}{section}

	\DeclareMathSymbol{\T}{\mathbin}{AMSb}{"54}

	\DeclareMathOperator{\logdet}{logDet}
\begin{document}
 \author{Yang Liu}
 \address{
 Max Planck Institute for Mathematics,
 Vivatsgasse 7,
 53111 Bonn,
 Germany
 }
 \email{yangliu@mpim-bonn.mpg.de}
 \title[Hypergeometric function and Modular Curvature II.  ] {
     Hypergeometric function and Modular Curvature.
     II.  Connes-Moscovici functional relation after Lesch's work
}



\date{\today} 

\begin{abstract}

\input{abs.tex}

\end{abstract}

\maketitle

\tableofcontents

\input{intro.tex}
\input{varprob.tex}
\input{cm_funrel}

\bibliographystyle{hplain}

\bibliography{mylib1}

\end{document}

%% file: abs.tex
As the second part of the sequel, we investigate the variation of rearrangement
operators (more precisely, the spectral functions behind)
arising in the study of modular geometry on noncommutative (two) tori. 
We initiate a systematic approach by introducing transformations corresponding 
to basic operations in calculus, like differentiation and integration by
parts. 
As for applications, we extend, in a uniform way,
the Connes-Moscovici's functional relations on noncommutative two tori
attached to the variation of second heat coefficients
to noncommutative tori of arbitrary  dimension. 
Moreover, those transformations lead to more internal relations among the
hypergeometric family obtained in part I of the sequel, which allows us to obtain,
the first time, a computer-aid free verification of those Connes-Moscovici 
type functional relations.

%% file: intro.tex
\section{Introduction}
\label{sec:intro}

As a continuation of  \cite{Liu:2018aa}, the paper concerns  the variational aspect 
of the $a_2$-coefficient of the heat trace asymptotic on
noncommutative tori and toric noncommutative manifolds
\footnote{a.k.a Connes-Landi deformations or $\theta$-deformations, 
cf. \cite{Connes:2001wv,Connes:2002wh}, \cite{Brain:2013wp}, also 
\cite{Rieffel:1993wb}.
}. The integration of the $a_2$-coefficients allows us to establish the analogue 
of Gauss-Bonnet theorem on noncommutative two tori \cite{Connes:2011tk} 
and \cite{Fathizadeh:2012vg} and the full $a_2$-coefficients 
gives rise to the notion of  scalar curvature in conformal geometry
on noncommutative 
two tori \cite{MR3194491}, \cite{Lesch2016Modular-Curvatu},
also \cite{MR3148618,MR3359018} and later
on toric noncommutative manifolds \cite{LIU2017138,Liu:2015aa}.
The complexity of the calculation increases dramatically for the next coefficient,
the $a_4$-term.  An experimental computation has been achieved in
\cite{2016arXiv161109815C} which is in need of more conceptual understanding and 
further simplification.
To this end, we have accomplished in \cite{Liu:2018aa} 
the first step suggested by Connes, 
namely, to find a  suitable basis of functions to record the spectral functions
arising from the rearrangement process. It means that, for instance, 
one can replace the explicit expressions of the functions $K_1$ to $K_{
20}$  of the $a_4$-term in \cite{2016arXiv161109815C}
by  combinations of the hypergeometric family
$H_\alpha$,  $\alpha \in  \Z_{ \ge 0}^n$ obtained \cite{Liu:2018aa}
\footnote{The author has finished the computation but the result will be postponed
to future publications.},
similar to the form shown in \cref{eq:b2cal-KDeltak,eq:b2cal-HDeltak}.
As we can see from 
\cite[\S9 and Appendix C]{2016arXiv161109815C}, their original form 
generated from a CAS (computer algebra system) takes pages to record.
The next step, which is one of the primary motivations of the paper, 
is to search for simplifications for the functional relations 
cf.  \cite[\S7]{2016arXiv161109815C}.
Back to the $a_2$-term, there are only two functions and the relation 
(\cite[Eq. (0.4)]{MR3194491})  simply reads  
\begin{align}
    \label{eq:intro-CM-original}
    - \frac{1}{2} \tilde H(s_1, s_2) =
    \frac{ \tilde K(s_2) -\tilde K(s_1) }{ s_1+s_2}  
    + \frac{ \tilde K(s_1 + s_2) -\tilde K(s_2) }{s_1} 
    - \frac{ \tilde K(s_1 + s_2) -\tilde K(s_1) }{s_2} 
    .
\end{align}
It reflects the variational nature of modular Gaussian curvature $K_\varphi$
(\cite[Eq. (4.34)]{MR3194491}) and similar results in the paper are 
often referred as Connes-Moscovici type functional relations. 

Even for the $a_2$-term,
when performing similar calculus on higher dimensional examples, 
it is not obvious, at least to the author, that one can obtain such
a simple equation as  \cref{eq:intro-CM-original}. 
After repeating the variational computation several times,
the author finally realized that
the sum of  \cite[Eq. (4.5)-(4.8)]{LIU2017138}
\footnote{The relations have  appeared before in the proof of 
Theorem 5.1 in \cite{MR3369894}. }
indeed admits reduction
to a similar form  of \cref{eq:intro-CM-original}.
The first contribution (\S\ref{sec:varcalwrptncvar}) of 
the paper consists of a systematic framework (built upon the work 
of Lesch \cite{leschdivideddifference})
for the variational calculus for a noncommutative variable. 
The fundamental obstruction is the fact that the variable and its derivatives
do not commute.
There is a rearrangement process that compress the Ansatz to some 
rearrangement operators (or equivalently, to the underlying spectral functions,
cf.  \S\ref{subsec:smooth-funcal}).
To keep track of the change of those functions under basic operations in
variational calculus, we introduce  the following transformations 
\begin{itemize}
    \item  with respect to the log-Weyl factor $h$:
        $\set{
        \pmb\iota ,  \pmb\tau_j, 
        \blacktriangle_{0,j}^+, \blacktriangle^+, \blacktriangle^- } $,

    \item  with respect to the Weyl factor $k = e^h$:
        $\set{
        \pmb\eta, \pmb\sigma_j, \blacksquare_{0,j}^+, \blacksquare^+,
        \blacksquare^- } $.
\end{itemize}
The transformations $\left\{ \blacktriangle_{0,j}^+,
\blacktriangle^+, \blacktriangle^- \right\} $ 
(resp. $ \set{ \pmb\sigma_j, \blacksquare_{0,j}^+, \blacksquare^+,
        \blacksquare^-}$)
 arising from differentiation are generated by divided difference, thanks to 
\cite[\S3.5]{leschdivideddifference}.
The cyclic operator  $ \pmb\tau_j$ (resp. $ \pmb\sigma_j$)
resembles integration by parts 
with respect to the derivation $[ \cdot , h]$ and the volume weight 
$\varphi_j$ in \cref{eq:phi_j-defn}.
A noteworthy new observation 
(Proposition \ref{prop:internal-relations-btri-tau})
is a set of internal relations, 
which explains the cancellation behind \cite[Eq. (4.5)-(4.8)]{LIU2017138}
mentioned above.  
As a consequence, the whole variational calculus can be generated by only three
transformations 
 $\set{
        \pmb\iota ,  \pmb\tau_j,  \blacktriangle^+
    }$ (resp. $\set{
        \pmb\eta, \pmb\sigma_j,  \blacksquare^+
    }$). 
In contrast to classical calculus, formulas under the simple the change of variable 
$h \to k = e^h$ could be quite subtle, again due to the noncommutativity
between $h$ and its derivatives. We keep parallel discussions on purpose.   
For instance, the main result has two versions \cref{thm:gradF-h} 
and \cref{thm:CM-intermsof-k} in terms of $h$ and  $k$ respectively,
which generalizes \cite[Thm. 4.10]{MR3194491}. Later, in 
Prop. \ref{prop:k-2-h-CMrel}, we carry out a direct verification of the 
equivalence of the two versions.

The abstract discussion  in
\S\ref{sec:varcalwrptncvar} is designed for gradient computation
 in the modular geometry on noncommutative tori explained in \S\ref{sec:var-FEH}. 
We first extend the geometric functional
attached to the $a_2$-coefficient from dimension two (cf. \cite[\S4]{MR3194491})
two to higher dimensions.
Their functional gradients can be reached via two routes
\S\ref{sec:closed-formulas-EH-OPS}. 
One relies on variation on the heat trace $\Tr(e^{-t\Delta_\varphi})$ or the
spectral zeta function $\zeta_{\Delta_\varphi}(s)$.
The results are recorded in 
\S\ref{subsec:mdocur-on-T^m-theta} and \ref{subsec:mdocur-on-T^m-theta}.
The other approach ( \S\ref{subsec:EH-action-closedfor} 
and \ref{subsec:OPS-local}) makes use of \cref{thm:gradF-h,thm:CM-intermsof-k}
, which brings in the functional relations.
By equating the two methods, we see that the two coefficients 
(\cref{eq:b2cal-KDeltak,eq:b2cal-HDeltak}) of  $R_{\Delta_\varphi} $ are
subject to the corresponding functional relations inherited from the gradient
$\grad_h F$ (resp. $\grad_k F$).
A crucial technical result that is only quoted without proof
is Prop.  \ref{prop:varprob-scacur-for-Delta-varphi},
the explicit form of $R_{\Delta_\varphi}$, as it is 
 the main focus of part I of the sequel \cite{Liu:2018aa}.

If one is willing to think of  the lengthy computation of the heat coefficients 
as a  physics theory, then
the functional relations (like \cref{eq:intro-CM-original})
derived from variation provides, so far, 
the most robust experiment for testing its validation.
For the $a_4$-coefficient \cite{2016arXiv161109815C},
similar verification is indeed the most sophisticated part hidden behind 
the paper.
The reader is encouraged to go through the \emph{Mathematica} notebook files
attached to  \cite{MR3194491} and \cite{2016arXiv161109815C} to see 
how the tests were actually achieved.
With the better understanding of the variational calculus, 
we are able to provide in \S\ref{sec:ver-of-fun-rel}, the first time,
a computer-aid free verification of the functional relations.
The new input is the compatibility between transformations in 
\S\ref{sec:varcalwrptncvar} and the hypergeometric family $H_\alpha$ 
arising in pseudo-differential calculus.
For instance, the cyclic operator $ \pmb\sigma_j$ indeed  
corresponds to cyclic permutation of the index $\alpha$,
see Prop. \ref{prop:sigma0-acts-on-Habc}.
Furthermore, such ``by hand'' calculation reveals more information.
For example, all the functional relations  share the same root: they 
can be reduced to same the relations among some initial values of the
hypergeometric family, see  Props.
\ref{prop:tobechecked-KDelvs-TDel}, \ref{prop:KDelk-TOPS-tobechecked} and
\ref{prop:varification-CM}. 
Another interesting observation, which is hard to achieve without using the
hypergeometric family, is that the two functions  given in Prop.
\ref{prop:varprob-scacur-for-Delta-varphi} form a continuous family
(i.e. for $m \in [2,\infty)$, $m$ is the dimension)
 of solutions to the Connes-Moscovici type functional relations.

\subsection*{Acknowledgment}
The author would like to thank Alain Connes, Henri Moscovici and Matthias Lesch
for their comments and inspiring conversations, furthermore is greatly indebted
for their consistent support and encouragement on this project.
The author would also like to thank Farzad Fathizadeh for comments and discussions.

 Most of the work was carried out under the  postdoctoral fellowship at
MPIM, Bonn. The author would like to thank MPI for its marvelous working
environment.
The author would also like to thank IHES and Riemann Center for Geometry and
Physics,  Leibniz University, Hannover   for kind support and excellent working
environment during his visit\footnote{IHES: Nov-Dec, 2017.
Leibniz University: Jan-Mar 2019}. 

%% file: varprob.tex
\section{Variation on the Rearrangement Operators}
\label{sec:varcalwrptncvar}

\subsection{Schwartz functional calculus on 
    $\mathcal A^{\bar \otimes n+1}$
}
\label{subsec:smooth-funcal}

We follow \cite[\S3]{leschdivideddifference} to set up notations for the
rearrangement process. 
Let $\mathcal A$ be a $C^*$-algebra.
At algebraic level, a toy model of the rearrangement is given by the
contraction map 
\begin{align}
    \cdot: \mathcal A^{\otimes n+1} \times \mathcal A^{\otimes n}
    & \rightarrow \mathcal A \nonumber\\  
    (a_0, \dots, a_n)  \cdot (\rho_1, \dots, \rho_n) & \mapsto
    a_0 \rho_1 a_1 \cdots \rho_n a_n .
    \label{eq:contraction-defn}
\end{align}
 When looking from right to the left,  we factor out $a_i$'s 
as an operator $(a_0, \dots, a_n) \in  L( \mathcal A^{\otimes  n} , \mathcal A)$.
To extend the construction, we consider the following set of generators of 
those rearrangement operators.
For the scope of the paper, we shall focus only  one element
$h \in  \mathcal A$ and its associated
multiplication operators (at the $l$-th slot), that is, 
when acting on $ \mathcal A^{\otimes  n}$, we have 
\begin{align*}
    h^{(l)} = (1, \ldots, h, \ldots,1) 
    \in  L( \mathcal A^{\otimes  n} , \mathcal A), \, \, \, \, 
    \text{$h$ appears at the $l$-th slot}. 
\end{align*}
labeled by a superscript $l=0,1,\ldots,n$.
For instance, when $n=1$, $h^{(0)}$ and $h^{(1)}$ are simply the left and the
right multiplication respectively. 
We fix such a self-adjoint $h = h^* \in \mathcal A$ 
with its exponential $k = e^h$ and denote by bold font letters 
\begin{align}
    \label{eq:varprob-modop-modder}
  \mathbf  y = \op{Ad}_k = k^{-1}(\cdot) k, \,\,\, \mathbf x = \log
  \mathbf y = -\op{ad}_h = [\cdot,h] 
\end{align}
the corresponding conjugation and commutator operators.
In a similar way, each of them yields $n$ operators 
$\mathbf y_l, \mathbf x_l \in L(\mathcal A^{\otimes  n}, \mathcal A)$
whose subscript $l=1, \ldots, n$  indicates
the operator  only acts on the $l$-th factor of 
elementary tensors, that is, 
\begin{align}
 \mathbf x_l   = -h^{(l-1)} + h^{(l)}, \,\,\,
 \mathbf y_l   
 = e^{-h^{(l-1)}} e^{-h^{(l)}} = (k^{(l-1)})^{-1} k^{(l)}, \, \, 
 l=1, \ldots, n,
    \label{eq:paritla-ad-and-Ad}
\end{align}
and vice versa: 
\begin{align}
\label{eq:substi-relations-k-mod}  
    \begin{split}
    k^{(l)} = (k^{(0)})^{-1} \mathbf y_1 \cdots \mathbf y_l,
    \, \, \, \, 
    h^{(l)} = h^{(0)} + \mathbf x_1  + \cdots +\mathbf x_l.
    \end{split}
\end{align}
The Schwartz  functional calculus can be think of  a way to evaluate
Schwartz functions $f (x_1 , \ldots , x_n)$ at the tuple of operators 
$\brac{ \mathbf x_1, \ldots ,  \mathbf x_n }$
so that it  obeys similar arithmetics as in  classical multi-variable calculus.
To be precise, it is an algebra homomorphism: 
\begin{align*}
    C^\infty(U^{n}) \to L(\mathcal A^{\otimes  n} , \mathcal A): 
    f \mapsto  f (\mathbf x_1, \ldots, \mathbf x_n),
\end{align*}
where $U \defeq U_h \subset \R$ is a bounded open subset containing the spectrum of
$\mathbf x$.
To define the operator, 
we first pick an extension, still denote by  $f$ \footnote{
The functional calculus on depends only the restriction of $f$ on $U^n$, 
(more precisely, on $(\mathrm{spec} (\mathbf x))^{n} \subset \R^n$), 
cf. \cite[\S3.2]{leschdivideddifference}.
}, 
that belongs to the Schwartz space $\mathscr S(\R^n)$,
so that the Fourier transform $\hat f$ exists with the normalization that 
\begin{align*}
    f(x_1,\dots,x_n) = \int_{\R^n} \hat f(\xi_1,\dots,\xi_n) e^{i (x_1\xi_1 + 
            \cdots  + x_n\xi_n
    )} d\xi_1\cdots\xi_n,
\end{align*}
then define
\begin{align}
   f (\mathbf x_1,\dots,\mathbf x_n) = \int_{\R^n}
\hat f(\xi_1,\dots,\xi_n) e^{i (\mathbf x_1\xi_1 + 
            \cdots  + \mathbf x_n\xi_n
    )} d\xi_1\cdots d \xi_n,
    \label{eq:varprob-defn-Kx1xn}
\end{align}
that is, when acting on elementary tensors
$(\rho_1, \ldots, \rho_n) \in  \mathcal A^{\otimes  n}$,
\begin{align*}
    &\, \, 
    f (\mathbf x_1,\dots,\mathbf x_n)
    \brac{ \rho_1 \otimes \cdots \otimes  \rho_n} \\
    = &\, \, 
     \int_{\R^n} \hat f(\xi_1,\dots,\xi_n)
     e^{i\xi_1 \mathbf x} (\rho_1) \cdots 
     e^{i\xi_n \mathbf x} (\rho_n) d\xi \\
= &\, \, 
     \int_{\R^n} \hat f(\xi_1,\dots,\xi_n)
     (\mathbf y)^{i\xi_1 } (\rho_1) \cdots 
    (\mathbf y)^{i\xi_n }  (\rho_n) d\xi. 
\end{align*}
The last line above also gives a definition of the functional calculus 
$f(\mathbf y_1, \dots, \mathbf y_n)$ 
for the modular operators $\mathbf y_l = e^{\mathbf x_l}$, 
namely, by the substitution
\begin{align*}
    f(\mathbf y_1, \dots, \mathbf y_n) \defeq 
    f_{\exp}(\mathbf x_1, \dots, \mathbf x_n)
    \defeq f(e^{\mathbf x_1}, \dots, e^{\mathbf x_n}).
\end{align*}

\begin{defn}
    \label{defn:C(R)-C(R+)}
    For any $n\in \N$,  we will denote by $C_{\mathscr S,h} (\R^n)$
    (resp. $C_{\mathscr S,h}(\R^n_+)$), or simply $C_{\mathscr S} (\R^n)$ and 
    $C_{\mathscr S}(\R^n_+)$ when $h$ is fixed, the
    collection of all the $n$-variable spectral functions $f(x_1, \dots, x_n)$
    such that the functional calculus 
    \begin{align*}
        f(\mathbf x_1, \dots, \mathbf x_n) \;\;
        (\text{resp.} \; f(\mathbf y_1, \dots, \mathbf y_n))
    \end{align*}
    is well-defined.   
\end{defn}

\subsection{Divided differences}

\label{subsec:DividedDiff}
The crucial role of divided differences in the variational calculus was
pointed out by Lesch \cite{leschdivideddifference}.
For a one-variable function $f(z)$, the divided difference can be defined
inductively as below:
\begin{align*}
    f[x_0] & \defeq f(x_0) ;\\
    f[x_0, x_1, \dots, x_n] & \defeq 
    \left( 
        f[x_0, \dots, x_{n-1}] - f[x_1, \dots, x_{n}]
    \right)/ (x_0 - x_n)
\end{align*}
For example, 
\begin{align*}
    f[x_0,x_2] = (f(x_0) - f(x_1)) / (x_0 - x_1). 
\end{align*}
By induction, one derives in general: 
\begin{align}
    f[x_0, x_1, \dots, x_n] = \sum_{l=0}^n  
    f(x_l)  \prod_{s=0, s \neq l}^n (x_l -x_s)^{-1}.
    \label{eq:divided-gen-formula}
\end{align}
In particular, the function $f[x_0, x_1, \dots, x_n]$ is symmetric in all the
arguments. This fact will be quoted  several  times in later discussions.
For multivariable functions, we shall use a subscript to indicate on which
variable the divided difference acts, for example: 
\begin{align*}
    f(z_1,z_2,z_3)
    [x_1, \dots, x_s]_{z_2}
\end{align*}
means the divided difference is taken with respect to the function 
 $f(z_1, \bullet , z_3)$. 
Through out the paper, we fix the variable $z$  as the default choice for the
divided difference operator: $[\bullet, \dots, \bullet] \defeq [\bullet,
\dots, \bullet]_z$.
The following basic properties will also be needed:
\begin{enumerate}
    \item Leibniz rule:
        \begin{align}
            \label{eq:Leibniz-rule}
            (fg)[x_0, \dots,x_n] = f(x_0) g[x_0,\dots,x_n]       
            + f[x_0,\dots,x_n] g(x_n)
        \end{align}
        \item Composition rule:
\begin{align}
    (f[y_0,\dots,y_q, z])[x_0, \dots,x_p]_z = f[y_0,\dots,y_q, x_1, \dots, x_p]
    \label{eq:Composition-rule}
\end{align}
        \item The confluent case: suppose there are $\alpha+1$ copies of $x$
            in the arguments of the divided difference, then:
            \begin{align*}
                f[y,x, \ldots , x] = \frac{1}{\alpha !} 
                \partial_x^\alpha f[y,x] .             
            \end{align*}
\end{enumerate}

 \subsection{The cyclic transformations}
\label{subsec:int-by-parts-mod-derivation}
 Let us keep the notations in
 \cref{eq:varprob-modop-modder,eq:paritla-ad-and-Ad,eq:substi-relations-k-mod}
 and further assume that there exists a tracial functional 
 $\varphi_0:\mathcal A \rightarrow \mathbb{C}$ on the algebra $\mathcal A$ that
 plays the role of integration (volume form functional).
Trace property of $\varphi_0$ leads to the 
integration by parts formula with respect to the derivation $\mathbf x$:
\begin{align}
    \varphi_0( \mathbf x(a) \cdot b) = \varphi_0(a \cdot [(-\mathbf x)(b)]),
    \;\; \forall a,b \in \mathcal A.
    \label{eq:int-by-p-babymod}
\end{align}
Let $j\in\R$ be a real parameter, denote by 
\begin{align}
    \varphi_j(a) \defeq \varphi_0(k^j a) = \varphi_0(e^{jh} a),
    \;\; \forall a \in \mathcal A
    \label{eq:phi_j-defn}
\end{align}
the rescaled  volume functional by the factor $k^j = e^{jh}$. 
We push Eq. \eqref{eq:int-by-p-babymod} further:
\begin{align*}
    &\, \, 
  \varphi_j \left( \mathbf x_n ( \rho_1 \otimes \cdots \otimes \rho_n) \right)  
    =
  \varphi_0
  \left(e^{jh} \rho_1 \cdots \rho_{n-1} \cdot  \mathbf x (\rho_n) \right)  \\
    =&\, \, 
    \varphi_0 \left( (-\mathbf x) (e^{jh} \rho_1 \cdots \rho_{n-1})  \rho_n \right)  
    =
\varphi_0 \left( e^{jh} \cdot 
(-\mathbf x) (\rho_1 \cdots \rho_{n-1}) \cdot  \rho_n \right)  
\\
    =&\, \, 
\varphi_0 \left( e^{jh}
    (-\mathbf x_1 - \cdots - \mathbf x_{n-1} )
    (\rho_1 \otimes  \cdots \otimes  \rho_{n-1})
\cdot   \rho_n \right)  \\
    =&\, \, 
\varphi_j \left( 
    (-\mathbf x_1 - \cdots - \mathbf x_{n-1} )
    (\rho_1 \otimes  \cdots \otimes  \rho_{n-1})
\cdot  \rho_n \right) .  
\end{align*}
The computation suggests that one can reduced the number of arguments 
of the rearrangement operator $f(\mathbf x_1, \dots, \mathbf x_n)$
defined in \cref{eq:varprob-defn-Kx1xn} when the integration $\varphi_j$ 
is applied.
More precisely, we have transformations: 
\begin{align*}
    \pmb\iota:C_{\mathscr S} (\R^n) \rightarrow C_{\mathscr S}(\R^{n-1}), \,\,\,
    \pmb\eta:C_{\mathscr S}(\R_+^n) \rightarrow C_{\mathscr S}(\R_+^{n-1}), \,\,\,
    n=1,2,\dots.
\end{align*}
\begin{lem}
    \label{lem:iota-k-defn}
    Let $f\in C_{\mathscr S} (\R^n)$ be a spectral function with $n$-arguments and $\varphi_j$
    is the rescaled weight in \eqref{eq:phi_j-defn} with $j\in\R$.
    \begin{align}
        \label{eq:iota-contraction}
        \varphi_j\brac{
        f(\mathbf x_1, \dots, \mathbf x_n)
        (\rho_1 \otimes \cdots \otimes  \rho_n)
    } = \varphi_j \brac{
    \pmb \iota(f)(\mathbf x_1, \dots \mathbf x_{n-1})
    (\rho_1 \otimes  \cdots \otimes  \rho_{n-1}) \rho_n
}
    \end{align}
    where the operator $\pmb\iota:C_{\mathscr S}(\R^n) \rightarrow C_{\mathscr
    S}(\R^{n-1})$, $n=1,2,
    \dots$, reduces the number of variable by one 
     by restricting spectral functions onto a hyperplane:
\begin{align}
   \begin{split}
    \pmb \iota(f)(x_1, \dots, x_{n-1}) &= f(x_1, \dots,x_{n-1}, 
    -x_1- \cdots-x_{n-1}), 
   \,\,\, n>1 ,\\
   \pmb\iota(f)(x) &= f(0), \,\,\, n=1.
   \end{split}
   \label{eq:defn-bold-iota}
\end{align}
\end{lem}

\begin{lem}
    \label{lem:eta-h-defn}
    Let $f\in C_{\mathscr S}(\R_+^n)$ be a spectral function with $n$-arguments
    and $\varphi_j$
    is the rescaled weight in \eqref{eq:phi_j-defn} with $j\in\R$.
    \begin{align}
        \varphi_j\brac{
            f(\mathbf y_1, \dots,\mathbf y_n)
            (\rho_1 \otimes  \cdots \otimes  \rho_n)     
        } = \varphi_j \brac{
            \pmb \eta (\mathbf y_1, \dots, \mathbf y_{n-1})
            (\rho_1 \otimes  \cdots \otimes  \rho_{n-1}) \rho_n
        }   
        \label{eq:eta-contraction}
    \end{align}
    where the operator $\pmb\eta :C_{\mathscr S}(\R_+^n) \rightarrow
    C_{\mathscr S}(\R_+^{n-1})$, $n=1,2,
    \dots$, reduces the number of variable by one: for $n>1$, 
    \begin{align}
        \begin{split}
    \pmb    \eta (f)(y_1, \dots,y_{n-1}) &=
        f(y_1, \dots, y_{n-1}, (y_1 \cdots y_n)^{-1}),
        \,\,\, \text{when $n>1$,} \\
\pmb\eta(f)(y) &= f(1),
        \,\,\, \text{when $n=1$}
        .
        \end{split}
   \label{eq:defn-bold-eta}
    \end{align}
\end{lem}

 Notice that \cref{eq:int-by-p-babymod} is the infinitesimal version of 
 the KMS-property of the weight $\varphi_j$:
\begin{align}
    \varphi_j( a \cdot b) = \varphi_j( \mathbf y^j(b) \cdot a ).
    \label{eq:KMS}
\end{align}
It leads to cyclic transformations $ \pmb\tau_j$ and $\pmb\sigma_j$ that 
generate all cyclic permutations on $\rho_1, \ldots,\rho_n$ in 
\cref{eq:tao-int-by-parts} and \cref{eq:sigma-int-by-parts} respectively.
\begin{lem}
    \label{lem:tao-defn}
    Let $f(x_1, \dots, x_n)$ be a function with $n$-arguments and $\varphi_j$
    is the rescaled weight in \eqref{eq:phi_j-defn} with $j\in\R$.
\begin{align}
    \varphi_j\brac{
        f(\mathbf x_1, \dots, \mathbf x_n)
        (\rho_1 \otimes  \cdots \otimes \rho_n) \cdot 
    \rho_{n+1} } = \varphi_j\brac{
        \pmb \tau_j(f)(\mathbf x_1, \dots, \mathbf x_n)
        (\rho_2  \otimes  \cdots \otimes \rho_{n+1}) \cdot  \rho_1 
}
        \label{eq:tao-int-by-parts}
 \end{align}
where
\begin{align}
    \label{eq:tao-defn}
    \pmb \tau_j(f) (x_1, \dots, x_n) = \brac{ e^{-jx_1} 
        f(x_1 , \dots, x_n)    
    }\Big| M_{\mathrm{cyc}}^{(n)},
\end{align}
where the notation $f(\vec x)| M \defeq f(M\cdot \vec x)$ means applying the
linear transformation $M$ onto the arguments $($as a column vector$)$.
The $n\times n$ matrix is given by:
\begin{align}
    \label{eq:Mcycn-defn}
    M_{\mathrm{cyc}}^{(n)}=   \left[
\begin{array}{cccc}
    -1 & \cdots &-1 & -1 \\
 1 & & & 0 \\
  & \ddots & & \vdots \\
  && 1&0\\
\end{array}
\right],
\end{align}
which implements the substitutions:
\begin{align*}
    & x_1 \mapsto -x_1 - \cdots - x_n, 
     x_2 \mapsto x_1, \; x_3 \mapsto x_2, \;
    \dots \; , x_n \mapsto x_{n-1} .  
\end{align*}
\end{lem}
\begin{rem}
Let us abbreviate $M\defeq M_{\mathrm{cyc}}^{(n)}$.
We compute all  $ \pmb \tau_j^l$,  for $1\le l\le n$:
    \begin{align*}
        \pmb \tau_j^l(f) (x_1, \dots, x_n) &=  
        (e^{-jx_1}|M) \cdots   (e^{-jx_1}|M^l)  (f(\vec x)| M^l)  
        \\
        &=  (e^{-jx_1}|M+M^2 + \cdots +  M^l)  (f(\vec x)| M^l),
    \end{align*}
    where $\vec x = (x_1, \dots, x_n)$ and
    \begin{align*}
        e^{-jx_1}|(M+M^2 + \cdots + M^l) = e^{j(x_1+ \cdots + x_{n-l+1})}.       
    \end{align*}
    They implement all  cyclic permutations on $\set{\rho_1, \dots,\rho_{n+1}}$.
\end{rem}
\begin{proof}
    Let $\vec{\mathbf x} = (\mathbf x_1, \dots, \mathbf x_n)$ 
    and $\vec \xi = (\xi_1, \dots, \xi_n)$. We first look at the case in which 
    $j=0$ and the function  $f(\vec x) = e^{i\abrac{\xi, x}}$.
    Since $\varphi_0$ is a trace, or using
    Eq. \eqref{eq:int-by-p-babymod} instead, we compute: 
    \begin{align*}
        &\,\, \varphi_0\left( 
            e^{i\abrac{\xi, \vec{\mathbf x}}} 
            (\rho_1 \otimes  \cdots \otimes  \rho_n) \cdot \rho_{n+1}      
        \right)   
        \\
 =&\,\,
 \varphi_0\left( 
     e^{i\xi_1 \mathbf x}(\rho_1) \cdots 
     e^{i\xi_n \mathbf x}(\rho_n) \cdot \rho_{n+1}
 \right)     \\
 =&\,\,
 \varphi_0\left( 
     \rho_1 \cdot e^{-i\xi_1 \mathbf x}\left[ 
     e^{i\xi_2 \mathbf x}(\rho_2) \cdots 
     e^{i\xi_n \mathbf x}(\rho_n) \cdot \rho_{n+1}
     \right]
 \right)
 \\
 =&\,\,
 \varphi_0\left( 
     \left[ 
     e^{i(\xi_2 - \xi_1) \mathbf x}(\rho_2) \cdots 
     e^{i(\xi_n - \xi_1) \mathbf x}(\rho_n) 
     e^{-i\xi_1 \mathbf x}(\rho_{n+1}) 
     \right]
 \cdot \rho_1
 \right) \\
 =&\,\,
 \varphi_0\left( 
     e^{
     i(\xi_2 - \xi_1) \mathbf x_1 + \cdots
     i(\xi_n - \xi_1) \mathbf x_{n-1}  - i\xi_1 \mathbf x_{n}    
     }
     (\rho_2 \otimes  \cdots \otimes  \rho_{n+1}) \cdot \rho_1  
 \right) 
\\
= &\,\,
\varphi_0\left( 
    e^{i \abrac{\vec \xi, M \mathbf{\vec x}}}
     (\rho_2 \otimes \cdots \otimes  \rho_{n+1}) \cdot \rho_1  
 \right),
    \end{align*}
    where the matrix $M \defeq M_{\mathrm{cyc}}^{(n)}$ is defined in
    Eq. \eqref{eq:Mcycn-defn}. 
    The general case follows from the Schwartz functional calculus 
   in \cref{eq:varprob-defn-Kx1xn}:
    \begin{align*}
        \varphi_0\left( 
        f(\vec{\mathbf x}) (\rho_1 \otimes \cdots \otimes  \rho_n) \cdot \rho_{n+1}    
        \right) & = 
        \int_{\R^n} \hat f(\vec \xi) 
        \varphi_0\left( 
            e^{i\abrac{\vec \xi, \vec{\mathbf x}}}
(\rho_1 \otimes \cdots \otimes  \rho_n) \cdot \rho_{n+1}
        \right) d\vec \xi \\
        &=
        \int_{\R^n} \hat f(\vec \xi) 
        \varphi_0
\left( 
    e^{i \abrac{\vec \xi, M \mathbf{\vec x}}}
     (\rho_2 \otimes  \cdots \otimes  \rho_{n+1}) \cdot \rho_1  
 \right)
        d\vec \xi \\
&=   
\varphi_0 \left( 
\left( 
        \int_{\R^n} \hat f(\vec \xi) 
    e^{i \abrac{\vec \xi, M \mathbf{\vec x}}} d\vec \xi
\right) 
     (\rho_2 \otimes \cdots \otimes  \rho_{n+1}) \cdot \rho_1  
\right) 
\\
&= \varphi_0\left( 
    f( M \cdot \vec{\mathbf x}) 
     (\rho_2 \otimes  \cdots \otimes  \rho_{n+1}) \cdot \rho_1  
\right)
    \end{align*}
    So far, we have proved   \eqref{eq:tao-defn} for the tracial weight
    $\varphi_0$, that is for $\pmb\tau_0$. For  $j \neq 0$,
    \begin{align*}
        \varphi_j\left( 
             f(\vec{\mathbf x}) 
(\rho_1 \otimes    \cdots \otimes  \rho_n)
\cdot \rho_{n+1}
        \right)  & =
\varphi_0\left( 
    e^{jh}
             f(\vec{\mathbf x}) 
(\rho_1 \otimes \cdots \otimes  \rho_n) \cdot \rho_{n+1}
        \right)  
  \\      
     & =   
     \varphi_0 \left( 
            f(\vec{\mathbf x}) 
(\rho_1 \cdots \rho_n) \cdot (\rho_{n+1} e^{jh}) 
        \right)  
\\
&= \varphi_0\left( 
    f( M \cdot \vec{\mathbf x}) 
    (\rho_2 \otimes  \cdots \otimes (\rho_{n+1} e^{jh}) ) \cdot \rho_1  
\right)
\\
&= \varphi_0\left( 
     e^{jh} 
     e^{j(\mathbf x_1+ \cdots + \mathbf x_n)}
     f( M \cdot \vec{\mathbf x}) 
    (\rho_2 \otimes \cdots \otimes  \rho_{n+1} ) \cdot \rho_1  
\right),
    \end{align*}
  where is spectral function $
     e^{j( x_1+ \cdots +  x_n)}
     f( M \cdot \vec{ x}) 
     $ agrees with the right hand side of \eqref{eq:tao-defn}.  
\end{proof}
After the substitution $y =e^x$, we obtain the parallel version in terms of the
Weyl factor $k$ and the modular operator $\mathbf y$:
\begin{lem}
    \label{lem:taoexp-defn}
    Let $f\in C_{\mathscr S } (\R^n_+)$ and $\varphi_j$,  $j\in\R$,
    be the rescaled volume weight. 
    For all $\rho_1, \dots, \rho_{n+1} \in \mathcal A$,
    \begin{align}
        \varphi_j\brac{
            f(\mathbf y_1, \dots, \mathbf y_n)
            (\rho_1 \otimes  \cdots \otimes  \rho_n) \cdot \rho_{n+1}
        } = 
        \varphi_j\brac{
            \pmb \sigma_j (f)   (\mathbf y_1, \dots, \mathbf y_n)
        (\rho_2 \otimes  \cdots \otimes  \rho_{n+1}) \cdot \rho_{1}
        }
        \label{eq:sigma-int-by-parts}
    \end{align}
 where
 \begin{align}
    \label{eq:taoexp-defn}
     \pmb \sigma_j (f) (y_1, \dots, y_n) = 
     (y_1 \cdots y_n)^{-j} f((y_1 \cdots y_n)^{-1}, y_1, \dots,y_{n-1}).
 \end{align}
\end{lem}

 \subsection{Taylor expansion and the first variation }
\label{subsec:diff-mod-action}

 Recall the exponential expansion:
 \begin{align*}
     e^{a+b}  =e^a + \sum_{n=1}^\infty 
     \int_{0\le s_n \le \cdots \le s_1\le 1 }
     e^{(1-s_1)a} \cdot b \cdot e^{(s_1 - s_2)a} \cdot b \cdot \cdots
     b \cdot e^{ s_{n}a}  ds.
 \end{align*}
 The integrand  of summands above can be rewritten
 (using notations in \S\ref{subsec:smooth-funcal}) as:
 \begin{align*}
     \exp\left( 
         (1-s_1)a^{(0)} + (s_1 - s_2)a^{(1)} + \cdots s_n a^{(n)}
     \right) \cdot (b \otimes  \cdots \otimes  b).    
 \end{align*}
 The Genocchi–Hermite formula turns integration over the standard $n$-simplex
 into a divided difference:
 \begin{align*}
     \int_{0\le s_n \le \cdots \le s_1\le 1 }
     e^{(1-s_1)a} \cdot b \cdot e^{(s_1 - s_2)a} \cdot b \cdot \cdots
     b \cdot e^{ s_{n}a}  ds = 
     e^{z}[a^{(0)}, \cdots , a^{(n)}]_z.
 \end{align*}
For any self-adjoint elements $a,b \in \mathcal A$ and a function $f$ whose 
Schwartz functional calculus (cf. Eq. \eqref{eq:varprob-defn-Kx1xn}) makes sense,
we have the noncommutative Taylor expansion (cf. 
\cite[Prop.  3.7]{leschdivideddifference})
\begin{align}
    \label{eq:nctaylor-exp-f}
    f(a+b) \backsim_{b\rightarrow 0}
   \sum_{n=0}^\infty  f[a^{(0)} , \dots , a^{(n)}] \cdot 
        (b \otimes \cdots \otimes  b).
\end{align}

Let $\delta: \mathcal A \to \mathcal A$ be a derivation with the associated 
one-parameter group of automorphisms $\alpha_t: \mathcal A \to \mathcal A$, $t
\in  \R$:
 \begin{align}
     \delta(a) = \frac{d}{dt}\Big|_{t=0} 
     \alpha_t(a), \,\,\, a\in \mathcal A.
     \label{eq:derivations-as-derivatives}
 \end{align}
 As a warm up, we compute the derivatives of the exponential function.
 Put $a=h$ and  $b = \alpha_t(h)$ in the exponential expansion, we obtain 
  the Duhamel's formula in terms of divided difference:
  \begin{align*}
      \delta( e^h)
      &= \int_0^1 e^{(1-u) h} \delta(h) e^{uh} du = e^h 
      \int_0^1 e^{ u \mathbf x}  du (\delta(h)) = 
    e^h  \frac{e^{u \mathbf x}}{ \mathbf x} \Big |_{u=0}^1 (\delta(h)) 
    \\
    &= 
    e^h \frac{e^{\mathbf x} -1}{ \mathbf x} (\delta(h))
   = e^h \exp[0,\mathbf x] (\delta(h)).
 \end{align*}

\begin{lem}
     \label{lem:1st-derivative-exp-power}
     Let $j\in\R$ and $h = h^* \in \mathcal A$ be a log-Weyl factor and $\delta$ be
     a derivation as in eq.
     \eqref{eq:derivations-as-derivatives}.
     The first derivative of the exponential is  given by:
     \begin{align}
\label{eq:del-e^jh}
         \delta(e^{jh})  &= e^{jh} G_{\exp}^{(1)} (\mathbf x;j)(\delta(h)) 
         .
     \end{align}
     In terms of the Weyl factor $k =e^h$:
     \begin{align}
         \delta(k^j)  &= k^{j-1}  G_{\mathrm{pow}}^{(1)}(\mathbf y;j)(\delta(k)) 
         .
\label{eq:delk^j}
     \end{align}
The spectral functions are given in terms of the divided differences of the
exponential and power functions:
\begin{align}
    G_{\exp}^{(1)} (x;j) = e^{jz}[0,x]_z, \,\,\,
    G_{\mathrm{pow}}^{(1)}   (y;j) = z^j[1,y]_z.
         \label{eq:Gexp1+Gpow1}
\end{align}
 \end{lem}
 \begin{rem}
     With $y = e^x$, we observe that
     \begin{align}
         G_{\mathrm{pow}}^{(1)} (y;j) = G_{\exp}^{(1)}(x;j) 
         (G_{\exp}^{(1)} (x;1))^{-1} =
         e^{jz}[0,x]_z ( \exp[0,x])^{-1}.
         \label{eq:Gpow-vs-Gexp}
     \end{align}
 \end{rem}
 \begin{proof}
     \cref{eq:del-e^jh} is simply the Duhamel's formula proved above 
     with $h \mapsto  jh$, see also \cite[Example 3.9]{leschdivideddifference}.
     As a consequence, we can solve for $\delta(h)$ (with $j=1$) in terms of 
     $\delta(k)$: 
     \begin{align}
         \label{eq:delh-to-delk}
         \delta(h) = k^{-1} (\exp[0,\mathbf x])^{-1} (\delta(k))
         .
     \end{align}
     \cref{eq:delk^j} follows quickly: for $j\in\R$,
 \begin{align*}
     \delta(k^j) &= \delta(e^{jh}) = e^{jh} 
     ( e^{jz}[0, \mathbf x]_z (\delta(h)) ) 
     = k^j e^{jz}[0, \mathbf x]_z \left( 
          k^{-1}(\exp[0,\mathbf x])^{-1}(\delta(k))
     \right) \\
     &= k^{j-1} \left( 
             e^{jz}[0, \mathbf x]_z \exp[0,\mathbf x])^{-1}
     \right) (\delta(k)) =
     k^{j-1} (z^j[0,\mathbf y]_z (\delta(k))) .
 \end{align*}
 \end{proof}

 Recall that for fixed $h$, the associated commutator and conjugation
 $ \mathbf x , \mathbf y $ belong to $\mathcal A^{\otimes 2} \subset
L(\mathcal A, \mathcal A)$
 according to \cref{eq:paritla-ad-and-Ad}. One can derive the Taylor
 expansion below by applying \cref{eq:nctaylor-exp-f} to $\mathcal A^{\otimes  2}$.
\begin{prop} [\cite{leschdivideddifference}, Prop. 3.11]
    \label{prop:varprob-Texp-in-logk}
Let $h = h^* \in \mathcal A$ be a log-Weyl factor and   $\mathbf x
= -\mathrm{ad}_h =[\cdot ,h]$ be the corresponding modular derivation. 
Let $f(x) \in C(\R)$, given a self-adjoint perturbation: $h \rightarrow h+b$
with $b=b^*$,
we have the Taylor expansion for the modular action $f(\mathbf x_{h+b})$ upto
the first order:
\begin{align} 
    \label{eq:varprob-Texp-in-logk}
    f(\mathbf x_{h+b})(\rho)
    & = f(\mathbf x)(\rho) 
- (f[\mathbf x_1+ \mathbf x_2, \mathbf x_2])  (b \cdot \rho) 
\\
&+ (f[\mathbf x_1 + \mathbf x_2, \mathbf x_1])  (\rho \cdot b) +  o(b),
    \nonumber
    \end{align}
    as $b \rightarrow 0$ and $\forall \rho \in \mathcal A$. 
\end{prop}
\begin{proof}
    See \cite[\S3.5]{leschdivideddifference}.
\end{proof}

The first variation of the rearrangement operator $f(\mathbf x)$ 
follows immediately from the Taylor expansion \cref{eq:varprob-Texp-in-logk}.
\begin{cor}
        \label{cor:blacktri+-j-defn}
    For $f(x) \in C_{\mathscr S} (\R)$,  $\rho \in \mathcal A$,
and a derivation $\delta$ given in \cref{eq:derivations-as-derivatives},
    we have
\begin{align}
     \delta( f(\mathbf x)) (\rho) & = 
     [\delta, f(\mathbf x)](\rho) =
\delta \brac{ f(\mathbf x)(\rho) } - f(\mathbf x)( \delta(\rho))
 \nonumber \\
 &=
   (f[\mathbf x_1+ \mathbf x_2, \mathbf x_1])  ( \rho \otimes  \delta(h) ) 
  - (f[\mathbf x_1+ \mathbf x_2, \mathbf x_2])  ( \delta(h) \otimes  \rho) 
\label{eq:[delta,f]}
.
\end{align}
With new notations:
$\blacktriangle^{\pm}: C_{\mathscr S}(\R) \rightarrow C_{\mathscr S}(\R^2)$ 
\begin{align*}
    \blacktriangle^+(f)(x_1,x_2) = f[x_1, x_1+x_2], \;\;
    \blacktriangle^-(f)(x_1,x_2) = f[x_2, x_1+x_2], \;\; 
\end{align*}
Eq.  \eqref{eq:[delta,f]} can be rewritten as 
\begin{align*}
    \delta(f(\mathbf x)) = 
    [\delta, f(\mathbf x)] =
    \blacktriangle^+(f)(\mathbf x_1 , \mathbf x_2) \delta(h)^{(1)}
    -\blacktriangle^-(f)(\mathbf x_1 , \mathbf x_2) \delta(h)^{(0)} 
    .
\end{align*}
\end{cor}
We also need a similar result for
$[\delta, e^{jh} f(\mathbf x)]$, which has an extra component coming from 
$\delta(e^{jh}) = e^{jh}  G^{(1)}_{\exp }(\mathbf x;j)(\delta(h))$. 
\begin{cor}
        \label{cor:blacktri0-j-defn}
    For $j\in\R$, $f(x) \in C_{\mathscr S} (\R)$,
    we have
    \begin{align}
        \label{eq:del-ejh*f(x)}
        \delta\brac{ e^{jh} f(\mathbf x)}
        &=[\delta, e^{jh} f(\mathbf x)]
      \\
       &= e^{jh^{(0)}}( \blacktriangle_{0,j}^+(f) - \blacktriangle^-(f))
      (\mathbf x_1, \mathbf x_2) (\delta h)^{(0)} 
      + e^{jh^{(0)}}\blacktriangle^+(f)(\mathbf x_1, \mathbf x_2) (\delta h)^{(1)}
      , \nonumber
    \end{align}
    where the operator $\blacktriangle_{0,j}^+:C(\R) \rightarrow C(\R^2)$: 
    \begin{align}
        \label{eq:blacktri0-j-defn}
        \blacktriangle_{0,j}^+(f)(x_1, x_2) =
        G^{(1)}_{\exp}(\mathbf x_1;j) f(x_2),
    \end{align}
    where $ G^{(1)}_{\exp}(x;j)$ is defined in  \cref{eq:Gexp1+Gpow1}.
\end{cor}
\begin{proof}
   For any $\rho \in \mathcal A$, 
   \begin{align*}
       [\delta, e^{jh}] f(\mathbf x)(\rho) = 
 e^{jh}   G^{(1)}_{\exp }(\mathbf x;j)
        (\delta(h)) f(\mathbf x) (\rho) = 
       e^{jh^{(0)}}  G^{(1)}_{\exp }(\mathbf x_1;j) f(\mathbf x_2)
       (\delta(h) \otimes  \rho).
   \end{align*}
 The rest follows from 
 \begin{align*}
     [\delta, e^{jh}  f(\mathbf x)]  = 
      [\delta, e^{jh}] f(\mathbf x) + e^{jh} [\delta, f(\mathbf x)]  
 \end{align*}
 where $[\delta, f(\mathbf x)]$ have been computed in the previous result.
\end{proof}

Let us move on to the multiplicative version (regarding to the change of
variable $h \to k = e^h $): 
\begin{align*}
    \blacksquare_{0,j}^+, \blacksquare^+, \blacksquare^-:
    C_{\mathscr S} (\R_+) \rightarrow C_{\mathscr S} (\R^2_+),
\end{align*}
of the operators 
$\set{ \blacktriangle_{0,j}^+, \blacktriangle^+,\blacktriangle^-}$.
They all increase the number of arguments by one  via divided difference.
\begin{lem}
    \label{lem:bsqpm-defn}
    Let $f\in C_{\mathscr S}(\R_+)$ and
    $\delta$ be a derivation on $ \mathcal A$,
    \begin{align}
        &\;\; \delta(k^j f(\mathbf y)) = [\delta,k^j f(\mathbf y)]  \nonumber \\ 
        =&\;\;
        (k^{j-1})^{(0)} \brac{
         \blacksquare_{0,j}^+(f) - \blacksquare^-(f) 
} (\mathbf y_1, \mathbf y_2) (\delta k)^{(0)} 
        + (k^{j-1})^{(0)} \blacksquare^+(f) (\mathbf y_1, \mathbf y_2) 
        (\delta k)^{(1)}, 
        \label{eq:[delta-f(y)]}
    \end{align}
    where
    \begin{align}
        \label{eq:bsqs-defn}
        \begin{split}
      &  \blacksquare_{0,j}^+(f)(y_1, y_2) = f(y_2) (z^j[0,y_1]_z), \\
      &  \blacksquare^+(f)(y_1 , y_2) = f[y_1, y_1 y_2], \;\;     
        \blacksquare^-(f)(y_1 , y_2) =y_2 (f[y_2, y_1 y_2]).
        \end{split}
    \end{align}
\end{lem}
\begin{proof}
    For $f \in C_{\mathscr S } (\R_+)$, denote
    $f_{\exp} = f\circ \exp \in C_{\mathscr S }(\R)$.
    For any $\rho \in
    \mathcal A$, we apply \eqref{eq:del-ejh*f(x)}:
    \begin{align*}
        &\,\,     [\delta,k^j f(\mathbf y)](\rho) =
        [\delta, e^{jh} f_{\exp}(\mathbf x)](\rho) \\
        =&\,\,
        e^{jh}   \left( 
        \blacktriangle^+_{0,j}(f_{\exp}) -
        \blacktriangle^- (f_{\exp}) 
        \right)
        (\mathbf x_1, \mathbf x_2)
        (\delta(h) \otimes  \rho) +
     e^{jh}   \blacktriangle^+ (f_{\exp}) 
        (\mathbf x_1, \mathbf x_2)
        (\rho  \otimes  \delta(h) ).
    \end{align*}
    It remains to replace all $\delta(h)$ by $\delta(k)$ via Eq. 
    \eqref{eq:delh-to-delk}. Let us do it one by one:
    \begin{align*}
        \blacktriangle^+_{0,j}(f_{\exp}) (\mathbf x_1, \mathbf x_2) 
        (\delta(h) \otimes  \rho) &= k^{-1}
        e^{jz}[0,\mathbf x_1]_z  (\exp[0,\mathbf x_1])^{-1}
f_{\exp} (\mathbf x_2)
        (\delta(k) \otimes  \rho) \\
        &=
        k^{-1} z^{j}[1,\mathbf y_1]_z
        f(\mathbf y_2)
        (\delta(k) \otimes  \rho) \\
        &= k^{-1} \blacksquare^+_{0,j}(f)(\mathbf y_1, \mathbf y_2)
        (\rho \otimes  \delta(k) ) ,
    \end{align*}
    where we have used  Eq. \eqref{eq:Gpow-vs-Gexp}. For the second term,
    \begin{align*}
        \blacktriangle^-(f_{\exp}) (\mathbf x_1, \mathbf x_2) 
        (\delta(h) \otimes  \rho )&= k^{-1}
        f_{\exp}[\mathbf x_2, \mathbf x_1 + \mathbf x_2]
        (\exp[0,\mathbf x_1])^{-1}
(\delta(k) \otimes \rho )\\
&= k^{-1} \mathbf y_2(f[\mathbf y_2, \mathbf y_1 \mathbf y_2]) 
(\delta(k) \otimes  \rho )\\
      &= k^{-1} \blacksquare^-(f)(\mathbf y_1, \mathbf y_2)
      (\rho \otimes  \delta(k) ).
    \end{align*}
    At last,
    \begin{align*}
        \blacktriangle^+(f_{\exp}) (\mathbf x_1, \mathbf x_2) 
        (\rho \delta(h))&= k^{-1}
        f_{\exp}[\mathbf x_1, \mathbf x_1 + \mathbf x_2]
        e^{-\mathbf x_1} (\exp[0,\mathbf x_2])^{-1} (\rho \otimes  \delta(k)) \\
        &= k^{-1} f[\mathbf y_1, \mathbf y_1 \mathbf y_2]
        (\rho \otimes  (\delta(k))     
      \\
      &= k^{-1} \blacksquare^+(f)(\mathbf y_1, \mathbf y_2)
      (\rho \otimes  \delta(k) ).
    \end{align*}
\end{proof}

So far, we have obtained 
Lemma \ref{lem:tao-defn} 
 , Corollary \ref{cor:blacktri+-j-defn} and
 \ref{cor:blacktri0-j-defn}, respectively in Lemma \ref{lem:taoexp-defn} and
 \ref{lem:bsqpm-defn}, two sets of transformations:
\begin{align*}
    \set{\pmb\tau_j, \blacktriangle_{0,j}^+, \blacktriangle^+,\blacktriangle^-} 
    \,  \,  \,  \text{vs.}\,    \,  \,   
 \set{\pmb\sigma_j, \blacksquare_{0,j}^+,\blacksquare^+,\blacksquare^-}
\end{align*}
linked via the change of coordinate $y =e^x$. 
For $f\in C_{\mathscr S} (\R^n_+)$, we denote the change of variable 
by $f_{\exp} = f \circ \exp$, that is
\begin{align*}
    f_{\exp}(x_1, \dots, x_n) \defeq f(e^{x_1}, \dots, e^{x_n}) 
    = f(y_1, \dots, y_n),
\end{align*}
where $y_l= e^{x_l}$, $l=1,\dots,n$. 
\begin{prop}
    \label{prop:compare-blacksq-blacktri}
    Keep notations as above. For $f\in C_{\mathscr S}(\R_+)$, 
    $y =e^x$, $y_l = e^{x_l}$ with $l=1,2$, we have
    \begin{align}
        \begin{split}
            \blacksquare_{0,j}^+(f)( y_1 , y_2) &= 
        \blacktriangle_{0,j}^+(f_{\exp})(x_1, x_2) (\exp[0,x_1])^{-1} , \\
        \blacksquare^-(f)( y_1 , y_2) &= 
        \blacktriangle^-(f_{\exp})(x_1, x_2) (\exp[0,x_1])^{-1} , \\
        \blacksquare^+(f)( y_1 , y_2) &= 
        \blacktriangle^+(f_{\exp})(x_1, x_2) (e^{x_1}\exp[0,x_2])^{-1}.  
        \end{split}
        \label{eq:comparison-1}
    \end{align}
    For $ f \in C_{\mathscr S}(\R_+)$ and $\tilde f \in C_{\mathscr S}(\R^2_+)$,
    \begin{align} 
        \label{eq:comparison-2}
        \pmb\sigma_j(f) (y)= 
        \pmb\tau_j(f_{\exp}) (x), \;\;
        \pmb\sigma_j(\tilde f) (y_1 , y_2)= 
        \pmb\tau_j(\tilde f_{\exp}) (x_1, x_2).
\end{align}
\end{prop}
\begin{proof}
    The first set of comparison Eq. \eqref{eq:comparison-1} is a byproduct of
    the proof of Lemma \ref{lem:bsqpm-defn}. The verification of
    \eqref{eq:comparison-2} is straightforward.
\end{proof}

\subsection{Reduction relations}
\label{subsec:reduction-relations}
We now have arrived at the first key result of the paper. 
An important takeaway from the relations
described below is the fact that the variational calculus on the 
rearrangement operators 
is generated by two transformations:
$\set{\pmb\tau_j, \blacktriangle^+}$
(resp. $\set{\pmb\sigma_j, \blacksquare^+}$). 
\begin{prop}
    \label{prop:internal-relations-btri-tau}
    As operators on $C_{\mathscr S} (\R)$, we have: 
    \begin{align}
    \label{eq:btri-to-tau}
    (\blacktriangle_{0,j}^+ - \blacktriangle^- )(f) (x_1, x_2)=
    (\pmb \tau_j \cdot \blacktriangle^+ \cdot \pmb \tau_j    )
        (f) (x_1, x_2)
        .
    \end{align}
    The multiplicative version reads: 
    $\forall f \in C_{\mathscr S}(\R_+)$:
\begin{align}
    \label{eq:bsq-to-tau}
(\blacksquare_{0,j}^+ - \blacksquare^- )(f) ( y_1,y_2 )=
(\pmb \sigma_{j-1} \cdot \blacksquare^+ \cdot \pmb \sigma_j)(f) 
( y_1,y_2 )
.
\end{align}
\end{prop}
\begin{proof}
    We only prove \eqref{eq:bsq-to-tau} as an example. 
Recall that on $C_{\mathscr S}(\R_+)$ and $C_{\mathscr S}(\R^2_+)$, 
    the cyclic operator $\pmb \sigma_j$ is given by:
    \begin{align*}
        \pmb\sigma_j(f)(y) = y^{j} f(y^{-1}), \;\;
        \pmb\sigma_j(f)(y_1,y_2) = (y_1y_2)^{j} 
        f((y_1y_2)^{-1} , y_1).
    \end{align*}
    For any $f\in C_{\mathscr S}(\R_+)$,
    \begin{align*}
    &\;\;    (\pmb \sigma_j \cdot \blacksquare^+ \cdot \pmb \sigma_j)(f)     
        (y_1, y_2) \\
        =&\;\; 
        (y_1 y_2)^j 
    (\blacktriangle^+\cdot\pmb \sigma_j)(f)( (y_1y_2)^{-1}, y_1) 
 =(y_1 y_2)^j 
 \pmb\sigma_j(f)[ (y_1y_2)^{-1}, y_2^{-1}] \\
=&\;\; 
(y_1 y_2)^j 
\frac{
    \pmb\sigma_j(f) ((y_1y_2)^{-1}) - \pmb \sigma_j(f)( y_2^{-1})
}{ (y_1y_2)^{-1} - y_2^{-1}}
=
\frac{
    \pmb\sigma^2_j(f) (y_1y_2) - y_1^{j}\pmb \sigma^2_j(f)( y_2)
}{ (y_1y_2)^{-1} - y_2^{-1}}
\\
=&\;\;
(y_1 y_2)
\frac{
    f (y_1y_2) - y_1^{j}  f(y_2)
}{1- y_1},
    \end{align*}
    here we have used the fact that $\pmb\sigma_j$ is of order two when acting
    on one-variable functions $C_{\mathscr S}(\R_+)$.
    Observe that $\pmb\sigma_j = (y_1 y_2) \pmb\sigma_{j-1}$ 
    on $C_{\mathscr S}(\R_+^2)$, thus:
    \begin{align*}
        (\pmb \sigma_{j-1} \cdot \blacksquare^+ \cdot \pmb \sigma_j)(f)     
        (y_1, y_2) 
     =    
\frac{
    f (y_1y_2) - y_1^{j}  f(y_2)
}{1- y_1},
    \end{align*}
    On the other hand,
    \begin{align*}
    (\blacksquare_{0,j}^+ - \blacksquare^-)(f)(y_1, y_2) &= 
    \frac{y_1^j -1 }{ y_1 -1} f(y_2) - 
    \frac{ f(y_1y_2) - f(y_2)}{y_1 -1} \\
    &= \frac{
        y_1^j f(y_2) - f(y_1y_2)   
    }{y_1-1} ,  
\end{align*}
which agrees with $(\pmb \sigma_{j-1} \cdot \blacksquare^+ \cdot \pmb \sigma_j)
(f)$ obtained above.
The proof of \eqref{eq:bsq-to-tau} is complete.

\end{proof}

In fact,  \eqref{eq:bsq-to-tau} and 
\eqref{eq:btri-to-tau} are equivalent due to the correspondence in Prop.
\ref{prop:compare-blacksq-blacktri}. For instance, let us assume 
\eqref{eq:btri-to-tau} and would like to derive \eqref{eq:bsq-to-tau}. Denote 
$y_l = e^{x_l}$ for $l=1,2$. We compute 
$(\blacksquare^+\cdot \pmb\sigma_j)(f)$ in terms of $x_1, x_2$ fowllowing Prop.
\ref{prop:compare-blacksq-blacktri}:
\begin{align*}
    (\blacksquare^+\cdot \pmb\sigma_j)(f)(y_1,y_2) = 
(\blacktriangle^+ \cdot \pmb\tau_j)( f_{\exp} ) (x_1, x_2)
(e^{x_1} \exp[0,x_2])^{-1}
\end{align*}
Now apply $\pmb\sigma_j$ on both sides: 
\begin{align*}
    (\pmb\sigma_j\cdot \blacksquare^+\cdot \pmb\sigma_j)(f)(y_1,y_2)  
=    
(\pmb\tau_j \cdot \blacktriangle^+ \cdot \pmb\tau_j)( f_{\exp} ) (x_1, x_2)
\pmb\tau_0 \brac{
(e^{x_1} \exp[0,x_2])^{-1} 
},
\end{align*}
here we have used the fact that $\pmb\tau_j(f f') = \pmb\tau_0(f)
\pmb\tau_j(f') = \pmb\tau_0(f') \pmb\tau_j(f)$.
To continue, we now make use of \cref{eq:btri-to-tau}:
\begin{align*}
(\pmb\tau_j \cdot \blacktriangle^+ \cdot \pmb\tau_j)( f_{\exp} ) (x_1, x_2)
&= (\blacktriangle_{0,j}^+ - \blacktriangle^-)(f_{\exp}) (x_1, x_2)
=  \exp[0,x_1] 
(\blacksquare_{0,j}^+ -\blacksquare^-)(f)(y_1, y_2)
\\
\pmb\tau_0 \brac{
(e^{x_1} \exp[0,x_2])^{-1} 
} &=\brac{
    e^{x_1+x_2}  \exp[0,x_1]  
  }^{-1} 
 \end{align*}
 Finally, we have reached \eqref{eq:bsq-to-tau} by multiplying the two terms
 together:
 \begin{align*}
(\pmb\tau_j \cdot \blacktriangle^+ \cdot \pmb\tau_j)( f_{\exp} ) (x_1, x_2)
 =   e^{x_1+x_2}  
(\blacksquare_{0,j}^+ -\blacksquare^-)(f)(y_1, y_2) = 
(y_1 y_2) (\blacksquare_{0,j}^+ -\blacksquare^-)(f)(y_1, y_2)
.
 \end{align*}

\begin{cor}
    \label{cor:d(f(x))-simplified}
    Let $\mathbf x = [\cdot , h]$ and $\mathbf y = e^{-h}(\cdot) e^h$ be the
    modular derivation and modular operator of a self-adjoint $h\in \mathcal A$, 
    and $f \in C(\R)$, $\tilde f \in C_{\mathscr S}(\R_+)$ and $j\in\R$.  For a derivation 
    $\delta: \mathcal  A \rightarrow \mathcal A$, we have
    \begin{align}
         \label{eq:del-e^jhf(x)}
        \delta( e^{jh} f(\mathbf x)) &= 
        (e^{jh})^{(0)}
        (\pmb\tau_j \cdot \blacktriangle^+ \cdot \pmb\tau_j ) (f)
         (\mathbf x_1, \mathbf x_2)(\delta h)^{(0)}  
         + (e^{jh})^{(0)} \blacktriangle^+(f)        
         (\mathbf x_1, \mathbf x_2)(\delta h)^{(1)} , \\
         \delta( k^j \tilde f(\mathbf y)) &=   
    (k^{j-1})^{(0)} 
     (\pmb\sigma_{j-1} \cdot \blacksquare^+ \cdot \pmb\sigma_j) (\tilde f)
     (\mathbf y_1, \mathbf y_2) (\delta k)^{(0)} +
     (k^{j-1})^{(0)}    \blacksquare^+(\tilde f) 
     (\mathbf y_1, \mathbf y_2) (\delta k)^{(1)} .
         \label{eq:del-k^jf(y)}
    \end{align}
\end{cor}

As an example, we apply the result onto Eq. \eqref{eq:del-e^jh} and
Eq. \eqref{eq:delk^j} to compute the second derivative of $e^{jh} = k^j$,
$j\in\R$, in terms of $h$ and $k$ respectively,
by taking advantage of the following identities:
\begin{align}
    \pmb\tau_j
         \left( G_{\exp}^{(1)}(x;j) \right) &= 
G_{\exp}^{(1)}(x;j)
, \,\,\,
(\pmb\tau_j \cdot \blacktriangle^+) \left( 
 G_{\exp}^{(1)}(x;j) \right)
= \blacktriangle^+ \left( 
 G_{\exp}^{(1)}(x;j) \right)
    \label{eq:inv-of-Gexp-taoj}
\end{align}
and
\begin{align}
    \pmb\sigma_{j-1} \left( 
        G_{\mathrm{pow}}^{(1)}(y;j)
    \right) = 
        G_{\mathrm{pow}}^{(1)}(y;j)
        , \,\,\,
        (\pmb\sigma_{j-2} \cdot \blacksquare^+) 
        \left( 
            G_{\mathrm{pow}}^{(1)}(y;j)
        \right) = 
        \blacksquare^+ \left( 
            G_{\mathrm{pow}}^{(1)}(y;j)
        \right)  ,
    \label{eq:inv-of-Gpow-sigmaj}
\end{align}
where $G_{\exp}^{(1)}$ and $G_{\mathrm{pow}}^{(1)}$  are given in Eq.
\eqref{eq:Gexp1+Gpow1}.
To see the identities, we first observe that the functions 
$e^{jz}$ and $z^j$ are multiplicative, thus  
    \begin{align*}
        e^{j x_0} (e^{jz}[x_1, \dots , x_n]_z ) &=
        e^{jz}[x_1 + x_0, \dots , x_n + x_0]_z \\
        y_0^{j-n+1} (z^j[y_1, \dots, y_n]_z) &=
z^j[y_0 y_1, \dots,y_0 y_n]_z
.
    \end{align*}
    In particular,
    \begin{align*}
        \pmb\tau_j( e^{jz}[0,x] ) &= e^{jx} (e^{jz}[0,-x])      
        = e^{jz}[0,x] \\
        \pmb\tau_j( e^{jz}[0,x_1, x_1+x_2] ) &= e^{j(x_1+x_2)} 
        (e^{jz}[0, -x_1 -x_2, -x_2]) = e^{jz}[ x_1+x_2, 0, x_1].
    \end{align*}
    For the function $z^j$,
    \begin{align*}
        \pmb\sigma_{j-1}(z^j[1,y]) &= y^{j-1} z^j[1,y^{-1}]
        =  z^j[y,1], \\
    \pmb\sigma_{j-2}(z^j[1,y_1, y_1 y_2]) &= (y_1y_2)^{j-2} 
        z^j[1,(y_1 y_2)^{-1}, y_2^{-1}] = 
        z^j[y_1 y_2, 1, y_1].
    \end{align*}
    Since divided differences are symmetric in their  arguments, we have verified 
    \eqref{eq:inv-of-Gexp-taoj} and \eqref{eq:inv-of-Gpow-sigmaj}.
\begin{lem}
    \label{lem:2nd-del-e^h-k^j}
    Let $\delta_1$, $\delta_2$ be derivations on $\mathcal A$ and $j\in\R$.
    We have
    \begin{align}
        \label{eq:del^2-e^jh}
        \delta_1 (\delta_2(e^{jh})) 
        &= e^{jh} G_{\exp}^{(2)}(\mathbf x;j)
        ( \delta_1 (\delta_2 h)) \\ 
        & +
        G_{\exp}^{(1,1)} (\mathbf x_1 , \mathbf x_2;j)  
        \left( 
            \delta_1(h) \otimes  \delta_2(h) 
            + \delta_2(h) \otimes  \delta_1(h)
        \right) .  \nonumber
    \end{align}
    In terms of $k =e^h$,
\begin{align}
    \label{eq:del^2-k^j}
    \delta_1(\delta_2(k^j)) 
    &= k^{j-1}
    G_{\mathrm{pow}}^{(2)} (\mathbf y;j)(\delta_1 \delta_2( k) ) \\
   & + k^{j-2} G_{\mathrm{pow}}^{(1,1)}(\mathbf y_1, \mathbf y_2;j)
         ( \delta_1(k) \otimes  \delta_2(k)  
         +  \delta_2(k) \otimes  \delta_1(k) ) , \nonumber
\end{align}
where the spectral functions are given by:
\begin{align}
    G_{\exp}^{(2)} (x;j) = 
    G_{\exp}^{(1)} (x;j) = e^{jz}[0,x]_z, \,\,\,
    G_{\mathrm{pow}}^{(2)} (y;j) = G_{\mathrm{pow}}^{(1)} (y;j) = 
    z^j[1,y]_z,
    \label{eq:Gexp2+Gpow2}
\end{align}
and
\begin{align}
    G_{\exp}^{(1,1)} (x_1 , x_2;j) &= \blacktriangle^+ 
        \left( 
    G_{\exp}^{(2)} (x;j) 
\right) = ( e^{jz}[0,x]_z )[x_1 , x_1 + x_2]_x
= e^{jz}[0,x_1, x_1+x_2]_z,
    \label{eq:Gexp11}
    \\
    G_{\mathrm{pow}}^{(1,1)} (y_1 , y_2;j) &= \blacksquare^+
    \left( 
    G_{\mathrm{pow}}^{(1)} (y;j)  
    \right) = 
    \left( z^j[1,y]_z \right)[y_1 ,y_1 y_2]_y =
    z^j[1, y_1 , y_1 y_2].
\end{align}
\end{lem}
\begin{rem}
    One usually computes higher derivatives of $e^{jh}$ by the exponential
    expansion (cf. \cite[\S6.1]{MR3194491}). Our argument is very similar to
    \cite[Example 3.9]{leschdivideddifference},  
    but the new notations seem to reveal some
    hierarchy behind the noncommutative analogue of Taylor coefficients, such as 
    $G_{\exp}^\mu$ and $G_{\mathrm{pow}}^\mu$, $\mu \in \set{(2),(1,1)}$,
    of the functions $e^{jz}$ and $z^j$. We leave the exploration to future
    papers.
\end{rem}
\begin{proof}
    According to \cref{eq:del-k^jf(y)}:
  \begin{align*}
      &\,\,\delta_1( \delta_2(k^j) ) = \delta_1\left( 
          k^{j-1} G_{\mathrm{pow}}^{(1)} (\mathbf y;j) (\delta_2(k))
      \right) \\
      =&\,\,     
           k^{j-1} 
          \left( 
              \pmb\sigma_{j-2} \cdot \blacksquare^+ \cdot \pmb\sigma_{j-1}
          \right)
 \left( 
           G_{\mathrm{pow},j}^{(1)} 
          \right) 
          (\mathbf y_1 , \mathbf y_2)
          ( \delta_1(k) \delta_2(k) )  + 
          k^{j-1} \blacksquare^+ 
          \left( 
           G_{\mathrm{pow},j}^{(1)} 
          \right)
          (\mathbf y_1 , \mathbf y_2)
          ( \delta_2(k) \delta_1(k) )   \\
      +&\,\,     
          k^{j-1} G_{\mathrm{pow}}^{(1)} (\mathbf y;j) 
          ( (\delta_1\delta_2) (k)).
  \end{align*}
  To reach \eqref{eq:del^2-k^j}, it suffices to show
  $ \blacksquare^+(G_{\mathrm{pow},j}^{(1)}) =
  (\pmb\sigma_{j-2} \cdot \blacksquare^+ \cdot \pmb\sigma_{j-1}) 
  (G_{\mathrm{pow},j}^{(1)} )
  $, which follows immediately from  Eq. \eqref{eq:inv-of-Gpow-sigmaj}.
  Similarly, \eqref{eq:del^2-e^jh} follows from \eqref{eq:inv-of-Gexp-taoj}.

\end{proof}

\subsection{Variation on Local Expressions} 
\label{subsec:var-wrt-log-Wfactor}
For a fixed Weyl factor $k = e^h$,  $h = h^* \in  \mathcal A$. 
A differential calculus is usually generated by a family of derivations. 
We consider only one derivation $\delta$ (on $ \mathcal A$)
to simplify the notation.
Similar to the commutative setting, differential expressions like 
$L = L(h, \delta(h),\delta^2(h), \ldots)$ are generated by $h$ and its derivatives.
The new ingredients are  rearrangement operators of the form:
with $f \in  C_{\mathscr S }(\R^n)$,  $ \tilde f \in  C_{\mathscr S }(\R^n_+)$,
$\mathbf x = [\cdot, h]$ and $\mathbf y = k^{-1} (\cdot )k$:
\begin{align*}
    f(\mathbf x_1, \ldots, \mathbf x_n), \, \,  
   \tilde f(\mathbf y_1, \ldots, \mathbf y_n): 
    \mathcal A^{\otimes  n} \to \mathcal A, 
\end{align*}
appearing  as ``coefficients ''of differential operators. 
We shall refer the corresponding integrations $\varphi_0(L)$ as local
expressions, where $\varphi_0: \mathcal A \to \mathbb{C}$ is a tracial
functional.

Let us start with a change of coordinate $k \mapsto \log k$ formula for
differential expressions consisting of (upto) second derivatives of $k$.
 \begin{lem}
    \label{lem:ch-var-Kand-H}
     Let  $\delta$ be a derivation on $ \mathcal A$ and $j \in  \R$.
     Consider the following element $R \in \mathcal A$ 
     \begin{align*}
         R = k^{j-1} K(\mathbf y)(\delta^2 (k) )+
         k^{j-2} H(\mathbf y_1,\mathbf y_2)
         (\delta(k) \otimes  \delta(k) ).   
     \end{align*}
     The change of variable $k \mapsto  h = \log k$ is given by
 \begin{align*}
     R = e^{jh}\left( 
         \tilde K(\mathbf x) (\delta^2 h) + 
         \tilde H(\mathbf x_1, \mathbf x_2)
         (\delta(h) \otimes  \delta(h) )   
     \right),
 \end{align*}
where the spectral functions are transformed as below:
\begin{align}
    \label{eq:ch-var-Kand-H}
    \begin{split}
    \tilde K(x) &= 
    K(e^x) G_{\exp}^{(1)}(x) =
    K(e^x) \exp[0,x], \\
    \tilde H(x_1, x_2) &= 
    2 K(e^{x_1 + x_2}) G_{\exp}^{(1,1)}(x_1,x_2) 
    + H(e^{x_1}, e^{x_2}) 
    e^{x_1} G_{\exp}^{(1)}(x_1) G_{\exp}^{(1)}(x_2)
\\
  &=  2 K(e^{x_1 + x_2}) \exp[0,x_1 , x_1+x_2]
    + H(e^{x_1}, e^{x_2}) \exp[0,x_1] \exp[x_1, x_1+x_2].    
    \end{split}
\end{align}
 \end{lem}
 \begin{proof}
     We apply Eq. \eqref{eq:del^2-e^jh} (with $j=1$) 
     \begin{align*}
         K(\mathbf y)( \delta^2 k) &= k K(\mathbf y)\left( 
             \exp[0,\mathbf x](\delta^2 h)   +
             \exp[0,\mathbf x_1, \mathbf x_1 + \mathbf x_2]      
             ( \delta(h) \otimes  \delta(h) )    
         \right)     \\
 &=      
 k K(e^{\mathbf x}) \exp[0,\mathbf x] (\delta^2 h) 
+ k K(e^{\mathbf x_1 + \mathbf x_2}) 
  \exp[0,\mathbf x_1, \mathbf x_1 + \mathbf x_2]      
  ( \delta(h) \otimes  \delta(h) ).    
     \end{align*}
 To get $ \tilde H$,    we start with   Eq. \eqref{eq:del-e^jh}: 
     \begin{align*}
         H(\mathbf y_1,\mathbf y_2)
         (\delta(k) \delta(k) ) &=    
         H(\mathbf y_1,\mathbf y_2)
         \left( 
             ( k \exp[0,\mathbf x]( \delta h)  ) \otimes 
             ( k \exp[0,\mathbf x]( \delta h)  )
         \right) \\
  &=       
  k^2   H(e^{\mathbf x_1}, e^{\mathbf x_2})
      e^{\mathbf x_1} \exp[0,\mathbf x_1] \exp[0,\mathbf x_2]
      ( \delta(h) \otimes  \delta(h) ).    
\end{align*}
Notice that  the rearrangement operator $e^{\mathbf x_1}$ in the second line
moves the second $k$ (in the first line) to the very left. 
Since the exponential function is multiplicative, 
we have $e^{x_1} \exp[0,x_2] = \exp[x_1, x_1+x_2]$.
The proof is complete.
 \end{proof}

For a real parameter $j \in \R$, consider functional  of the form:
\begin{align}
    F(h) = \varphi_0(e^{jh} f(\mathbf x)(\delta h) \cdot \delta h), \, \, \, \, 
     f \in  C_{\mathscr S}(\R) 
    \label{eq:F-defn}
\end{align}
For any self-adjoint $a \in \mathcal A$,
denote by $ \pmb\delta_a $ the variation along $a$:
\begin{align}
    \label{eq:var-setup-a}
    h\rightarrow h+ \varepsilon a, \,\, \text{and} \,\,
    \pmb \delta_a\defeq \frac{d}{d\varepsilon}\Big|_{\varepsilon = 0}.
\end{align}
\begin{defn}
    \label{defn:grad_hF-abstract}
    Let $F(h)$ be the functional on self-adjoint elements given in 
     \cref{eq:F-defn}.
The functional gradient $\grad_h F\in \mathcal A$ at $h$ with
respect to the inner produce given by $\varphi_0$ is the unique element 
determined by the equation:
\begin{align}
    \pmb\delta_a F(h)  =\varphi_0(\pmb\delta_a(h) \grad_h F), \;\;
    \forall a=a^* \in \mathcal A.
    \label{eq:grad_hF-abstract}
\end{align}
\end{defn}

\begin{thm}
    \label{thm:gradF-h}
    Keep the notations as above. For the functional $F(h)$ given in eq.
    \eqref{eq:F-defn}, the gradient defined in \eqref{eq:grad_hF-abstract} 
    has the following explicit formula:
    \begin{align*}
        \grad_h F = e^{jh} \brac{
            K_f(\mathbf x)(\delta^2 h) + 
            H_f(\mathbf x_1 , \mathbf x_2) (\delta(h) \otimes  \delta(h))
        },     
    \end{align*}
    where the one-variable spectral function is the average of $f$ with respect
    to the cyclic operator (of order two) $\pmb\tau_j$:
  \begin{align}
      \label{eq:K_f-defn}
      K_f & = - (1+\pmb\tau_j)(f)  .
  \end{align}
The two-variable function $H_f$ is determined by $K_f$ in terms of the
following Connes-Moscovici type functional relation:  
\begin{align}
  H_f & =    ( (1+\pmb\tau_j - \pmb\tau_j^2)\cdot \blacktriangle^+) (K_f) 
      \label{eq:H_f-CM-relation}
\end{align}
\end{thm}
\begin{rem}
    After substituting the definitions of $\pmb\tau_j$ and $\blacktriangle^+$ in the
    previous section, we have explicit relations:
    \begin{align*}
        K_f(x) = f(x) + e^{jx} f(-x) 
    \end{align*}
    and for $H_f$: 
\begin{align*}
    \blacktriangle^+(K_f) (x_1 , x_2) &= K_f[x_1, x_1 + x_2], \,\,\,
(\pmb\tau_j^2 \cdot \blacktriangle^+(K_f) ) (x_1 , x_2) = e^{jx_1}
   K_f[x_2, -x_1 ], \\
(\pmb\tau_j \cdot\blacktriangle^+)(K_f) 
(x_1 , x_2) &=  e^{-j(x_1 + x_2)} K_f[-x_1- x_2, -x_2] .
\end{align*}
\end{rem}
\begin{proof}
 According to the Leibniz property,  
 the variation splits into three terms:
 \begin{align*}
       \pmb\delta_a F(h) 
     &=  \varphi_0\brac{
         \pmb\delta_a( e^{jh} f(\mathbf x)) (\delta h) \cdot (\delta h)   
     }\\ 
     & +  \varphi_0\brac{
     e^{jh} f(\mathbf x) (\delta (\pmb\delta_a(h)) \cdot (\delta h)
     }
     + \varphi_0\brac{
     e^{jh} f(\mathbf x) (\delta h) \cdot (\delta (\pmb\delta_a(h)) 
     } ,
 \end{align*}
 where we have used the fact that $\pmb\tau_a$ and $\delta$ commute. 
 We postpone the computation of the first term to  
Lemma \ref{lem:Grad-explicit}
which gives rise to a contribution 
$\mathrm{ Grad } (f,\delta h, \delta h)$ (cf. \cref{eq:Grad-explicit}) 
\begin{align*}
  \mathrm{ Grad } (f,\delta h, \delta h) =   
  - e^{jh} (\pmb\tau_j^2 \cdot \blacktriangle^+)(K_f)
  (\mathbf x_1, \mathbf x_2) ( \delta(h) \otimes  \delta(h) ).
\end{align*}
The last two terms are of the same form and can be handled together:
\begin{align*}
     & \;\;  \varphi_0\brac{
     e^{jh} f(\mathbf x) (\delta (\pmb\delta_a(h)) \cdot (\delta h)
     }
     + \varphi_0\brac{
     e^{jh} f(\mathbf x) (\delta h) \cdot (\delta (\pmb\delta_a(h)) 
     } \\
   =&\;\;  \varphi_0\brac{
         e^{jh} (1+ \pmb\tau_j)(f)(\mathbf x)(\delta h) \cdot  
         (\delta ( \pmb\delta_a h) )   
     } 
   = \varphi_0\brac{
       ( \pmb\delta_a h) \delta \left( 
         e^{jh}  K_f(\mathbf x)(\delta h) 
       \right)
   } ,   
\end{align*}
 Note that, in both equations above, $K_f$ appears through
 $K_f = -(1+ \pmb\tau_j)(f)$ (cf. \cref{eq:K_f-defn}).
So far, we have reached:
\begin{align*}
    \grad_h F = \mathrm{Grad}(f, \delta h, \delta h) +
\delta \left( 
         e^{jh}  K_f(\mathbf x)(\delta h) 
     \right). 
\end{align*}
It remains to compute:
\begin{align*}
\delta \left( 
         e^{jh}  K_f(\mathbf x)(\delta h) 
     \right) &=  
         e^{jh}  K_f(\mathbf x)(\delta^2 h) +
         \delta (e^{jh}  K_f(\mathbf x) ) (\delta h) \\
 &=         
         e^{jh}  K_f(\mathbf x)(\delta^2 h) +
         e^{jh} 
         (\pmb\tau_j \cdot \blacktriangle^+ \cdot \pmb\tau_j
         + \blacktriangle^+) (K_f) 
         (\mathbf x_1, \mathbf x_2)      
     ( \delta (h) \otimes \delta (h) ) \\
 &=         
         e^{jh}  K_f(\mathbf x)(\delta^2 h) +
         e^{jh} 
         ( (\pmb\tau_j +1 ) \cdot \blacktriangle^+
         ) (K_f) 
         (\mathbf x_1, \mathbf x_2)      
         ( \delta (h) \otimes  \delta (h) ). 
\end{align*}
To see the second $=$ sign, we expand $\delta (e^{jh}  K_f(\mathbf x) )$
via Corollary \ref{cor:d(f(x))-simplified} with $\rho = \delta h$. To reach the
third $=$ sign, we need the fact that 
$K_f$ is $\pmb\tau_j$-invariant because $\pmb\tau_j$ is of order two when
acting on one-variable functions. 

\end{proof}

\begin{lem}
\label{lem:Grad-explicit}
    Keep notations. For any $\rho_1, \rho_2 \in \mathcal A$ and self-adjoint $a
    \in  \mathcal A$:
    \begin{align}
        \label{eq:Grad-defn}
        \varphi_0(\pmb\delta_a(e^{jh} f(\mathbf x) ) (\rho_1) \cdot \rho_2) =
            \varphi_0(\pmb\delta_a(h) \mathrm{Grad}_{\mathrm I}
            (h,\rho_1,\rho_2) ) 
    \end{align}
    where
    \begin{align}
\label{eq:Grad-explicit}
        \mathrm{Grad}(h,\rho_1,\rho_2) = e^{jh}
        (\pmb\tau_j^2 \cdot \blacktriangle^+ \cdot (1+\pmb\tau_j) ) (f)
        (\mathbf x_1, \mathbf x_2)  (\rho_1 \cdot \rho_2)
        .
    \end{align}
\end{lem}
\begin{proof}
  We first expand $\pmb\delta_a(e^{jh} f(\mathbf x)) (\rho_1)$ 
    using Corollary   \ref{cor:d(f(x))-simplified} and then 
    apply the $\pmb\tau_j$ operation (see, Lemma \eqref{lem:tao-defn})
    to move $\pmb\delta_a(h)$ to the very left
    as indicated in \eqref{eq:Grad-defn}:
      \begin{align*}
     &\;\;   \varphi_0(\pmb\delta_a(e^{jh} f(\mathbf x)(\rho_1) \cdot \rho_2) \\
   =&\;\;     
   \varphi_0\brac{
                e^{jh} 
        (\pmb\tau_j \cdot \blacktriangle^+ \cdot \pmb\tau_j) (f)          
            (\mathbf x_1, \mathbf x_2) 
            (\pmb\delta_a(h) \otimes  \rho_1)   \cdot \rho_2
            }    
      + \varphi_0\brac{
                e^{jh} 
                (\blacktriangle^+) (f)          
            (\mathbf x_1, \mathbf x_2) 
            (\rho_1 \otimes  \pmb\delta_a(h)  )   \cdot \rho_2
            }     \\
=&\;\;
   \varphi_0\brac{
                e^{jh} 
        (\pmb\tau_j^2 \cdot \blacktriangle^+ \cdot \pmb\tau_j) (f)          
            (\mathbf x_1, \mathbf x_2) 
            ( \rho_1 \otimes  \rho_2)   \cdot \pmb\delta_a(h)
            }    
      + \varphi_0\brac{
                e^{jh} 
                (\pmb\tau_j^2 \cdot \blacktriangle^+) (f)          
            (\mathbf x_1, \mathbf x_2) 
            (\rho_1 \otimes  \rho_2   )   \cdot\pmb\delta_a(h) 
            }     
\\            
 =&\;\;
 \varphi_0\brac{ \pmb\delta_a(h)
     \left\{ 
     e^{jh}
     (\pmb\tau_j^2 \cdot \blacktriangle^+ \cdot (1 + \pmb\tau_j)) (f)
 (\mathbf x_1, \mathbf x_2) (\rho_1 \otimes  \rho_2) 
     \right\}
 }.
      \end{align*}
      By definition, $\mathrm{Grad}(h,\rho_1,\rho_2) $ is given by the
      expression enclsoed by the curly brackets.
\end{proof}

We  perform the parallel  computation  with respect
to the Weyl factor $ k=e^h$ itself.
Same functional $F(h)$ in Eq.  \eqref{eq:F-defn} can be written as:
\begin{align}
    \label{eq:F(k)-defn}
     F(k) = \varphi_0(k^j \tilde f(\mathbf y)(\delta k) \cdot \delta k), \, \, 
     \tilde f \in  C_{\mathscr S}(\R_+) 
\end{align}
The variation is identical: for any self-adjoint $a \in \mathcal A$, 
we perturb the its logarithm along $a$:
\begin{align*}
    k \mapsto \exp(\log k + \varepsilon a), \,\,\,
    \pmb \delta_a \defeq \frac{d}{d\varepsilon} \Big|_{\varepsilon=0}
\end{align*}
The functional gradient $\grad_k F$ is defined in a slightly different  way:
\begin{align}
    \label{eq:grad_kF-defn}
    \pmb \delta_a F(k) = 
    \varphi_0\left( 
     \pmb\delta_a(k) \grad_k F
    \right),
    \;\; \forall a=a^* \in \mathcal A. 
\end{align}
Observe that $ \pmb\delta_a (h) = a$ and
\begin{align*}
     \varphi_0\left( a \grad_h F \right) 
     &=
     \varphi_0\left(  \pmb\delta_a (k) \grad_k F \right) =
     \varphi_0\left( k e^z[0,\mathbf x]( \pmb\delta_a (h)) \grad_k F \right)\\
     &=  
     \varphi_0\left(  e^z[0,\mathbf x](a) \grad_k F k\right) =
     \varphi_0\left( a e^z[0,- \mathbf x] ((\grad_k F) k)\right)\\
     &=  
     \varphi_0\left( a k e^{\mathbf x} e^z[0,- \mathbf x] (\grad_k F )\right)
     \varphi_0\left( a k  e^z[ \mathbf x , 0] (\grad_k F )\right)
\end{align*}
that is:
\begin{align}
    \label{eq:grad_hF-to-grad_kF}
 \grad_h F = k e^z[0,\mathbf x] (\grad_k F).    
\end{align}
\begin{thm}
    \label{thm:CM-intermsof-k}
    Consider the functional $F$ given in Eq.  \eqref{eq:F(k)-defn}. 
    The functional gradient $\grad_k F$
    defined in Eq.  \eqref{eq:grad_kF-defn} is of the form:
    \begin{align*}
        \grad_k F = k^{j} \tilde K_{\tilde f} (\mathbf y)(\delta^2 k) 
        +
        k^{j-1} \tilde H_{\tilde f} (\mathbf y_1, \mathbf y_2)
        (\delta(k) \otimes  \delta(k) ),
    \end{align*}
    where
\begin{align}
    \label{eq:tildeK=avrf}
    \tilde K_{\tilde f}(y) &= -(1+ \pmb\sigma_{j})(\tilde f)(y),
    \\
    \tilde H_{\tilde f}(y_1, y_2)
                           &=
    (1+ \pmb\sigma_{j-1} - \pmb\sigma_{j-1}^2)\cdot \blacksquare^+
        (\tilde K_{\tilde f} ) (y_1, y_2)
    \label{eq:tildeH=CMtileK}
\end{align}
\end{thm}
\begin{rem}
    We have the explicit form for $ \tilde K_{ \tilde f } $:
    \begin{align*}
        - \tilde K_{\tilde f} (y) = \tilde f(y) + y^j  \tilde f(y^{-1}).
    \end{align*}
    and for $ \tilde H_{ \tilde f } $:
    \begin{align*}
        \blacksquare^+ (\tilde K_{\tilde f})(y_1, y_2)
        &= \tilde K_{\tilde f}[y_1, y_1 y_2],
        \\
          (\pmb\sigma_{j-1}\cdot \blacksquare^+) (\tilde K_{\tilde f})(y_1, y_2)
          &=  (y_1y_2)^{j-1} 
    \tilde K_{\tilde f}[(y_1 y_2)^{-1},  y_2^{-1}] \\ 
          (\pmb\sigma_{j-1}^2 \cdot \blacksquare^+) (\tilde K_{\tilde f})(y_1, y_2)
          &=  (y_1)^{j-1} 
          \tilde K_{\tilde f}[ y_2 , y_1^{-1}].
    \end{align*}
\end{rem}
\begin{proof}
    Let us start with the Leibniz property:
    \begin{align*}
        \pmb\delta_a F(k) &=
        \varphi_0\left( 
            \pmb\delta_a( k^j \tilde f(\mathbf y) ) (\delta(k)) \cdot \delta(k)     
        \right)   \\
   &+     \varphi_0\left( 
           k^j \tilde f(\mathbf y)    (\delta(\pmb\delta_a ( k ))) \cdot \delta(k)     
        \right)   
   +     \varphi_0\left( 
       k^j \tilde f(\mathbf y) \delta(k) \cdot (\delta(\pmb\delta_a ( k )))    
   \right)  . 
    \end{align*}
To combine the two terms in the second line, 
we first use the cyclic operator $\pmb\sigma_j$
to move $\pmb\delta_a(k)$ to the very left and then add up the spectral functions: 
    \begin{align*}
   &\;\;    \varphi_0\left( 
           k^j \tilde f(\mathbf y)    (\delta(\pmb\delta_a ( k ))) \cdot \delta(k)     
        \right)   
   +     \varphi_0\left( 
       k^j \tilde f(\mathbf y) \delta(k) \cdot (\delta(\pmb\delta_a ( k )))    
        \right)   \\
  = &\;\;     \varphi_0\left( 
      k^j (1+\pmb\sigma_j)(\tilde f)
      (\mathbf y) \delta(k) \cdot (\delta(\pmb\delta_a ( k )))    
        \right)   
  =      \varphi_0\left( \pmb\delta_a ( k )
      \delta\brac{   k^j  \tilde K_{\tilde f}
      (\mathbf y) \delta(k)  
  }
        \right)   .
    \end{align*}
    Again, Eq. \eqref{eq:tildeK=avrf} brings in $ \tilde K_f$.
    The first line is computed in Lemma \ref{lem:Gradexp-f-r1-r2} below. 
    As a result,
    \begin{align*}
        \grad_k F = k^{j-1} \mathrm{Grad}_{\exp} (\tilde f, \delta(k), \delta(k))
    +  \delta\brac{   k^j  \tilde K_{\tilde f}
    (\mathbf y) \delta(k)  },
    \end{align*}
    where
    \begin{align*}
        \mathrm{Grad}_{\exp} (\tilde f, \delta(k), \delta(k))
        & =( \pmb\sigma_{j-1}^2 \cdot \blacksquare^+)(1+\pmb\sigma_j)
        ( \tilde f ) (\mathbf y_1, \mathbf y_2) 
        (\delta(k) \delta(k))\\
        & = - ( \pmb\sigma_{j-1}^2 \cdot \blacksquare^+) 
        ( \tilde K_{\tilde f} ) (\mathbf y_1, \mathbf y_2) 
        (\delta(k) \delta(k)).
    \end{align*}
    The second term is given by Eq. \eqref{eq:del-k^jf(y)} in Corollary 
    \ref{cor:d(f(x))-simplified}:    
    \begin{align*}
      \delta\brac{   k^j  \tilde K_{\tilde f}
    (\mathbf y) ( \delta(k) )  }& = k^{j-1}
    (\pmb\sigma_{j-1} \cdot \blacksquare^+ \cdot \pmb\sigma_j
    + \blacksquare^+ )
    (\tilde K_{\tilde f}) 
    (\mathbf y_1, \mathbf y_2) 
    ( \delta(k) \delta(k) )
  +    k^j  \tilde K_{\tilde f}
    (\mathbf y) ( \delta^2(k) )   \\
    &=   k^{j-1}
 ( (1 + \pmb\sigma_{j-1} ) \cdot \blacksquare^+)
    (\tilde K_{\tilde f}) 
    (\mathbf y_1, \mathbf y_2) 
    ( \delta(k) \delta(k) )
  +    k^j  \tilde K_{\tilde f}
  (\mathbf y) ( \delta^2(k) )  . 
    \end{align*}
    We complete the proof of Eq. \eqref{eq:tildeH=CMtileK} by adding up the
    corresponding terms.
\end{proof}

\begin{lem}
    \label{lem:Gradexp-f-r1-r2}
    Keep notations. For any $\rho_1, \rho_2 \in \mathcal A$,
 and    $a =  a^* \in  \mathcal A$, we have
    \begin{align}
        \varphi_0\left( 
            \pmb\delta_a\left(  
                k^j \tilde f(\mathbf y)   
            \right) (\rho_1) \cdot \rho_2
        \right)     
  =
  \varphi_0 \left( 
      \pmb\delta_a(k) k^{j-1} \mathrm{Grad}_{\exp}(\tilde f, \rho_1 , \rho_2)
  \right),
        \label{eq:phi0-dela-vol-fy}
    \end{align}
 with   
 \begin{align*}
       \mathrm{Grad}_{\exp}(\tilde f, \rho_1 , \rho_2) = 
       (\pmb\sigma_{j-1}^2 \cdot \blacksquare^+
    \cdot (1+\pmb\sigma_j)
) (\tilde f)
(\mathbf y_1, \mathbf y_2) (\rho_1 \rho_2) .
 \end{align*}
\end{lem}
\begin{proof}
We first apply Corollary   \ref{cor:d(f(x))-simplified} to expand 
$\pmb\delta_a\left(  k^j \tilde f(\mathbf y)   \right) $ and
then use $\pmb\sigma_j$ to move 
$\pmb\delta_a(k)$ to the vary right in the local expression
(cf. Lemma \ref{lem:taoexp-defn}):
        \begin{align*}
        \varphi_0\left( 
            \pmb\delta_a\left(  
                k^j \tilde f(\mathbf y)   
            \right) (\rho_1) \cdot \rho_2
        \right)     
     &=
     \varphi_0\left( 
         k^{j-1}  
         (\pmb\sigma_{j-1} \cdot \blacksquare^+ \cdot \pmb\sigma_j) (\tilde f)    
         (\mathbf y_1 ,\mathbf y_2) ( \pmb\delta_a(k) \otimes \rho_1) \cdot \rho_2
     \right) \\
   &+
     \varphi_0\left( 
         k^{j-1} 
         (\blacksquare^+) (\tilde f)    
         (\mathbf y_1 ,\mathbf y_2) 
         ( \rho_1 \otimes  \pmb\delta_a(k) ) \cdot \rho_2
     \right) \\
 &=    \varphi_0\left( 
         k^{j-1} 
         (\pmb\sigma_{j-1}^2 \cdot \blacksquare^+ \cdot (1+\pmb\sigma_j ) ) 
         (\tilde f)    (\mathbf y_1 ,\mathbf y_2) ( \rho_2  \otimes  \rho_2)
 \pmb\delta_a(k) 
     \right) 
     \\
 &=    \varphi_0\left( 
     \pmb\delta_a(k) \left\{ 
         k^{j-1} 
         (\pmb\sigma_{j-1}^2 \cdot \blacksquare^+ \cdot (1+\pmb\sigma_j ) ) 
         (\tilde f)    (\mathbf y_1 ,\mathbf y_2) ( \rho_2  \otimes  \rho_2)
     \right\}
 \right) .
        \end{align*}
        By definition, $\mathrm{Grad}_{\exp}(\tilde f, \rho_1 , \rho_2) $ 
is equal to the part between curly brackets.
\end{proof}

\begin{cor}
    The functional $F(k)$ in  \cref{eq:F(k)-defn} is the zero functional if and
    only if the spectral function satisfies $(1+\pmb\sigma_j)(f) = 0$.
\end{cor}
\begin{proof}
    As shown in \cref{eq:tildeK=avrf} and \cref{eq:tildeH=CMtileK}, 
    $ \tilde K_{ \tilde f } =0$ implies $ \tilde H_{ \tilde f } =0$.
    In particular, $\grad_k F = 0 $ if and only if $ \tilde K_{ \tilde f } =0$,
    that is $(1+\pmb\sigma_j)(f) = 0$.
\end{proof}

In fact, the relation $(1+\pmb\sigma_j)(f) = 0$ implies directly that 
$F(k) = F(e^h) =0$ for
all self-adjoint $h$ without using 
\cref{eq:tildeK=avrf} and \cref{eq:tildeH=CMtileK}: 
 with the help of $\pmb\sigma_j$ (cf. Lemma \ref{lem:tao-defn}), we get
\begin{align*}
    F(k) = \varphi_j( f(\mathbf y)(\delta k) \cdot (\delta k) )
    = \varphi_j\left( 
        \pmb\sigma_j(f)(\mathbf y)(\delta k) \cdot (\delta k)
    \right)  .
\end{align*}
The argument was first used in the first proof of Gauss-Bonnet on
$\mathbb{T}_\theta^2$, cf. \cite[\S3.3]{Connes:2011tk}.

%% file: cm_funrel.tex
\section{Modular curvature as a functional gradient}
\label{sec:var-FEH}

\subsection{Notations for $\T^m_\theta$} 
\label{subsec:notation-for-Tm}
Let us quickly review the basic
notations of the differential calculus on noncommutative $m$-tori $\T^m_\theta$
from deformation point of view (along a $m$-torus action).
The smooth structure of the noncommutative manifold $\T^m_\theta$ is
represented by the deformed algebra
$C^\infty(\T^m_\theta) = ( C^\infty(\T^m), \times_\theta)$
which takes the underlying topological vector space from its 
commutative counterpart 
and the multiplication $ \times_\theta$ is deformed along a torus action
and parametrized by  $m$ by $m$ skew-symmetric matrices $\theta$:
for any $f,g\in C^\infty(\T^m_\theta)$,
\begin{align}
    f \times_\theta g \defeq \sum_{r,l \in \Z^m}
    \exp(2\pi i \abrac{\theta r, l}) f_r g_l.
    \label{eq:defn-f-timestheta-g}
\end{align}
where $f = \sum_{r\in \Z^m} f_r$ and $g = \sum_{l \in \Z^m} g_l$ are 
the isotypical decomposition with respect to the $\mathbb{T}^m$-action.
Any translation invariant (i.e. $\mathbb{T}^m$-invaraint)
measure $d\mu$ on $\mathbb{T}^m$ can be deform to a tracial functional:  
\begin{align*}
    \varphi_0: C^\infty(\T^m_\theta) \rightarrow \mathbb C: 
    f \mapsto \int_{\T^m} f d\mu,
\end{align*}
where we further assume that $ \varphi_0$ is normalized $\varphi_0(1) =1$.
The Hilbert space of $L^2$-functions can be
recovered by the standard GNS construction with respect to $\varphi_0$:
\begin{align}
    \label{eq:calH-Hilsp} 
  \mathcal H \defeq L^2( C^\infty(\mathbb{T}^m_\theta) , \varphi_0)   
\end{align}
whose the inner product is given by
\begin{align}
    \label{eq:calH-dotprd} 
    \left<a,b \right> \defeq \left<a,b \right>_{\varphi_0} =
    \varphi_0( b^* a), \, \, \, \forall a,b \in C^\infty(\T^m_\theta) .
\end{align}

The  coordinate system on $\R^m$: $(x_1, \dots, x_m)$ induces one on the
quotient $\mathbb{T}^m = \R^m / 2\pi i \Z^m$. We use the standard Euclidean
metric and denote by $\nabla$ the metric connection. 
The associated covariant differentials
\begin{align}
    \label{eq:nabla-ma-nctori}
    \nabla_l \defeq \nabla_{\partial_{x_l}}, \;\;
    l = 1,\dots,m,
\end{align}
generate the algebra of differential operators. In particular, the flat
Laplacian reads:
\begin{align}
    \Delta = -(\nabla^2_1 + \cdots + \nabla^2_m).
    \label{eq:flat-del-m-nctori}
\end{align}
We will also feel free to use the classical notations  
\cite[\S1.3]{MR3194491} (for $\mathbb{T}^2_\theta$) and
\cite[\S6]{Rieffel:1990wo}
(for general $\mathbb{T}^m_\theta$). 
The volume functional is denoted by $\varphi_0$ in all the references 
and 
the basic derivations $\delta_l$ in \cite[\S6]{Rieffel:1990wo}
are slightly different from the covariant differential above,
that is, $\delta_l = -i \nabla_l$, $l=1, \ldots ,m$.

 \subsection{Modular Gaussian Curvature on $\mathbb{T}^2_\theta$}
\label{subsec:mod-gauss-T2}

Modeled on spectral geometry of Riemannian manifolds
(cf. \cite[Ch. 4]{gilkey1995invariance}), local invariants to be investigated 
are derived from spectral asymptotic of some Laplacian type operators 
attached to  the underlying Riemannian metric. 
On noncommutative two tori $\mathbb{T}^2_\theta$, $ \theta \in  \R \setminus \Q$,
we consider conformal change of the flat metric,  starting with
rescaling the flat volume functional, i.e., the canonical trace $\varphi_0$, by a 
 Weyl factor $k = e^h$ where $h = h^* \in C^\infty(\T^2_\theta)$: 
\begin{align}
    \varphi(a) = \varphi_0(a e^{-h})  = \varphi_0(a k^{-1}).
    \label{eq:phi-weight-defn}
\end{align}
The associated Dolbeault Laplacian is changed 
from the flat one $\Delta = \bar \partial^* \bar \partial$ in 
\cref{eq:flat-del-m-nctori} to 
$\bar \partial^*_\varphi\bar \partial$, where the new adjoint is
taken with respect to the inner product:
$\abrac{a,b}_\varphi \defeq \varphi(b^* a)$, $\forall a,b
\in C^\infty(\mathbb{T}^2_\theta)$.
To keep the same Hilbert space $\mathcal H$ (or the same inner product 
in \cref{eq:calH-dotprd}) of $L^2$-functions
defined in \cref{eq:calH-Hilsp}, we shall work with
\begin{align*}
    \Delta_\varphi = k^{1 / 2} \Delta k^{1 / 2} 
     : \mathcal H \to \mathcal H
\end{align*}
which is anti-unitary equivalent to $\bar \partial^*_\varphi\bar \partial$,
see \cite[\S1.5]{MR3194491}.

The analogue of the Gaussian curvature is based on the following 
spectral realization on Riemann surfaces, known as the 
conformal anormaly, due to  Polyakov 
\cite{Polyakov1981Quantum-fermionic,Polyakov1981Quantum-bosonic},
of  the functional assigning each metric $g$ to 
the Ray-Singer determinant of its scalar Laplacian $\Delta_g$: 
 $g \to  -\logdet' \Delta_g \defeq \zeta'_{\Delta_g }(0)$.
In \cite{osgood1988extremals}), a slightly modified version was introduced 
whose gradient flow recovers Hamilton's Ricci flow.

\begin{defn}[Modular Gaussian Curvature]
In \cite{MR3194491}, the precise analogue of the OPS-functional 
(Osgood-Phillips-Sarnak)
on $\mathbb{T}^2_\theta$ is defined to be 
\begin{align}
    \label{eq:FOPS-defn}
    F_{\mathrm{OPS} } (h)  = \zeta'_{\Delta_\varphi} (0) + 
    \log \varphi_0(e^{-h})
\end{align}
as a function on self-adjoint elements in $C^\infty(\mathbb{T}^2_\theta)$.
The functional gradient in the sense of \cref{eq:grad_hF-abstract}
yields the notion of modular Gaussian curvature:
$\mathsf K_{\varphi} \defeq \grad_h F_{\mathrm{OPS}}$. 
\end{defn}

The analytic continuation of the spectral zeta function  is often proved 
by means of the heat trace asymptotic. 
On general noncommutative tori $\mathbb{T}^m_\theta$, 
we can establish the small time asymptotic by, for instance, Connes
pseudo-differential calculus \footnote{
We also mention two other approaches: 
\cite{Iochum:2017ver} using Duhamel type perturbation series of 
the exponential function and \cite{sukochev2019local} using 
harmonic analysis on $\mathbb{T}^2_\theta$.
}: 
\begin{align}
    \label{eq:heat-asym-delphi}
    \Tr(f e^{-t\Delta_\varphi} ) \backsim \sum_{j=0}^\infty
    V_j(f,\Delta_\varphi) t^{(j-m)/2}, \;\;
    \forall f \in C^\infty(\T^m_\theta).
\end{align}
Each $V_j(\cdot,\Delta_\varphi)$ is a linear functional in $f$ and is
absolutely continuous with respect to the volume functional $\varphi_0$
with a smooth functional density in $ C^\infty(\mathbb{T}^m_\theta)$.
For instance,  the density
$R_{\Delta_\varphi} \in C^\infty(\T^m_\theta)$ of the $V_2$-term 
is  determined by the property:
\begin{align}
    \label{eq:defn-Rdelpgi}
    V_2(f, \Delta_\varphi) = \varphi_0(f R_{\Delta_\varphi}),
    \,\,\, \forall f \in C^\infty(\T^m_\theta).
\end{align}

Back to dimension $m=2$,
let $P_\varphi$ be the orthogonal projection onto $\ker \Delta_\varphi$.
For any
$f\in C^\infty(\T^2_\theta)$, consider the zeta function
$\zeta_{\Delta_\varphi}(z) \defeq \zeta_{\Delta_\varphi}(1;z)$, where
\begin{align}
    \zeta_{\Delta_\varphi}(f;z) = \Tr(f \Delta_\varphi^{-z} (1-P_\varphi)), \,\,\,
    \Re z > 2,
    \label{eq:subsec-logdet-defn-zeta-with-a}
\end{align}
in which we define the inverse restricted on the kernel 
$\Delta_k^{-1} |_{\ker \Delta_k}  =1$ .
It is related to the heat trace by Mellin transform:
\begin{align}
    \label{eq:zetaviaheat-varprob}
    \zeta_{\Delta_\varphi}(f;z) = \frac{1}{\Gamma(z)}
    \int_{0}^{\infty} t^{z-1} \brac{
        \Tr(fe^{-t\Delta_\varphi}(1-P_\varphi))  
    } dt
\end{align}
and asymptotic in \eqref{eq:heat-asym-delphi}  implies that 
the right hand side above is a meromorphic function in $z$.
Moreover, $z=0$ is a regular point (cf. \cite{MR3194491} Eq.(3.11)) and:
\begin{align}
    \label{eq:zetazero-varprob}
    \zeta_{\Delta_\varphi}(f;0) = V_2(f,\Delta_\varphi )- \Tr(fP_{\varphi})
    = V_2(f,\Delta_\varphi ) -\frac{\varphi_0(f k^{-1} )}{\varphi_0(k^{-1})} 
    .
\end{align}
In particular,
\begin{align*}
    \zeta_{\Delta_\varphi}(1;0) = V_2(1,\Delta_\varphi) -1.
\end{align*}

One has the following variational formulas for $\zeta_{\Delta_{\varphi}}(0)$
and $\zeta'_{\Delta_{\varphi}}(0)$.
\begin{lem}
     \label{prop:logdet-local-varprob}
     Consider the two types of variation: dilation $h_s = s h$
     by $s \in  \R$ so that $k_s \defeq k^s$ 
and $\Delta_{\varphi_s} = k^{s/2} \Delta k^{s/2}$, and a small perturbation 
along a self-adjoint $a = a^* \in  C^\infty(\mathbb{T}^2_\theta)$:
$h_\varepsilon = h + \varepsilon a$ and 
$\Delta_{\varphi_\varepsilon} = e^{(h+ \varepsilon a)/2} \Delta
e^{(h+ \varepsilon a)/2}$. 
We have
    \begin{align}
        \label{eq:ds-zetafun}
        \frac{d}{ds} \zeta_{\Delta_{\varphi_s}}(z)
        = -z \zeta_{\Delta_{\varphi_s}}(\log k;z),
        \, \, \, \, 
        \frac{d}{d \varepsilon} \Big |_{ \varepsilon =0} 
        \zeta_{\Delta_{\varphi_\varepsilon}}(z)
        = -\frac{z}{2} \zeta_{\Delta_{\varphi}}
        \brac{ e^{-\mathbf x /2} e^z[\mathbf x,0]_z (a) ;z} .
    \end{align}
    Applying $d/dz|_{z=0}$ on both sides above yields:  
 \begin{align}
        \label{eq:ds-zetafun-dz-at-zero}
    -\frac{d}{ds} \zeta'_{\Delta_{\varphi_s}}(0) 
    =  \zeta_{\Delta_{\varphi_s}}(\log k ;0),
        \, \, \, \, 
       - \frac{d}{d \varepsilon} \Big |_{ \varepsilon =0} 
        \zeta'_{\Delta_{\varphi_\varepsilon}}(z)
        = \zeta_{\Delta_{\varphi}}
        \brac{ e^{-\mathbf x /2} e^z[\mathbf x,0]_z (a) ;z} .
\end{align}
\end{lem}
\begin{rem}
    The first type of the variation is a special case of the second type with
    $a = h$.
\end{rem}
\begin{proof}
    See \cite[\S4.1, 4.2]{MR3194491}, in which 
    $e^{-\mathbf x /2} e^z[\mathbf x,0]_z (a)$ is written as
    $\frac{1}{2} \int_{-1}^1 e^{\frac{uh}{2}} a e^{\frac{-uh}{2}} du$.
\end{proof}

After substituting   \cref{eq:zetazero-varprob} to the right hand sides of 
\cref{eq:ds-zetafun-dz-at-zero}, we have arrived at:
\begin{prop}
The OPS-functional in \cref{eq:FOPS-defn} is 
determined by the second heat coefficient in the following way$:$
    \begin{align}
        \label{eq:FOPS-in-terms-of-RDel}
        F_{\mathrm{OPS}}( k ) = \int_0^1 V_2(\log k, \Delta_{\varphi_s} )ds +
        \zeta'_{\Delta}(0),   
    \end{align}
    where $\Delta_{\varphi_s} = k^{s/2} \Delta k^{s/2}$ with $s\in\R$.
\end{prop}
\begin{proof}
    Again, we refer the details to \cite[\S4.1, 4.2]{MR3194491}. 
\end{proof}

\begin{prop}
    \label{prop:gradFOPS=varphi-RDel}
    For any self-adjoint $a \in  C^\infty(\mathbb{T}^2_\theta)$, 
    the variation of the OPS-functional $($see \cref{eq:FOPS-defn}$)$
    is given by$:$
    \begin{align}
     \label{eq:gradFOPS=varphi-RDel}
    \varphi_0( a \grad_h F_{\mathrm{OPS} }) =
    \frac{d}{d \varepsilon} \Big |_{\varepsilon = 0} 
    F_{\mathrm{OPS}} ( e^{h+ \varepsilon a}) =
  - V_2 \brac{ e^{-\mathbf x /2} e^z[\mathbf x,0]_z (a),\Delta_\varphi} .
    \end{align}
   In particular, 
   \begin{align}
    \label{eq:gradFOPS_h=varphi-RDel}
       \grad_h F_{\mathrm{OPS} } = 
    -   e^{\mathbf x /2} e^z[-\mathbf x,0]_z ( R_{\Delta_\varphi} ).
   \end{align}
\end{prop}
We  also work with the operator
\begin{align}
    \label{eq:defn-Del_k}
    \Delta_k \defeq k \Delta = k^{1/2} (k^{1/2} \Delta k^{1/2}) k^{-1/2} 
    =k^{1/2} \Delta_\varphi k^{-1/2} ,
\end{align}
whose complete symbol $\sigma(\Delta_k)(\xi) = k \abs{\xi}^2$ has only one term, the
leading part so that the heat asymptotic is much easier to compute. 
Their heat traces are related in a similar way: $\Tr(f e^{-t\Delta_\varphi} )
= \Tr( (k^{1/2} f k^{-1/2} e^{-t\Delta_k} )$ for all $f\in C^\infty(\T^2_\theta)$.
    Therefore, 
    $V_l(f,\Delta_\varphi) = V_l(\mathbf y^{-1/2}(f) ,\Delta_k)$,
   $l \in \N$. 
 For the second coefficient, we have
 $\varphi_0(f R_{\Delta_\varphi}) =
     \varphi_0( \mathbf y^{-1/2}(f) R_{\Delta_k} )$,
that is
\begin{align}
    R_{\Delta_\varphi} = \mathbf y^{1/2} ( R_{\Delta_k} ).
    \label{eq:Rdelk-vsRdelphi}
\end{align}

\begin{cor}
    \label{prop:gradFOPS=k-RDel}
In terms of $R_{\Delta_k}$ defined above, the two versions of gradient 
$($cf. \cref{eq:grad_hF-abstract} and \cref{eq:grad_hF-to-grad_kF}$)$
can be
rewritten as below:
\begin{align}
    \label{eq:gradFOPS=k-RDel}
    \grad_k F_{\mathrm{OPS}} = -k^{-1} R_{\Delta_k}, \,\,\,
    \grad_h F_{\mathrm{OPS}} = - \exp[0,\mathbf x] ( R_{\Delta_k} ). 
\end{align}
\end{cor}
\begin{rem}
 The OPS-functional gives
    rise to a renormalization of the coefficient $(2-m)/2$ of the EH-functional in 
    \cref{eq:gradFEH-V2-varprob} in later section.  
\end{rem}
\begin{proof}
  We apply  \cref{eq:Rdelk-vsRdelphi} to replace $ R_{\Delta_\varphi}$
  in \cref{eq:gradFOPS_h=varphi-RDel}:
    \begin{align*}
    \grad_h F_{\mathrm{OPS}} =
    -   e^{\mathbf x /2} e^z[-\mathbf x,0]_z ( R_{\Delta_\varphi} )=
    -   \brac{e^{\mathbf x /2} e^z[-\mathbf x,0]_z e^{\mathbf x /2}}
    ( R_{\Delta_k } )=
    - \exp[0,\mathbf x] ( R_{\Delta_k} ), 
    \end{align*}
    where we have used the multiplicative property of the exponential function: 
    \begin{align*}
      e^x e^z[-x,0]_z = e^z[-x + x , x]_z = e^z[0,x]_z   .
    \end{align*}
    The formula for $\grad_k F_{\mathrm{OPS}}$ follows immediately from 
    the relation in \cref{eq:grad_hF-to-grad_kF}.
\end{proof}

\subsection{Modular curvature on higher dimensional $\mathbb{T}^m_\theta$}
\label{subsec:mdocur-on-T^m-theta}

Among known examples of noncommutative manifolds, 
not even the conformal geometry has
been fully explored. The author did a case study on 
a larger class toric noncommutative manifold
in \cite{LIU2017138}, and on noncommutative tori,
some variations related to exploration of scalar 
curvature \footnote{a more complete list of references can be found 
in surveys  \cite{lesch2018modular,fathizadeh2019curvature}.}
can be found in \cite{Bhuyain2012}, 
\cite{MR3402793} and \cite{1811.04004}. 
To go beyond conformal geometry on $\mathbb{T}^m_\theta$,
Ponge and Ha \cite{ha2019laplace} proposed a construction of Laplace-Beltrami
operators for more general Riemannian metrics
(influenced by \cite{2013SIGMA...9..071R}),
but no explicit forms of local invariants are available at the current stage.

In the paper, we  take the setting in  dimension two
$\Delta \mapsto  \Delta_\varphi = e^{h / 2} \Delta e^{h / 2}$, where 
$h = h^*\in C^\infty(\mathbb{T}^m_\theta)$,
as a simplified model of the conformal change of metric $g \mapsto  e^{-h}g$
for general noncommutative tori $\mathbb{T}^m_\theta$. 
What is crucial to our discussion is some variational properties 
like Proposition \ref{prop:varprob-V_j-step1} which
shall hold for so-called conformally covariant operators in the sense of 
\cite[Eq. (1.1)]{branson1986conformal}. The squared spinorial Dirac operator
on toric noncommutative manifold in \cite{LIU2017138}
$ \slashed D_h^2$, where  $ \slashed D_h = e^h \slashed D e^h$, 
is other set of examples in noncommutative geometry.

It well-known that, on Riemannian manifolds $(M,g)$, the EH
(Einstein-Hilbert) action $g \mapsto  \int_M S_g dg$
has a spectral realization as the second heat coefficient of 
the scalar Laplacian $\Delta_g$:
\begin{align*}
    f \mapsto  V_2(f,\Delta_g) = c_m \int_M (f S_g) / 6 dg 
\end{align*}
for some  universal constant $c_m$ (depending on the dimension).
Moreover, when $\dim M \ge 3$, that is without the restriction of Gauss-Bonnet,
the EH-action has non-trivial variation inside a conformal class of metrics, 
and the functional gradient gives the scalar curvature $S_g$ a variational
interpretation.
Consequently, on noncommutative tori $\mathbb{T}^m_\theta$, $m \ge 2$, we 
define the modular curvature to be the functional gradient of the
following  Riemannian functional.
\begin{defn}
    \label{defn:F-defn-EH+OPS}
    On $\mathbb{T}_\theta^m$, we view 
    $\Delta \mapsto  \Delta_\varphi = e^{h/2} \Delta e^{h/2}$ as a conformal
    change of the flat metric represented by $\Delta$ in
    \cref{eq:flat-del-m-nctori}. 
The Riemannian functional $F$ is defined on the conformal class of metrics
parametrized by Weyl factors $k = e^h$, whose tangent space consists of
self-adjoint elements $h  = h^*\in  C^\infty(\T^m_\theta)$: 
\begin{align}
    \label{eq:F-defn-EH+OPS}
    F(k) \defeq F(h) = \begin{cases}
        F_{\mathrm{OPS}} (h) = 
     \zeta'_{\Delta_\varphi}(0) + \log \varphi_0(k^{-1}),
        , &\text{ if $m=2$, }
        \\
        F_{\mathrm{EH}}(h) = V_2(1,\Delta_{\varphi})
        , &\text{ if $m>2$, }
    \end{cases}
\end{align}
where $V_2(\cdot ,\Delta_\varphi)$ is the second heat coefficient in 
\cref{eq:heat-asym-delphi}.
\end{defn}

\begin{defn}[Moduar Curvature]
    \label{defn:gradF-defn-EH+OPS}
   We call the  gradients of $F$ above (in terms of $k$ or $h$), 
    \begin{align*}
    \grad_k
 F \; \text{or} \;
 \grad_h F \in C^\infty(\T^m_\theta)    
    \end{align*}
the \emph{modular curvature} of the metric associated with 
the Laplacian $\Delta_\varphi$.
\end{defn}
Again, the gradients are defined by means of G\'ateaux differential 
 using the inner product associated to the canonical trace $\varphi_0$
 (cf.  \cref{defn:grad_hF-abstract} and \cref{eq:grad_hF-to-grad_kF}).

The variation of $F$ relies on the variation of the heat coefficient
$V_2(\cdot , \Delta_k)$, where $\Delta_k$ and  $\Delta_\varphi$ are related 
in \cref{eq:defn-Del_k} and \cref{eq:Rdelk-vsRdelphi}. 
\begin{prop}
    \label{prop:varprob-V_j-step1}
As before, consider  the variation along a self-adjoint direction
$a \in C^\infty(\mathbb{T}_\theta^m)$: $h \mapsto  h+ \varepsilon a$ and  
set $\pmb\delta_a \defeq d /d\varepsilon |_{ \varepsilon_0}$. 
Then for heat coefficients in in \cref{eq:heat-asym-delphi},
we have for $l=0,1,2,\dots$, 
    \begin{align}
    \label{eq:varprob-V_j-step1}
        \pmb\delta_a V_l(1, \Delta_k) = 
        \frac{l-m}{2}
        V_l(\pmb\delta_a(k)k^{-1} , \Delta_k )  =
        \frac{l-m}{2}
        V_l\brac{
            \exp[0,-\mathbf x]    \brac{
                \pmb\delta_a(h)
        }, \Delta_k }
        .
    \end{align}
\end{prop}
\begin{proof}
    Start with $
        \pmb\delta_a(\Delta_k) = \pmb\delta_a(k)\Delta =
        \pmb\delta_a(k)k^{-1} \Delta_k $.
    Apply the Duhamel's formula,  for $t>0$,
    \begin{align}
        \label{eq:diff-heatop-varprob}
        \begin{split}
      \pmb\delta_a \Tr \brac{ e^{-t\Delta_k}} &= -t 
    \Tr \brac{
        \pmb\delta_a(\Delta_k)  e^{-t\Delta_k}
    } = \Tr \brac{
        \pmb\delta_a(k) k^{-1} \Delta_k e^{-t\Delta_k}
    }\\
    &= t\frac{d}{dt} \Tr \brac{
    \pmb\delta_a(k) k^{-1} e^{-t\Delta_k }
    }
    .        
        \end{split}
    \end{align}
Since both $\pmb\delta_a$ and $d/dt$  pass through the asymptotic
expansion, we can continue: 
\begin{align*}
    \sum_{l=0}^\infty \pmb\delta_a V_l(1, \Delta_k) t^{(l-m)/2} &= 
     \sum_{l=0}^\infty    V_l(\pmb\delta_a(k) k^{-1}, \Delta_k) 
     t \frac{d}{dt} t^{(l-m)/2} \\
     &= \sum_{l=0}^\infty  \frac{l-m}{2}  V_l(\pmb\delta_a(k) k^{-1}, \Delta_k) 
     t^{(l-m)/2}.
\end{align*}
The first = sign in Eq. \eqref{eq:varprob-V_j-step1} follows from comparing the
coefficients of $t^{(l-m)/2}$.
To see the second equal sign, we  need Lemma 
\ref{lem:1st-derivative-exp-power}:
\begin{align*}
    \pmb\delta_a(k) k^{-1} = 
    e^h \frac{e^{\mathbf x} -1}{ \mathbf x} \pmb\delta_a(h) k^{-1}  
    = \mathbf y^{-1} \brac{\frac{e^{\mathbf x} -1}{ \mathbf x}   }
    ( \pmb\delta_a(h) )
    = \frac{1- e^{-\mathbf x} }{ \mathbf x} ( \pmb\delta_a(h) ).
\end{align*}
\end{proof}

\begin{cor}
        \label{cor:gradFEH-V2-varprob}
    When the dimension $m \ge 2$, we have 
    \begin{align}
        \label{eq:gradFEH-V2-varprob}
        \grad_k F_{\mathrm{EH}} = \frac{2-m}{2} k^{-1} R_{\Delta_k}, 
\;\;
        \grad_h F_{\mathrm{EH}} = \frac{2-m}{2} \exp[0,\mathbf x]( R_{\Delta_k} )
        .
    \end{align}
\end{cor}
\begin{rem}
    When $m =2$, it follows immediately that $\grad_k F_{\mathrm{EH}} = \grad_h
    F_{\mathrm{EH}} =0$. The result is known as Connes-Moscovici's variational
    proof of Gauss-Bonnet theorem on noncommutative two tori $\T^2_\theta$.
\end{rem}
\begin{proof}
For the EH-action $F_{\mathrm{EH} } $:
\begin{align*}
    \pmb\delta_a F_{\mathrm{EH}}(k) = \pmb\delta_a V_2(1,\Delta_k) = 
    \frac{2-m}{2} V_2( \pmb\delta_a (k) k^{-1}, \Delta_k) = 
 \frac{2-m}{2}   \varphi_0\left( 
\pmb\delta_a (k) 
( k^{-1} R_{\Delta_k} )
    \right).
\end{align*}
Similarly,
\begin{align*}
    \pmb\delta_a F_{\mathrm{EH}}(h) &= \pmb\delta_a V_2(1,\Delta_k) = 
    \frac{2-m}{2} V_2( \exp[0,-\mathbf x] \pmb\delta_a (h) ,\Delta_k) \\
    &= 
 \frac{2-m}{2}   \varphi_0\left( 
     \exp[0,-\mathbf x] ( \pmb\delta_a (h)  )
 R_{\Delta_k}
    \right) 
=    
 \frac{2-m}{2}   \varphi_0\left( 
    \pmb\delta_a (h)    
 \exp[0,\mathbf x] ( R_{\Delta_k} )
    \right) .
\end{align*}
\end{proof}

 \section{Local Expressions of the geometric functional }
 \label{sec:closed-formulas-EH-OPS}

\subsection{Functional density of $V_2(\cdot,\Delta_k)$}
\label{subsec:fun-density-V2}
The full expression of $R_{ \Delta_{ \varphi } } $ on $\mathbb{T}^2_\theta$
was first achieved in \cite{MR3194491}
and further confirmed via independent calculation in \cite{MR3148618}.
In this section, we shall start with the following version obtained 
in part I of the sequel \cite[\S4]{Liu:2018aa} which works for 
arbitrary dimension $\mathbb{T}^m_\theta$. We will not repeat the 
computation here.
\begin{prop}
    \label{prop:varprob-scacur-for-Delta-varphi}
    For the perturbed Laplacian $\Delta_k = k \Delta$ acting on 
    $C^\infty(\mathbb{T}^m_\theta)$ with $m\ge 2$, the corresponding 
    functional density of the second heat coefficient $R_{\Delta_k}$ is given by,
  upto a  constant factor $\op{Vol}(\mathbb S^{m-2})/2$, 
\begin{align}
    \label{eq:varprob-scacur-for-Delta-varphi}
    R_{\Delta_k}  =
    \sum_{\alpha = 1}^m
    k^{-m/2} K_{\Delta_k} ( \mathbf y)(\nabla_\alpha^2 k;m)
    + k^{-m/2-1}  H_{\Delta_k}( \mathbf y_1, \mathbf y_2;m)
    (\nabla_\alpha k \otimes  \nabla_\alpha k) 
    .
\end{align}
The one-variable function  is given by
 \begin{align}
     K_{\Delta_k} (y;m) = \frac4m H_{3,1}(z;m) - H_{2,1}(z;m), \,\,\, z=1-y .
     \label{eq:b2cal-KDeltak}
 \end{align}
  Similarly, we have:
 \begin{align}
\label{eq:b2cal-HDeltak}
\begin{split}
      H_{\Delta_k} (y_1 , y_2 ;m)  &=
      (\frac{4 }{m}+2) H_{2,1,1}\left(z_1,z_2;m\right)  -\frac{4 (1-z_1)
     H_{2,2,1}\left(z_1,z_2 ;m\right)}{m}
     \\ &  
     -\frac{8
     H_{3,1,1}\left(z_1,z_2;m\right)}{m},
\end{split}
    \end{align}
    where $z_1 = 1- y_1$ and $z_2 = 1-y_1 y_2$. 
The precise definition of the hypergeometric families 
$\set{H_{ a,b} , H_{ a,b,c } }$ will be need in the next section,
see \cref{eq:hgeofun-onevar-K_acm,eq:Habc}. 
   \end{prop}
   \begin{rem}
    According to Eq. \eqref{eq:Rdelk-vsRdelphi}, we see that
$R_{\Delta_\varphi} $ is also of the form 
Eq. \eqref{eq:varprob-scacur-for-Delta-varphi}, 
with spectral functions: 
\begin{align}
\label{eq:varprob-K-Delta-varphi-vs-k}
K_{\Delta_\varphi} (y;m) &= \sqrt{y} K_{\Delta_k} (y;m), \\
    H_{\Delta_\varphi} (y_1,y_2;m) & = \sqrt{y_1 y_2}  H_{\Delta_k}
    (y_1, y_2;m).
    \label{eq:varprob-H-Delta-varphi-vs-k}
\end{align}
\end{rem}

\begin{prop}
\label{prop:RDeltak-from-k-logk}
In terms of $h = \log k \in C^\infty(T_\theta^m)$, the functional density
$R_{\Delta_k}$ can be rewritten as: 
\begin{align}
\label{eq:RDeltak-in-terms-of-h}
    R_{\Delta_k} = \sum_{\alpha = 1}^m
    e^{-m/2+1} \brac{
        \tilde K_{\Delta_k}(\mathbf x)
    (\nabla_\alpha^2 h) +
        \tilde H_{\Delta_k}(\mathbf x_1 , \mathbf x_2)
        (\nabla_\alpha h \otimes  \nabla_\alpha h)
},
\end{align}
with
\begin{align}
\begin{split}
    \tilde K_{\Delta_k}( x) &= 
    \exp[0, x] K_{\Delta_k}(e^{ x})\\
    \tilde H_{\Delta_k}( x_1 ,  x_2)
    &= e^{ x_1} \exp[0, x_1] \exp[0, x_2] 
    H_{\Delta_k}(e^{ x_1}, e^{ x_1})\\
    & + 2 K_{\Delta_k}(e^{ x_1+x_2 })
    \exp[0, x_1,  x_1 +  x_2].
\end{split}
\label{eq:RDeltak-from-k-logk}
\end{align}
\end{prop}
\begin{proof}
    One simply follows Lemma \ref{lem:ch-var-Kand-H} 
    to carry out the change of variable 
$k \mapsto  h = \log k$ and $\mathbf y \mapsto \mathbf x$.
The details are left to the reader.
\end{proof}

\subsection{The EH (Einstein-Hilbert) action}
\label{subsec:EH-action-closedfor}

Recall that
\begin{align*}
    F_{\mathrm{EH}} (k) \defeq V_2(1,\Delta_\varphi) = V_2(1,\Delta_k)
    = \varphi_0(R_{\Delta_k}),
\end{align*}
where $ \varphi_0( R_{\Delta_k } )$ can be further simplified using the operators 
$\pmb\iota:C_{\mathscr S} (\R^n) \rightarrow C_{\mathscr S}(\R^{n-1})$
 and $\pmb\eta:C_{\mathscr S}(\R_+^n) \rightarrow C_{\mathscr S}(\R_+^{n-1})$
in Lemma \ref{lem:iota-k-defn} and \ref{lem:eta-h-defn} respectively.
\begin{prop}
    \label{prop:varprob-Feh-closedform}
    On $\T^m_\theta$, the Einstein-Hilbert action defined above
    admits the following closed
    formula. In terms of the Weyl factor $k$:
    \begin{align}
\label{eq:varprob-Feh-closedform}
     F_{\mathrm{EH}} (k) =  
     \sum_{\alpha=1}^m
     \varphi_0 \brac{
    k^{-m/2-1} T_{\Delta_k} (\mathbf y;m)
    (\nabla_\alpha k) \cdot (\nabla_\alpha k)
    }.
   \end{align}
 In terms of $h = \log k$:
 \begin{align}
     F_{\mathrm{EH}} (h) =  
     \sum_{\alpha=1}^m 
     \varphi_0\left( e^{(-m/2+1)h}
         \tilde{T}_{\Delta_k}(\mathbf x;m)
         (\nabla_\alpha h) \cdot (\nabla_\alpha h)
     \right).
     \label{eq:FEH-closed-formula-h}
 \end{align}
 The function $T_{\Delta_k}$ and $\tilde T_{\Delta_k}$ are determined by the
 spectral functions in Eq. \eqref{eq:varprob-scacur-for-Delta-varphi} and
 \eqref{eq:RDeltak-in-terms-of-h} respectively:
\begin{align}
    \label{eq:varprob-defn-TDleta}
    T_{\Delta_k}( y ;m) 
    &= -\pmb\eta( K_{\Delta_k})  z^{-m/2}[1, y]
    + \pmb\eta( H_{\Delta_k}) (y) 
    \\
    &= -K_{\Delta_k}(1;m) \frac{y^{-m/2} -1}{y-1}
    + H_{\Delta_k}( y,y^{-1};m). \nonumber
\end{align}
Also,
\begin{align}
    \label{eq:varprob-defn-tilde-TDleta-h}
    \tilde T_{\Delta_k}( x ;m) 
    &= -\pmb\iota( \tilde K_{\Delta_k}) e^{(-m/2+1)z}[0,\mathbf x]_z 
    + \pmb\iota( \tilde H_{\Delta_k}) (x) 
    \\
    &= -\tilde K_{\Delta_k}(0;m) \frac{e^{-mx/2} -1}{x}
    + \tilde H_{\Delta_k}( x,-x;m). \nonumber
\end{align}
In term of hypergeometric functions:
\begin{align}
    \label{eq:TDelk-Hab} 
    T_{\Delta_k} (y;m) &= 
     (\frac4m +2) H_{3,1} (1-y;m) 
    - \frac{4y}{m}  H_{3,2}(1-y;m)  
    -\frac8m H_{4,1}(1-y;m)  \\
    &- \pmb\eta( K_{\Delta_k}) z^{-m/2}[1, y] .
\nonumber 
\end{align}
\end{prop}
\begin{proof}

 Let us apply the volume functional $\varphi_0$ in Eq.
 \eqref{eq:varprob-scacur-for-Delta-varphi}. 
 The second terms is handled by Lemma \ref{lem:eta-h-defn}:
 \begin{align*}
      \varphi_0\brac{
            k^{-m/2 -1} H_{\Delta_\varphi}(\mathbf y_1, \mathbf y_2)
            (\nabla_\alpha k \otimes  \nabla_k k)
        } &=
        \varphi_0\left( 
            k^{-m/2-1} \pmb\eta( H_{\Delta_k} ) (\mathbf y) (\nabla_\alpha k)
            \cdot (\nabla_\alpha k)
        \right), 
 \end{align*}
 and the first term involves the classical integration by parts 
 regarding to $ \nabla_\alpha$ 
    \begin{align*}
        \varphi_0\brac{
            k^{-m/2} K_{\Delta_\varphi}(\mathbf y)(\nabla^2 k)
        } &= 
       \pmb\eta ( K_{\Delta_\varphi}  ) 
        \varphi_0\brac{
        k^{-m/2} \nabla_\alpha^2 k} = 
       - \pmb\eta ( K_{\Delta_\varphi}  ) 
        \varphi_0\brac{
        \nabla_\alpha  (k^{-m/2}) \nabla_\alpha (k)}  
        \\ &= -
       \pmb\eta ( K_{\Delta_\varphi}  ) 
        \varphi_0 \brac{
            k^{-m/2-1} (z^{-m/2}[0,\mathbf y]_z)
            (\nabla_\alpha k) \cdot (\nabla_\alpha k)
        }
    ,
    \end{align*}
    here we have used Eq. \eqref{eq:delk^j}:
    $\nabla_\alpha k^{-m/2} = k^{-m/2-1}
    (z^{-m/2}[0,\mathbf y]_z) (\nabla_\alpha k)$. By adding up the two terms,
    we have finished the proof of 
    \cref{eq:varprob-Feh-closedform,eq:varprob-defn-TDleta}. 
    We leave the parallel verification of 
    \cref{eq:FEH-closed-formula-h,eq:varprob-defn-tilde-TDleta-h}
    to the reader.   

  Last but not least,    \cref{eq:TDelk-Hab} is obtained by 
    expressing  $\pmb\eta( H_{\Delta_k} )$ in terms of  the Gauss
   hypergeometric functions $H_{ a , b} $ according to the relation 
   in \cref{eq:Habcz0to-Hab}.
\end{proof}

Now, we are ready to apply \cref{thm:CM-intermsof-k} (resp.  \cref{thm:gradF-h})
to \cref{eq:varprob-Feh-closedform} (resp. \cref{eq:FEH-closed-formula-h}) 
with the parameter $j = -m/2-1$ (resp. $j=-m/2+1$).

\begin{prop}
    \label{prop:gradFEH}
    Let $T_{\Delta_k}$ and $\tilde T_{\Delta_k}$ be the spectral functions
    defined in 
    \textnormal{Prop. \ref{prop:varprob-Feh-closedform}} with the dimension
    $m\ge 2$.
  The functional gradient at metric $k$ is given by:  
    \begin{align}
        \label{eq:gradFEH-k-from-localexps}
        \grad_k F_{\mathrm{EH}} = \sum_{\alpha = 1}^m
        k^{-m/2-1} K_{\mathrm{EH}}(\mathbf y)
        (\nabla_\alpha^2 k) +
        k^{-m/2-2} H_{\mathrm{EH}}(\mathbf y_1 , \mathbf y_2)
        ( \nabla_\alpha k \cdot \nabla_\alpha k)
        ,
    \end{align}
    with
    \begin{align}
        \label{eq:CM-KEH-k}
    K_{\mathrm{EH}}  &=- (1+\pmb\sigma_{-m/2-1}) (T_{\Delta_k}), \,\,\,
     \\
    H_{\mathrm{EH}} &= (1+\pmb\sigma_{-m/2-2} - \pmb\sigma^2_{-m/2-2})
    \cdot \blacksquare^+(K_{\mathrm{EH}} ) .
        \label{eq:CM-HEH-k}
    \end{align}
    In terms of the log-Weyl factor $h$, 
\begin{align}
        \grad_h F_{\mathrm{EH}} = e^{( -m/2+1 ) h} 
       \left( \sum_{\alpha = 1}^m
         \tilde K_{\mathrm{EH}}(\mathbf x)
        (\nabla_\alpha^2 h) +
        \tilde H_{\mathrm{EH}}(\mathbf x_1 , \mathbf x_2)
        ( \nabla_\alpha h \cdot \nabla_\alpha h)
    \right),
    \label{eqLgradFEH-h-from-localexps}
\end{align}
with
\begin{align}
        \label{eq:CM-KEH-h}
    \tilde K_{\mathrm{EH}}  &=
    - (1+\pmb\tau_{-m/2+1}) ( \tilde T_{\Delta_k}), \\
\tilde H_{\mathrm{EH}} &= (1+\pmb\tau_{-m/2+1} - \pmb\tau^2_{-m/2+1})
    \cdot  \blacktriangle^+(\tilde K_{\mathrm{EH}} )
    .
        \label{eq:CM-HEH-h}
\end{align}
\end{prop}

In dimension two, we have $ \grad_k F_{\mathrm{EH}} = 0$, which is 
nothing but another version of the Gauss-Bonnet theorem
first achieved in \cite{Connes:2011tk}. In fact, the very first nontrivial
discovery on the $a_2$-term is  the observation that the function 
$\tilde T_{\Delta_k} (x)$ (denoted by $K(x)$ in 
\cite[Lemma 3.3]{Connes:2011tk}) is odd and the vanishing of the 
functional $F_{ \mathrm{EH}} $ follows immediately from a one-line argument
cf. \cite[\S3.3]{Connes:2011tk}), in particular, no need to look at 
the gradient.
The Proposition above confirms the vanishing of the gradient. Indeed,
$\tilde T_{\Delta_k} (x)$ is odd means that
\begin{align*}
    (1+ \pmb\tau_0)\tilde T_{\Delta_k}(x) 
    = \tilde  T_{\Delta_k}(x) + \tilde T_{\Delta_k} (-x) = 0,
\end{align*}
so that $K_{\mathrm{EH}} = \tilde K_{\mathrm{EH}} =0$ and then 
$\grad_h F_{\mathrm{EH}} = 0$ according to \cref{eq:CM-KEH-h,eq:CM-HEH-h}.
We will provide another verification of the oddness 
in Prop. \ref{prop:tobechecked-KDelvs-TDel} to illustrate the power of 
of the hypergeometric family $H_\alpha$ obtained in \cite{Liu:2018aa}.

\subsection{Osgood-Phillips-Sarnak (OPS)
functional
}
\label{subsec:OPS-local}

Let us specialized to noncommutative tori $m = 2$ so that
 $ \grad_k F_{\mathrm{EH}} = 0$, which leads to another functional 
 $F_{ \mathrm{OPS}} $. 
For $s \in  \R$, put $\Delta_{k_s} = k^s \Delta$. With the relation  
\cref{eq:Rdelk-vsRdelphi} in mind and the fact that $[h,k] =0$, we see that
$V_2(h,\Delta_{k_s}) = V_2(h,\Delta_{\varphi_s})$ and then  
\cref{eq:FOPS-in-terms-of-RDel} becomes:
\begin{align*}
     F_{\mathrm{OPS}}(k) =  -\int_0^1 V_2(h,\Delta_{k_s}) ds + 
     \zeta'_{\Delta}(0).
\end{align*}
The proposition below is another version of 
\cite[Thm 4.6]{MR3194491} in which the computation was carried out in terms of
$h$. We make a few comparisons before proceeding to the statement and the proof:
\begin{enumerate}
    \item  In terms of $h$, the factor $e^{-m / 2+1}$ in 
        \cref{eq:RDeltak-in-terms-of-h} disappears when $m=2$, while 
\cref{eq:varprob-scacur-for-Delta-varphi} reads
\begin{align}
    \label{eq:RDeltak-m=2}
       R_{\Delta_k} = \sum_{\alpha=1}^2
        k^{-1} K(\mathbf y) (\nabla_\alpha^2 k)
        + k^{-2} H(\mathbf y_1, \mathbf y_2)
        (\nabla_\alpha k \nabla_\alpha k) .
\end{align}
The factors $k^{-1}$ and $k^{-2}$ indeed bring in extra work which makes 
the appearance of 
$T_{\mathrm{OPS}}$ in \cref{eq:Tzeta-varprop} more intimating than its 
counterpart $K_2$ in \cite[Lemma 4.4]{MR3194491}.

\item We intentionally perform  calculation in terms of $k$
 to  illustrate subtleties  behind the change of variable  $h \mapsto  k = e^h$,
 by comparing with the parallel computation in \cite[\S4]{MR3194491}.

\item The Gauss-Bonnet functional equation $(1+\pmb\sigma_{-2})(T) = 0$,
    with $T$ defined  in  \cref{eq:T-in-dim-2},  plays 
    a key role in the computation, which should be compared with 
\cite[Lemma 4.3]{MR3194491}.
\end{enumerate}

\begin{prop}
\label{prop:Tzeta-varprop}
On  $\T^2_\theta$, the OPS-functional has the following local expression:
    \begin{align}
        F_{\mathrm{OPS}}(k) 
     =  \sum_{\alpha=1}^2
     \varphi_0\brac{
         k^{-2} T_{\mathrm{OPS}} (\mathbf y)
         (\nabla_\alpha k) \cdot \nabla_\alpha k
        }  + \zeta'_{\Delta}(0),
        \label{eq:Tzeta-varprop-defn}
    \end{align}
where the spectral function is given by
\begin{align} 
\label{eq:Tzeta-varprop}
  T_{\mathrm{OPS}} (y) &= \mathrm{I}_{T_{\zeta'}}(y)
+ \mathrm{II}_{T_{\zeta'}}(y)\\
&=
 \pmb\eta(K)
\ln[1,y] \int_0^1 \pmb\sigma_{-1}\left( 
    z^s[1,y]
\right) ds
+\frac12 \int_0^1 (z^{s}[1,y]_z)^{2}
T(y^s;2) \ln y ds .
 \nonumber 
\end{align}
where $K$, $H$ and $T$ are the abbreviations listed in
\cref{eq:KHindimtwo-varprob} and $\pmb\eta(K) = K(1) = 1/6$.
\end{prop}
\begin{rem}
    The validation of Eq. \eqref{eq:Tzeta-varprop} will be further confirmed in 
    \S\ref{subsec:2nd-proof-CM-T-K} by explicit verification of the functional
    relation Eq.  \eqref{eq:KDelk-TOPS-tobechecked}, 
\end{rem}
\begin{proof}
    We shall abbreviate the spectral functions 
    \cref{eq:b2cal-KDeltak,eq:b2cal-HDeltak,eq:varprob-defn-TDleta} 
    related to $R_{\Delta_k}$ by
    dropping the subscript:
    \begin{align}
        \label{eq:KHindimtwo-varprob}
        K(y) \defeq K_{\Delta_k}(y;2) ,\,\,\, T(y) \defeq T_{\Delta_k}(y;2)
        ,\,\,\, H(y_1,y_2) \defeq H_{\Delta_k}(y_1,y_2;2), 
    \end{align}
     so that for $m=2$, \cref{eq:varprob-defn-TDleta} reads:
    \begin{align}
        \label{eq:T-in-dim-2}
        T(y) = \pmb\eta(K)y^{-1} + \pmb\eta(H)(y).  
    \end{align}
    The $R_{\Delta_{k_s}}$ in \cref{eq:RDeltak-m=2} yields: 
    \begin{align*}
        V_2(h,\Delta_{k_s})  &= \varphi_0(h R_{\Delta_{k_s}} 
        ) \\ 
        &= \pmb\eta(K) 
        \varphi_0\brac{ h k^{-s} \nabla_\alpha^2 k^s }  
        + \varphi_0\brac{
            h k^{-2s} \pmb\eta(H)(\mathbf y^s )
            (\nabla_\alpha k^s) (\nabla_\alpha k^s)
        } ,
    \end{align*}
where Lemma \ref{lem:eta-h-defn} has been applied.
 For the   first term:
\begin{align*}
    \varphi_0\brac{ h k^{-s} \nabla_\alpha^2 k^s }  & = -
\varphi_0\brac{ (\nabla_\alpha h) k^{-s} \nabla_\alpha k^s }  
-  \varphi_0\brac{ h (\nabla_\alpha k^{-s} )\nabla_\alpha k^s }  \\
&=- \varphi_0\brac{ (\nabla_\alpha h) k^{-s} \nabla_\alpha k^s }  
+  \varphi_0\brac{ h k^{-2s} \mathbf y^{-s}(\nabla_\alpha k^{s} ) 
    (\nabla_\alpha k^s)
} ,
\end{align*}
where we  need 
$\nabla_\alpha k^{-s} = -k^{-s} (\nabla_\alpha k^s) k^{-s} = - \mathbf y^{-s}
(\nabla_\alpha k^s)$. 
To continue, with Eq. \eqref{eq:T-in-dim-2} in mind: 
\begin{align}
    \begin{split}
        &\;\; V_2(h,\Delta_{k_s}) \\
        =&\;\;  -\pmb\eta(K) 
\varphi_0\brac{ (\nabla_\alpha h) k^{-s} \nabla_\alpha k^s }   
 +  \varphi_0\brac{ h k^{-2s} 
     \left[  
         \pmb\eta(K) \mathbf y^{-s} + \pmb\eta(H) (\mathbf y^s)
         \right]
         (\nabla_\alpha k^{s}) (\nabla_\alpha k^s) 
}  \\
=&\;\; -\pmb\eta(K)
\varphi_0\brac{ (\nabla_\alpha h) k^{-s} \nabla_\alpha k^s }   
  +  \varphi_0\brac{ h k^{-2s} 
     T(\mathbf y^s) (\nabla_\alpha k^{s})\cdot (\nabla_\alpha k^s) 
 }  .
\end{split}
    \label{eq:V2Deltaks-varprob}
\end{align}
We would like to convert the two terms  in the RHS of \eqref{eq:V2Deltaks-varprob}
into the desired form in Eq. \eqref{eq:Tzeta-varprop-defn}.
The calculation  of the first term is postponed to Lemma
\ref{lem:1term-Tzata-varprop}, which leads to the first part of
\eqref{eq:Tzeta-varprop}:
\begin{align}
\label{eq:eq1-Tzeta-varprop}
\mathrm{I}_{T_{\zeta'}}(y)= 
 \pmb\eta(K)
\ln[1,y] \int_0^1 \pmb\sigma_{-1}\left( 
    z^s[1,y]
\right) ds .
\end{align}
The $-1$ in front of $\pmb\eta(K)$ in \eqref{eq:V2Deltaks-varprob} disappears
because of the minus sign
in the definition of $F_{\mathrm{OPS}}$.
 There is a extra factor  $h$
 in the second term of Eq. \eqref{eq:V2Deltaks-varprob},
 which can be turned into a modular derivation provided 
 the Gauss-Bonnet functional equation $(1+\pmb\sigma_{-2})(T) = 0$:
\begin{align*}
    \varphi_0\brac{
        k^{-2s} T(\mathbf y^s)(\nabla_\alpha k^s \cdot h)  (\nabla_\alpha k^s)
    } 
    &= \varphi_0\brac{
        k^{-2s} \pmb\sigma_{-2s}(T(\mathbf y^s) )
        (\nabla_\alpha k^s \cdot h)
    } \\
    &= - \varphi_0\brac{
        h k^{-2s} T(\mathbf y^s)(\nabla_\alpha k^s) \cdot (\nabla_\alpha k^s)
    }.
\end{align*}
 It follows that 
\begin{align*}
    \varphi_0\brac{
        h k^{-2s} T(\mathbf y^s)(\nabla_\alpha k^s)  (\nabla_\alpha k^s)
    } &= \frac12
\varphi_0\brac{
    \op{ad}_h \left[ k^{-2s} T(\mathbf y^s)(\nabla_\alpha k^s) \right]
    (\nabla_\alpha k^s)
    } \\
    &= -\frac12 \varphi_0\brac{
         k^{-2s}  \ln \mathbf y T(\mathbf y^s)
         (\nabla_\alpha k^s)  (\nabla_\alpha k^s)
    } .
\end{align*}
Before integrating in $s$, we further move the parameter to the modular operator.
\begin{align*}
    (\nabla_\alpha k^s)(\nabla_\alpha k^s)& 
    = (k^{s-1} z^s[1,\mathbf y]_z(\nabla_\alpha k) )^2 \\
&=k^{2(s-1)}  (\mathbf y_1)^{s-1}
(z^s[1,\mathbf y_1]_z) ( z^s[1,\mathbf y_2]_z)
(\nabla_\alpha k  \otimes  \nabla_\alpha k),
\end{align*}
in which the spectral function becomes:
\begin{align*}
    \pmb\eta( y_1^{s-1} (z^z[1,y_1]_z )  (z^z[1,y_2]_z) ) (y)=
 (z^z[1,y]_z)^2 
\end{align*}
after applying $\varphi_0$, that is
\begin{align*}
    \varphi_0\brac{
        h k^{-2s} T(\mathbf y^s)(\nabla_\alpha k^s)  (\nabla_\alpha k^s)
    }  &= -\frac12
\varphi_0\brac{
    k^{-2}  \left( 
 \ln \mathbf y T(\mathbf y^s) 
 (z^z[1,\mathbf y]_z)^2 
    \right)
    (\nabla_\alpha k)  (\nabla_\alpha k)
},
\end{align*}
which, after integration $\int_0^1$, gives rise exactly to the second part 
$\mathrm{II}_{T_{\zeta'}}(y)$ in \cref{eq:Tzeta-varprop}.
The proof is complete.
\end{proof}

Here is the last piece needed for \cref{eq:eq1-Tzeta-varprop}.
\begin{lem}
\label{lem:1term-Tzata-varprop}
Keep notations. For $\alpha = 1,2$, we have
    \begin{align*}
    \int_0^1 \varphi_0\brac{ (\nabla_\alpha h) k^{-s} \nabla_\alpha k^s }
     ds = \varphi_0 \brac{
        k^{-2} L(\mathbf y)(\nabla_\alpha k) (\nabla_\alpha k)
    } 
    \end{align*} 
    with 
    \begin{align*}
        L(y) = \ln[1,y] \int_0^1 \pmb\sigma_{-1}( z^s[1,y] ) ds 
        = \frac{-y+y \ln y +1}{(y-1)^2 y}.
    \end{align*}
\end{lem}
\begin{proof}
    We start with Lemma \ref{lem:1st-derivative-exp-power}:
    \begin{align*}
        \nabla_\alpha h &= k^{-1} (\exp[0,\mathbf x])^{-1}
        (\nabla_\alpha k)  
        = k^{-1} \ln[1,\mathbf y]
        (\nabla_\alpha k) , \,\,\,
        \nabla_\alpha k^s = k^{s-1}
        z^s[1,\mathbf y]_z.
    \end{align*}
    Hence
\begin{align*}
        & \,\,   \varphi_0\brac{ (\nabla_\alpha h) k^{-s} \nabla_\alpha k^s }
     \\
    =&\,\, \varphi_0\brac{
        k^{-1}  \left(  
            \ln[1,\mathbf y]
        (\nabla_\alpha k) \right)
         \left(  k^{-s} k^{s-1}
             (z^s[1,\mathbf y])
         (\nabla_\alpha k)\right)
     }  \\
     =&\,\,
     \varphi_0\brac{
         k^{-2} 
         \pmb\eta \left( 
             y_1^{-1} \ln[1,y_1] z^s[1,y_2]_z
         \right) 
         \Big|_{y = \mathbf y}
         (\nabla_\alpha k) 
         \cdot (\nabla_\alpha k)
     },
    \end{align*}
    where $\pmb\eta$ is defined in \eqref{eq:defn-bold-eta}:
    \begin{align*}
         \pmb\eta \left( 
             y_1^{-1} \ln[1,y_1] z^s[1,y_2]_z
         \right)  = y^{-1} \ln[1,y] z^s[1,y^{-1}]
         = \ln[1,y] \pmb\sigma_{-1}(z^s[1,y]).   
    \end{align*}
    The result follows from integrating the spectral function in $s$ from $0$
    to $1$.
\end{proof}

Again, upto the constant $\zeta_{\Delta}'(0)$,
the local expression of OPS-functional 
\cref{eq:Tzeta-varprop-defn}
is indeed landed in the setup of
\cref{thm:CM-intermsof-k} with $j=-2$.
As a consequence, the gradient is given as below.
\begin{prop}
    \label{prop:gradFOPS}
    For a Weyl factor $k \in C^\infty(\T^2_\theta)$, the gradient of the
    OPS-functional is given by:   
    \begin{align}
    \label{eq:gradFOPS-KOPS-HOPS}
        \grad_k F_{\mathrm{OPS}} = \sum_{\alpha = 1}^2
        k^{-2} K_{\mathrm{OPS}} (\mathbf y) (\nabla_\alpha^2 k) +
        k^{-3} H_{\mathrm{OPS}} (\mathbf y_1, \mathbf y_2)
        (\nabla_\alpha k \otimes  \nabla_\alpha k),
    \end{align}
 where
 \begin{align}
    \label{eq:KOPS-TOPS} 
    K_{\mathrm{OPS}} &= - (1+\pmb\sigma_{-2})  (T_{\mathrm{OPS}}), \\
     H_{\mathrm{OPS}} &=  
     (1+\pmb\sigma_{-3} - \pmb\sigma_{-3}^2) 
     \cdot \blacksquare^+ (K_{\mathrm{OPS}}) 
     .
    \label{eq:HOPS-KOPS} 
 \end{align}
\end{prop}

\subsection{Functional relations and hypergeometric functions}
\label{subsec:fun-relations}

To sum up, we have given two approaches to $\grad_h F$ 
by performing variation:
on heat traces and zeta functions in \S\ref{sec:var-FEH} and on the local
expression of $F$ in \S\ref{sec:closed-formulas-EH-OPS}.
To put the pieces together, we recall from
\S\ref{subsec:mod-gauss-T2}, \ref{subsec:mdocur-on-T^m-theta}:
\begin{align*}
    \grad_k F_{\mathrm{EH}} = \frac{2-m}{2} k^{-1} R_{\Delta_k} \,\,\,
    \grad_h F_{\mathrm{EH}} = \frac{2-m}{2} \exp[0,\mathbf x] ( R_{\Delta_k} ) 
    \,\,\, .
\end{align*}
When $m=2$:
\begin{align*}
    \grad_k F_{\mathrm{OPS}} = - k^{-1} R_{\Delta_k} , \,\,\,
    \grad_h F_{\mathrm{OPS}} = - \exp[0,\mathbf x] ( R_{\Delta_k} ) \,\,\,
    .
\end{align*}
By comparing the spectral  functions of the rearrangement operators,
we have
\begin{align*}
    K_{\mathrm{EH}} (y;m) &= \frac{2-m}{2} K_{\Delta_k} (y;m), \,\,
    H_{\mathrm{EH}} (y;m) = \frac{2-m}{2} H_{\Delta_k} (y_1,y_2;m), \,\,
    \\
    \tilde K_{\mathrm{EH}} (x;m) &= \frac{2-m}{2} \exp[0,x]
\tilde K_{\Delta_k} (x;m), \,\,
    \tilde H_{\mathrm{EH}} (x;m) = \frac{2-m}{2}  \exp[0,x_1+x_2]
    \tilde H_{\Delta_k} (x_1,x_2;m), \,\,
\end{align*}
In dimension two,
\begin{align*}
    K_{\mathrm{OPS}} (y) &=  -K_{\Delta_k} (y;2), \,\,
    H_{\mathrm{OPS}} (y) =  -H_{\Delta_k} (y_1,y_2;2), \,\, 
    \\
    \tilde K_{\mathrm{OPS}} (x) &=  \exp[0,x]
\tilde K_{\Delta_k} (x;2), \,\,
    \tilde H_{\mathrm{OPS}} (x) =   \exp[0,x_1+x_2]
    \tilde H_{\Delta_k} (x_1,x_2;2). \,\,
\end{align*}
Therefore functional relations  obtained in Propositions
\ref{prop:gradFOPS} and \ref{prop:gradFEH} for $K$ and  $H$ on the left hand sides,
can be transferred to functions $K_{\Delta_k}$ and $H_{\Delta_k}$ 
on the right hand sides. 
\begin{thm}
    \label{thm:CM-TDel-KDel}
    For one-variable functions, we have for $m \ge 2$
    \begin{align*}
        \frac{m-2}{2}  K_{\Delta_k} (y;m) 
        = (1+\pmb\sigma_{-m/2-1}) (T_{\Delta_k}) (y;m).
    \end{align*}
    When $m=2$,
    \begin{align*}
        K_{\Delta_k} (y;2) = (1+\pmb\sigma_{-2}) (T_{\mathrm{OPS}})(y;2),
    \end{align*}
    where $T_{\mathrm{OPS}}$ is defined in Prop. \ref{prop:Tzeta-varprop}.
\end{thm}

\begin{thm}
    \label{thm:CM-KDel-HDel}
    For all $m\ge 2$, the spectral function defined in Eq. 
    \eqref{eq:b2cal-KDeltak} and 
    \eqref{eq:b2cal-HDeltak} in terms of hypergeometric geometric functions
    fulfil the functional relation:
    \begin{align}
        \label{eq:CM-funrel-sym-ver-k}
        H_{\Delta_k}(y_1,y_2;m) = (1+ \pmb\sigma_{-m/2-2} -\pmb\sigma_{-m/2-2}^2)    
        \cdot \blacksquare^+ (K_{\Delta_k}(y;m) ) .
    \end{align}
    For the log-Weyl factor side, recall $\tilde K_{\Delta_k}$ and 
    $\tilde H_{\Delta_k}$ defined in Eq. \eqref{eq:RDeltak-from-k-logk},  
    denote by
    \begin{align}
        \label{eq:sdelk-to-delk-CM-funrel}
        \mathsf{\tilde  K}_{\Delta_k}(x;m)  = \exp[0,x]
        \tilde K_{\Delta_k} (x;m), \,\,\,
        \mathsf{\tilde H}_{\Delta_k} (x_1,x_2;m) = \exp[0,x_1+x_2]
        \tilde H_{\Delta_k} (x_1,x_2;m), \,\,\,
    \end{align}
    then
    \begin{align}
        \mathsf{\tilde H}_{\Delta_k} (x_1,x_2;m) =
        (1+\pmb\tau_{-m/2+1} - \pmb\tau_{-m/2+1}^2) \cdot\blacktriangle^+   
        (\mathsf{\tilde  K}_{\Delta_k}(x;m) )
        \label{eq:CM-funrel-sym-ver-h}
    \end{align}
\end{thm}
\begin{rem}
    In dimension $m=2$, \eqref{eq:CM-funrel-sym-ver-h} reads:
    \begin{align*}
        \mathsf{\tilde H}_{\Delta_k} (x_1,x_2) &=
        (1+\pmb\tau_{0} - \pmb\tau_{0}^2) \cdot\blacktriangle^+   
        (\mathsf{\tilde  K}_{\Delta_k}(x) ) \\
        &= K_{\Delta_k}[x_1, x_1+x_2] +
    K_{\Delta_k}[-x_1, -x_1-x_2]   - K_{\Delta_k}[-x_1, x_2]    \\
&=
\frac{ K_{\Delta_k} (x_1+x_2) - K_{\Delta_k} (x_1)}{ x_1} +
\frac{ K_{\Delta_k} (-x_1-x_2) - K_{\Delta_k} (-x_2)}{ -x_1} 
\\ & - 
\frac{ K_{\Delta_k} (x_1) - K_{\Delta_k} (-x_1)}{x_1 + x_2} .
    \end{align*}
    Provided the fact that $ \pmb\tau_0( K_{\Delta_k} ) =  K_{\Delta_k} $ 
    (that is $K_{\Delta_k}$ is an even function),  the right hand
    side above recovers exactly the original  version 
\cref{eq:intro-CM-original} of Connes-Moscovici.
\end{rem}

    The two versions of the Connes-Moscovici type relation 
    Eq. \eqref{eq:CM-funrel-sym-ver-k} and Eq. \eqref{eq:CM-funrel-sym-ver-h}
represent the same equation with respect to the change of variable
$k = h = \log k$. However, the verification of the equivalence  are much less
straightforward.
\begin{prop}
    \label{prop:k-2-h-CMrel}
   Let us simply define the functions $\mathsf{\tilde K}_{\Delta_k}$ and 
   $\mathsf{\tilde H}_{\Delta_k}$ in terms of $K_{\Delta_k}$ and 
   $H_{\Delta_k}$ according to  
   \cref{eq:RDeltak-from-k-logk,eq:sdelk-to-delk-CM-funrel} without any
   geometric setting.
The functional relations Eq. \eqref{eq:CM-funrel-sym-ver-k} implies 
Eq. \eqref{eq:CM-funrel-sym-ver-h} 
provided the fact \footnote{It follows from the relations in
    Theorem \ref{thm:CM-TDel-KDel}.}
that $\pmb\sigma_{-m/2-1}(K_{\Delta_k}) = K_{\Delta_k},$.
\end{prop}
\begin{proof}
Because of the cyclic property: $\pmb\tau^3_{-m/2+1}=1$ and 
$\pmb\sigma^3_{-m/2-1} =1$, Eq. \eqref{eq:CM-funrel-sym-ver-k} and
\eqref{eq:CM-funrel-sym-ver-h} are equivalent to 
\begin{align}
    \label{eq:eqversion-CM-funrel-sym-ver}
    (1+\pmb\sigma_{-m/2-2}) (H_{\Delta_k})
    = 2 \blacksquare^+(K_{\Delta_k})  , \,\,\,
    \text{resp.} \,\,\,
    (1+\pmb\tau_{-m/2+1}) (\mathsf{\tilde H}_{\Delta_k})
    = 2 \blacktriangle^+( \mathsf{\tilde K}_{\Delta_k})  . 
\end{align}
We now assume the first relation in \eqref{eq:eqversion-CM-funrel-sym-ver} and would like to prove the second one. Starting with 
\begin{align*}
    \blacktriangle^+( \mathsf{\tilde K}_{\Delta_k}) (x_1, x_2)
    =(( \exp[0,z])^2 K_{\Delta_k}(z))[x_1, x_1+x_2]_z
    =
    ( Q^{\mathrm{I}}   + Q^{\mathrm{II}} + Q^{\mathrm{III}}) (x_1, x_2)  
    ,
\end{align*}
where $Q^{\mathrm{I}}$ to $Q^{\mathrm{III}}$ are obtained by applying the Leibniz
property of divided differences:
\begin{align}
    \begin{split}
        Q^{\mathrm{I}} (x_1 ,x_2) & =  
        \exp[0,x_1, x_1+x_2]
        \exp [0,x_1+x_2] K_{\Delta_k}(e^{x_1+x_2} )  ,
        \\
        Q^{\mathrm{II}} (x_1 ,x_2) & = \exp[0,x_1] 
        \exp[0,x_1, x_1+x_2] 
        K(e^{x_1}),
        \\
        Q^{\mathrm{III}} (x_1 ,x_2) & =    \exp[0,x_1] \exp[0,x_1+x_2]
        K(e^{z}) [x_1, x_1+x_2]_z,
    \end{split}
    \label{eq:QI-QIII}
\end{align}
here we have used composition rule of divided differences:
$(\exp[0,z]) [x_1 , x_1+x_2]_z = \exp[0,x_1, x_1+x_2]$.
Also, we will freely use the following fact 
due to the multiplicative nature of the exponential function:
\begin{align*}
    e^{v} (\exp[u_1, \dots, u_n]) = \exp[u_1+v, \dots , u_n+v],
    \,\,\, u_1, \ldots,u_n, v \in \R.
\end{align*}
On the other side, with Eq. \eqref{eq:RDeltak-from-k-logk}, we see that 
\begin{align*}
    \mathsf{\tilde H}_{\Delta_k} (x_1,x_2) = \exp[0,x_1+x_2] 
    \tilde H_{\Delta_k} (x_1,x_2) =
    f_{\mathsf H} (x_1, x_2) H_{\Delta_k} (e^{x_1}, e^{x_2})
    + 2 Q^{\mathrm{I}} (x_1 ,x_2),
\end{align*}
with
\begin{align}
    f_{\mathsf H} (x_1 ,x_2) =  \exp[0,x_1] \exp[x_1,x_1+x_2] \exp[0,x_1+x_2]. 
    \label{eq:f_H-defn}
\end{align}
Now we claim that
\begin{align}
    \label{eq:claim-fH-QIII}
    f_{\mathsf H}  = \pmb\tau^2_{3}(f_{\mathsf H}) , \,\,\,
    Q^{\mathrm{II}} = \pmb\tau^2_{-m/2+1} ( Q^{\mathrm{I}} ). 
\end{align}
If true, we have 
$(1+\pmb\tau^2_{-m/2+1}) ( Q^{\mathrm{I}} ) =  Q^{\mathrm{II}} + Q^{\mathrm{I}} $
and
\begin{align*}
    \pmb\tau^2_{-m/2+1} ( f_{\mathsf H} \cdot H^{\exp}_{\Delta_k}  ) =
    \pmb\tau^2_{3}(f_{\mathsf H}) \cdot \pmb\tau^2_{-m/2+1} (H^{\exp}_{\Delta_k} )
    = f_{\mathsf H} \cdot \pmb\tau^2_{-m/2+1} (H^{\exp}_{\Delta_k} ),
\end{align*}
where $H^{\exp}_{\Delta_k} (x_1,x_2) \defeq H_{\Delta_k}(e^{x_1}, e^{x_2})$.
Then
\begin{align*}
    &\,\,
    (1+\pmb\tau^2_{-m/2-1})   ( f_{\mathsf H} \cdot H^{\exp}_{\Delta_k}  )
  (x_1, x_2) 
   =
  f_H (x_1, x_2)\cdot 
  (1+\pmb\tau_{-m/2+1})  (H^{\exp}_{\Delta_k} ) (x_1, x_2) 
  \\
  =&\,\,
  f_H (x_1, x_2)\cdot 
  (1+ \pmb\sigma_{-m/2-2})  (H_{\Delta_k} ) (y_1, y_2) 
 =  f_H (x_1, x_2)\cdot 
 2 \blacksquare^+ (K_{\Delta_k}) (y_1, y_2) 
 \\
 = &\,\,
 2 f_H (x_1, x_2) (\exp[x_1,x_1+x_2])^{-1} K_{\Delta_k}(e^z)[x_1, x_1+x_2]  
 = 2 Q^{\mathrm{III}} (x_1 ,x_2),
\end{align*}
which concludes the proof of the second relation in Eq.
\eqref{eq:eqversion-CM-funrel-sym-ver}. In the calculation above, we have used
our assumption (first equation in  Eq.
\eqref{eq:eqversion-CM-funrel-sym-ver}) and also Prop.
\ref{prop:compare-blacksq-blacktri} to carefully switch the variational operators
from $\set{\pmb\sigma_{-m/2-2}, \blacksquare^+}$ to 
$\set{\pmb\tau_{-m/2+1}, \blacktriangle^+}$.

Let us check the claim \eqref{eq:claim-fH-QIII}. Recall for $j\in\R$, 
$\pmb\tau_j^2$ is implemented by the substitutions
\begin{align*}
    x_1 \rightarrow x_2, \,\,\, x_2 \rightarrow -x_1-x_2, \,\,\,
    x_1+ x_2 \rightarrow -x_1
\end{align*}
followed by multiplying $e^{jx_1}$. Therefore
\begin{align*}
    \pmb\tau_3^2(f_{\mathsf H})  (x_1,x_2)
    & = e^{3x_1} \exp[0,x_2]
    \exp[x_2, -x_1] \exp[0,-x_1]  \\
    & = \exp[x_1,x_1+x_2]
    \exp[0,x_1+x_2] \exp[0,x_1] =
    f_{\mathsf H} (x_1,x_2),
\end{align*}
and 
\begin{align*}
    \pmb\tau_{-m/2+1} (Q^{\mathrm{I}}) (x_1,x_2)
    & = e^{(-m/2+1) x_1}
    \exp[0,x_2,-x_1] \exp[0,-x_1] K_{\Delta_k}( e^{-x_1})
    \\
   &=
   e^{2x_1}  \exp[0,x_2,-x_1] \exp[0,-x_1] 
   e^{(-m/2-1) x_1}
   K_{\Delta_k}( e^{-x_1})
    \\
  &=  \exp[x_1 ,x_2+x_1, 0] \exp[0,x_1]  
  \left( 
  y_1^{-m/2-1}
  K_{\Delta_k}(y_1^{-1})
  \right)  
    \\
  &=  \exp[x_1 ,x_2+x_1, 0] \exp[0,x_1]  
    K_{\Delta_k}(y_1) 
    = Q^{\mathrm{II}} (x_1 , x_2) ,
\end{align*}
notice that we have used the property $\pmb\sigma_{-m/2-1}(K_{\Delta_k})
= K_{\Delta_k}$  to complete the argument.
\end{proof}

\section{Verification of the functional relations}
\label{sec:ver-of-fun-rel}

It is highly nontrivial to see why the explicit functions 
$K_{ \Delta_k} $ and $H_{ \Delta_k} $ in 
\cref{eq:b2cal-KDeltak,eq:b2cal-HDeltak}
fulfill all the functional relations derived in the previous section. 
In the last part of the paper, 
we give a computer-aid free verification, to further illustrate
the power of hypergeometric family $H_\alpha$, $\alpha \in  \Z^n_{ >0}$
and the variational operators discussed in \S \ref{sec:varcalwrptncvar}.

\subsection{Hypergeometric functions appeared in the rearrangement process}
\label{subsec:hygeo-fun-in-rearr-lem}

Since only the second heat coefficient is concerned, we only need those of
one (Gauss hypergeometric functions $\pFq21$) 
and two variable (Appell's $F_1$ functions).
For $a,b,c\in\Z_+$, we denote:
\begin{align}
    \label{eq:hgeofun-onevar-K_acm}
    H_{a,b}(z;m) &= 
     \frac{ \Gamma(d_m )}{\Gamma(a) \Gamma(b) }
    \int_0^1
    (1-t)^{a-1} t^{b-1} (1- z t)^{-d_m } d t \\
  &=    \frac{\Gamma(  d_m   )}{ \Gamma(a+b)} 
  \pFq21 (d_m ,b;b+a;z),      \nonumber
\end{align}
where $m$ is the dimension and $d_m = a+b + m/2-2$. In a like manner, 
\begin{align}
\label{eq:Habc}
    & H_{a,b,c}(z_1, z_2;m )  \\
    = & \; 
    \frac{ \Gamma(d_m )}{\Gamma(a) \Gamma(b) \Gamma(c) }
       \int_0^1 \int_0^{1-t} (1 - t -u)^{a-1}
t^{b-1} u^{c-1} (1- z_1 t - z_2 u)^{-d_m } dudt 
 \nonumber  \\ 
=&\;   \frac{\Gamma(  d_m   )}{ \Gamma(a+b +c)} 
    F_1(d_m ;b,c,a+b+c; z_1 , z_2) ,   \nonumber  
\end{align}
where $d_m = a+b + c + m/2-2$. 
They fulfill the following differential  and divided difference relations 
(cf.  \cite[Theorem 3.3]{Liu:2018aa}):
\begin{align}
    \label{eq:Hab-to-Ha1}
    H_{a,c}(z;m) &= \frac{1}{(c-1)!} \frac{d^{c-1}}{dz} H_{a,1}(z;m) ,
    \\
    \label{eq:Habc-to-Ha11}
    H_{a,b,c}(z_1, z_2;m) &=  
    \frac{\partial_{z_1}^{b-1}}{(b-1)!} 
    \frac{\partial_{z_2}^{c-1}}{(c-1)!} 
    H_{a,1,1}(z_1, z_2;m)   ,
    \\
    \label{eq:Ha11-to-Ha1}
    H_{a,1,1}(z_1, z_2;m)   &= (z H_{a+1,1}(z;m))[z_1 , z_2]_z
    .
\end{align}
Since $F_1(\alpha;\beta,\beta';\gamma;z_1,z_2) =
F_1(\alpha;\beta',\beta;\gamma;z_2,z_1)$, we have
\begin{align}
    H_{a,b,c}(z_1, z_2;m) &=  H_{a,c,b}(z_2, z_1;m).  
    \label{eq:Habc-vs-Hacb}
\end{align}

An interesting takeaway of the relations above is the fact that
one can reach any
$H_{ a,b} $ or $H_{ a, b,c} $ from $H_{a,1}(z;m)$ via 
basic algebraic manipulations and differentiations, while
the Gauss hypergeometric functions $H_{a,1}(z;m)$ admits fast
evaluation, as functions in $a$,  $m$ and $z$, in many CASs (computer algebra
systems), such \textsf{Mathematica}. The observations allows us to achieve
symbolic verification of the functional relations for all dimensions 
$m = 2,3, \ldots$, in part I \cite{Liu:2018aa}.

\subsection{Further recurrence relations}
\label{subsec:fur-recurrence-relations}

We now discuss how the operators $\set{\pmb\eta, \pmb\sigma_j, \blacksquare^+}$ 
act on the $H$-family. 
Since the precise functions appeared in the rearrangement lemma
are $H_{a,b}(1-y;m)$ and $H_{a,b,c}(1-y_1, 1- y_1 y_2;m)$,
we introduce the change of variables 
\begin{align}
    \label{eq:zToy-ch-var}
    z =1-y, \,\,\, z_1 = 1-y_1, \,\,\, z_2 = 1- y_1 y_2 
\end{align}
and freely use abbreviations like
\begin{align*}
    H_{ a , b}  \defeq H_{ a , b} (z;m) , \, \, \, \,  
    H_{ a , b ,c}  \defeq H_{ a , b ,c } (z_1 , z_2;m) 
\end{align*}
when no confusion arises.
For example, for $f \in C(\R_+)$, we have
\begin{align*}
    \blacksquare^+(f)(y_1, y_2) = f[y_1, y_1 y_2] = f[1-z_1, 1-z_2] = -f[z_1,z_2] 
\end{align*}
and Eq. \eqref{eq:Ha11-to-Ha1} becomes:
\begin{align*}
    H_{a,1,1} = - \blacksquare^+(z H_{a+1,1}).
\end{align*}

The contraction map $\pmb\eta$ in Eq. \eqref{eq:defn-bold-eta} is obtained by
setting $y_2\rightarrow y_1^{-1}$, that is $ y_1 y_2 \rightarrow 1$ or $z_2
\rightarrow 0$ according to Eq.  \eqref{eq:zToy-ch-var}.
It is not difficult to see $ H_{a,b,c}(z,0;m) = H_{a+c,b}(z;m)$, that is, 
\begin{align}
    \pmb\eta \brac{ H_{a,b,c} } = H_{a+c,b} .
    \label{eq:Habcz0to-Hab}
\end{align}

For the cyclic permutation $\pmb\sigma_0$, the transformations  are 
$y \rightarrow y^{-1}$ for $H_{ a , b} $
and $y_1\rightarrow (y_1 y_2)^{-1}$ and $y_2 \rightarrow y_1$
for $H_{ a, b , c} $,
which are linear fractional transformations in terms of $z$, $z_1$ and $z_2$:
\begin{align*}
    z \rightarrow \frac{-z}{1-z} ,
    \,\,\,  z_1 \rightarrow \frac{z_2}{z_2-1},\,\,\,
    z_2 \rightarrow \frac{z_2-z_1}{z_2-1}.
\end{align*}
\begin{prop}
    \label{prop:sigma0-acts-on-Habc}
    The cyclic operator $\pmb\sigma_0$ indeed permutes the indices $\set{a,b,c}$ in
    the $H$-family:
    \begin{align}
    \label{eq:sigma0-acts-on-Hab}
        \pmb\sigma_0( H_{a,b}(\cdot ;m) )(z) : = 
        H_{a,b}\brac{ \frac{1}{1-z} ;m } = (1-z)^{a+b+m/2-2} H_{b,a}(z;m)       
        .
    \end{align}
    For two variable functions,   
    \begin{align}
    \label{eq:sigma0-acts-on-Habc}
    \pmb\sigma_0( H_{a,b,c}(\cdot, \cdot ;m) )(z_1,z_2) 
    & : = 
        H_{a,b,c}\brac{ 
            \frac{z_2}{z_2-1}, \frac{z_2-z_1}{z_2-1}
        ;m }\\
   &= (1-z_2)^{a+b+c+m/2-2} H_{b,c,a}(z_1,z_2;m).   
   \nonumber
    \end{align}
\end{prop}
In particular, we have:
\begin{align*}
    H_{a,1} = (1-z)^{m/2-1} \pmb\sigma_0 (H_{1,a} ), \,\,\,
    H_{a,1,1} = (1-z_2)^{m/2} \pmb\sigma_0 (H_{1,a,1} ),
\end{align*}
where $H_{1,a}$ and $H_{1,a,1}$ are derivatives of $H_{1,1}$ and $H_{1,1,1}$,
see Eq. \eqref{eq:Hab-to-Ha1} and \eqref{eq:Habc-to-Ha11}.
As a consequence, we have obtained an algorithm to write any $H_{ a,b} $ 
(resp. $H_{ a , b, c} $) as derivatives of $H_{ 1, 1} $ (resp. $H_{ 1,1,1 } $).

\begin{cor}
    \label{cor:Hb+1ac-Hbac}
    When $a,b,c$ and $a+b+m/2-2$ are all non-zero, we have:
    \begin{align}
    \label{eq:Hb+1ac-Hbac}
        \begin{split}
        b H_{b+1,a} &= (a+b+m/2-2) H_{a,b}-a(1-z) H_{b,a+1} ,
        \\
    b H_{b+1,c,a} &=
    (a+b+c+m/2-2)H_{b,c,a} - c(1-z_1) H_{b,c+1,a} - a(1-z_2) H_{b,c,a+1} .
        \end{split}
    \end{align}
\end{cor}
\begin{proof}
    They are obtained by computing $d/dz(\pmb\sigma_0(H_{a,b}) )$ and
$\partial_{z_2}(\pmb\sigma_0(H_{a,b,c}))$ in two ways respectively. We
will proof the two-variable case and left the rest to the reader. Denote
\begin{align*}
    J_{\pmb\sigma_0}(z_1, z_2) \defeq (J_1 (z_1, z_2), J_2 (z_1, z_2) ) 
      \defeq      \left( 
            \frac{z_2}{z_2-1}, \frac{z_2-z_1}{z_2-1}
            \right)
            .
\end{align*}
Then
\begin{align}
    &\,\,
    \partial_{z_2} (\pmb\sigma_0(H_{a,b,c})(z_1,z_2) ) =
        \partial_{z_2}
        \brac{
            H_{a,b,c}
            \left( 
J_{\pmb\sigma_0}(z_1, z_2)
            \right)
        }   \nonumber \\
        =&\,\, 
        (\partial_1 H_{a,b,c}  )
\brac{
J_{\pmb\sigma_0}(z_1, z_2)
        }  (\partial_{z_2} J_1) 
        +
        (\partial_2 H_{a,b,c}  )
\brac{
J_{\pmb\sigma_0}(z_1, z_2)
        }  (\partial_{z_2} J_2) 
     \nonumber    \\
        = &\,\,
        -\frac{1}{\left(z_2-1\right){}^2}
       b \pmb\sigma_0(H_{a,b+1,c}) (z_1,z_2) +
        \frac{z_1-1}{\left(z_2-1\right){}^2}
        c \pmb\sigma_0(H_{a,b,c+1}) (z_1,z_2) 
        \nonumber \\
= &\,\,
- (1-z_2)^{a+b+c+m/2-3} \left( 
    b H_{b+1,c,a} +c (1-z_1) H_{b,c+1,a}
\right),
    \label{eq:dsingmaHabc-1}
\end{align}
where we have used the differential relations from Eq \eqref{eq:Habc-to-Ha11}: 
$\partial_1 H_{a,b,c}   = bH_{a,b+1,c}$ and
$\partial_2 H_{a,b,c}   = cH_{a,b,c+1}$
and Eq. \eqref{eq:sigma0-acts-on-Habc} to remove the $\pmb\sigma_0$ in the
third line.

On the other hand, we can first apply Eq. \eqref{eq:sigma0-acts-on-Habc} and
then carry out the differentiation:
\begin{align}
    &\,\,
    \partial_{z_2} (\pmb\sigma_0(H_{a,b,c})(z_1,z_2) ) = \partial_{z_2}
    \brac{
        (1-z_2)^{a+b+c+m/2-2} H_{b,c,a} (z_1,z_2)
    } \nonumber \\
    =&\,\,
    (1-z_2)^{a+b+c+m/2-3} \left( 
        (a+b+c+m/2-2)     H_{b,c,a} + (1-z_2) a H_{b,c,a+1}
    \right)
    .
    \label{eq:dsingmaHabc-2}
\end{align}
We complete the proof by equating Eq. \eqref{eq:dsingmaHabc-1} and  
\eqref{eq:dsingmaHabc-2}.
\end{proof}

On the other hand, the differential equations attached to $\pFq21$ and $F_1$
leads to the following recurrence relations of the $H$-family. The
hypergeometric ODEs of $\pFq21$ are transformed into
\begin{align}
    \label{eq:Hab-ODE}
     B_2 H_{a,b+2} +
    B_1 H_{a,b+1} +
     B_0 H_{a,b} =0
\end{align}
with
\begin{align*}
    B_2 = b (b+1) (1-z) z, \,\,\,
    B_1 = b (-z(a+2b+m/2-2)+a+b ),\,\,\,
    B_0 = - b(a+b+m/2-2).
\end{align*}
The PDE system of Appell's $F_1$ reads:
\begin{align}
    \label{eq:Habc-PDE-1}
    C_{2,0}    H_{a,b+2,c} + 
    C_{1,1}  H_{a,b+1,c+1}
  +  C_{1,0} H_{a,b+1,c} + C_{0,1} H_{a,b,c+1} + C_{0,0} H_{a,b,c} &= 0
    \\
    C_{0,2}    H_{a,b,c+2} + 
   \tilde C_{1,1}  H_{a,b+1,c+1}
 +  \tilde  C_{1,0} H_{a,b+1,c} + \tilde C_{0,1} H_{a,b,c+1} 
   + \tilde C_{0,0} H_{a,b,c} &= 0
    \label{eq:Habc-PDE-2}
\end{align}
where the coefficients are given as below:
\begin{align*}
    &
    C_{2,0} = b (b+1) (1-z_1) z_1, \,\,\,
    C_{1,1}= b c (1-z_1 ) z_2, \,\,\,
    C_{0,1}= -b c z_2,
    \\ &
    C_{1,0} = b \left(- z_1 \left(a+2 b+c+ m/2 -1\right)+a+b+c\right)
    ,\,\,\,
    C_{0,0} = -b \left(a+b+c+ m/2 -2\right),
\end{align*}
 and
\begin{align*}
    &
    C_{0,2} = c (c+1) (1-z_2) z_2, \,\,\,
    \tilde   C_{1,1}= b c (1-z_2 ) z_1, \,\,\,
    \tilde C_{0,1}= -b c z_1,
    \\ &
    C_{0,1} = c \left(- z_1 \left(a+ b+2 c+ m/2 -1\right)+a+b+c\right)
    ,\,\,\,
    C_{0,0} = -c \left(a+b+c+ m/2 -2\right)
    .
\end{align*}
Some remarks:
\begin{enumerate}
    \item Eq. \eqref{eq:Habc-PDE-1} and  \eqref{eq:Habc-PDE-2} are equivalent
        provided the fact that $H_{a,b,c}(z_1, z_2;m) = H_{a,c,b}(z_2, z_1;m)$.
    \item By applying $\pmb\sigma_0$ onto Eq. \eqref{eq:Hab-ODE} and 
        \eqref{eq:Habc-PDE-1}, one obtains another set of 
        relations among $\set{ H_{a,b} , H_{a+1,b}  ,H_{a+2,b} }$ and 
        $\set{ H_{a,b,c} , H_{a+1,b,c}  ,H_{a+2,b,c} }$ which provide new
        routes for the reduction of $H_{a,1}$ (resp. $H_{a,1,1}$) to $H_{1,1}$
        (resp. $H_{1,1,1}$).      

\end{enumerate}

The recurrence relations allow us to express $H_{a,b}$ (resp. $H_{a,b,c}$) as
linear combinations of $H_{1,1}$ and $H_{1,2}$ 
(resp. $H_{1,1,1}$, $H_{1,2,1}$, $H_{1,1,2}$ and $H_{1,2,2}$) with rational
function coefficients. For the two variable functions, one can further remove
$H_{1,2,2}$ using the fact that  $H_{a,1,1} = (z H_{a+1,1}) [z_1, z_2]$ is
a divided difference.  In fact, 
\begin{align*}
    H_{a,2,1} & = \partial_{z_1}\left( 
    (z H_{a+1,1}) [z_1, z_2]
\right) = (z H_{a+1,1}) [z_1,z_1, z_2] ,
\\
H_{a,1,2} & = \partial_{z_2}\left( 
    (z H_{a+1,1}) [z_1, z_2]
\right) = (z H_{a+1,1}) [z_1,z_2, z_2] .
\end{align*}
Thus
\begin{align}
    H_{a,2,2} & = \partial_{z_1} \partial_{z_2} \left( 
    (z H_{a+1,1}) [z_1, z_2]
    \right) =
    (z H_{a+1,1}) [z_1,z_1, z_2, z_2] 
    \nonumber \\
    & = \frac{
    (z H_{a+1,1}) [z_1,z_1, z_2] 
-(z H_{a+1,1}) [z_1, z_2, z_2] 
    }{z_1 - z_2} 
     = \frac{
        H_{a,2,1} - H_{a,1,2}   
    }{z_1 - z_2}.
    \label{eq:Ha22-reduction}
\end{align}

As for the dependence on the dimension $m$ (assume $m\ge 2$), we recall from
\cite{LIU2017138} that: 
\begin{align*}
    H_{a,b}(z;m+2) & = a H_{a+1,b} + b H_{a,b+1},
    \\
    H_{a,b,c}(z_1,z_2;m+2) & = a H_{a+1,b,c} + b H_{a,b+1,c} + c H_{a,b,c+1},
\end{align*}
where $H_{a,b}\defeq H_{a,b}(z;m)$ and $H_{a,b,c}\defeq H_{a,b,c}(z_1,z_2;m)$.

\subsection{Initial values and relations}
Last but not least, we discuss  
 initial values and their relations:
\begin{align}
    H_{1,1}(z;m) &= \frac{\Gamma(m/2)}{\Gamma(2)} \int_0^1 (1-zt)^{-m/2} dt
    \nonumber \\ &=
    \begin{cases}
        - \log(1-u)[0,z]_u,  & m=2 \\
         \Gamma(m/2-1) (1-u)^{-m/2+1}[0,z]_u,  & m \neq 2
    \end{cases}
    \label{eq:H11m-z}
\end{align}
Notice that, for $m\neq 2$, one can use $\pFq21$ to make sense of $H_{0,1}(z;m)
= (1-z)^{1-\frac{m}{2}} \Gamma \left(\frac{m}{2}-1\right)$ even though the
integral representation in Eq. \eqref{eq:hgeofun-onevar-K_acm} diverges and then 
$H_{1,1} = H_{0,1}[0,z]$. In other words, if we define
\begin{align}
    \label{eq:H01-defn}
    H_{0,1}(z;m) & = 
    \begin{cases}
        - \log(1-z),  & m=2 \\
         \Gamma(m/2-1) (1-z)^{-m/2+1},  & m \neq 2
    \end{cases}
    \\
    H_{0,2} (z;m) & \defeq \frac{d}{dz} H_{0,1} (z;m) = 
    \Gamma (m/2) (1-z)^{-\frac{m}{2}}  .
    \label{eq:H02-defn}
\end{align}
then
\begin{align}
    \label{eq:H11-as-ddiff-H01}
    H_{1,1} = H_{0,1}[0,z],
\end{align}
and for $H_{1,2} \defeq H_{1,2}(z;m)$:
\begin{align}
    H_{1,2} &= \frac{d}{dz} H_{1,1} = \frac{d}{dz} H_{0,1}[0,z]
    = H_{0,1}[0,z,z] 
    =  \frac{H_{0,1}[z,z] - H_{0,1}[0,z]}{ z-0}
 \nonumber   \\
    &= z^{-1}\left( 
        \frac{d}{dz} H_{0,1} - H_{1,1}
    \right)   
    = 
    z^{-1}\left( 
        H_{0,2}  - H_{1,1}
    \right)
    .
    \label{eq:H21-as-H11}
\end{align}
We shall need the following initial relation later: for $m\neq 2$
\begin{align}
    \label{eq:init-H11-H12-H02}
    m H_{1,1}+2 z H_{1,2} -2 H_{0,2}[0,z] = (m-2) H_{1,1}
    +2 H_{0,2} - 2 H_{0,2}[0,z] =0 .
\end{align}
Too see the vanishing of the right hand side, note that
for $m\neq 2$, $(1-z)H_{0,2} = (-m/2+1) H_{0,1}$, 
applying the divided difference $[0,z]$ on both sides yields
\begin{align*}
    -H_{0,2} + H_{0,2}[0,z] = \frac{2-m}{2} H_{0,1}[0,z] =
    \frac{2-m}{2} H_{1,1}.
\end{align*}

\subsection{
Verification of Theorem \ref{thm:CM-TDel-KDel}
}
\label{subsec:2nd-proof-CM-T-K}

\begin{lem}
    \label{lem:reduction-needed-one-var}
    Let us abbreviate $H_{a,b} \defeq H_{a,b}(z;m)$.
All the required recurrence relations for the one-variable family are listed as
below.
    \begin{align}
\begin{split}
    & H_{2,1}= \frac{1}{2} m H_{1,1}+(z-1) H_{1,2}, \,\,\,
    H_{3,1}= \frac{1}{4} \left((m+2) H_{2,1}+2 (z-1) H_{2,2}\right), \,\,
    \\        
& 
H_{4,1}= \frac{1}{6} \left((m+4) H_{3,1}+2 (z-1) H_{3,2}\right)
.
\end{split}
        \label{eq:H21-31-41-remove}
    \end{align}
    \begin{align}
        \label{eq:H12-13-14-remove}
        \begin{split}
&        H_{1,3}= -\frac{ ((m+4) z-4) H_{1,2} +m H_{1,1}}{4 (z-1) z},\,\,\,
H_{1,4}= -\frac{ ((m+8) z-6) H_{1,3} +(m+2) H_{1,2}}{6 (z-1) z}
\\        
        & H_{2,3}= -\frac{ ((m+6) z-6) H_{2,2}+(m+2) H_{2,1}}{4 (z-1) z} 
        \end{split}
    \end{align}

    \begin{align}
        \label{eq:H22-32-remove}
        H_{2,2}= (m/2+1)
        H_{1,2}+2 (z-1) H_{1,3}, \,\,\,
        H_{3,2}= (m/4+1) H_{2,2}+ (z-1) H_{2,3}
        .
    \end{align}
\end{lem}
\begin{proof}
The first set \eqref{eq:H21-31-41-remove} is a special case of Lemma 
Corollary \ref{cor:Hb+1ac-Hbac}. The second set \eqref{eq:H12-13-14-remove}
comes from the differential equation \eqref{eq:Hab-ODE}. 
At last, we obtain  \eqref{eq:H22-32-remove}   
by differentiating the relations of $H_{2,1}$ and $H_{3,1}$ in
\eqref{eq:H21-31-41-remove}.
\end{proof}

At $y=1$, that is $z=1-y =0$, we have:
\begin{align*}
    H_{a,b}(0;m) =  \frac{\Gamma(a+b+m/2-1)}{\Gamma(a+b+1)}
    ,
\end{align*}
then
\begin{align*}
    K_{\Delta_k}(1;m) = -H_{2,1}(0;m) + \frac4m H_{3,1}(0;m)  
    = \Gamma(m/2)(4-m)/12
\end{align*}
moreover, with Eq. \eqref{eq:H02-defn}, we rewrite:
\begin{align*}
    - \pmb\eta(K_{\Delta_k}) u^{-m/2}[1,y]_u =  \frac{4-m}{12}
    \Gamma(m/2) (1-u)^{-m/2}[0,z]_u =
    \frac{4-m}{12} H_{0,2}[0,z]
    .
\end{align*}
We recall
\begin{align}
    \label{eq:recall-KDel-TDel}
    K_{\Delta_k}   = -H_{2,1}  + \frac4m H_{3,1}    ,\,\,\, 
    T_{\Delta_k}  = 
    \frac{4-m}{12} H_{0,2}[0,z]
    -\frac{4 (1-z) H_{3,2}}{m}+\frac{2 (m+2) 
    H_{3,1}}{m}-\frac{8 H_{4,1}}{m},
\end{align}
where $K_{\Delta_k}\defeq K_{\Delta_k}(y;m) = K_{\Delta_k}(1-z;m)$, similar
meaning for $T_{\Delta_k}$, and $H_{a,b} \defeq H_{a,b}(z;m)$ as before.

\begin{prop}
        \label{prop:tobechecked-KDelvs-TDel}
    Consider the spectral functions  $K_{\Delta_k}$ and $T_{\Delta_k}$ given  in
    Eq. \eqref{eq:recall-KDel-TDel}, then for $m>2$, the relation
    \begin{align}
        \label{eq:tobechecked-KDelvs-TDel}
        \frac{m-2}{2} K_{\Delta_k} = (1+\pmb\sigma_{-m/2-1})(T_{\Delta_k})
    \end{align}
    derived in Theorem $\ref{thm:CM-TDel-KDel}$
    can be reduced to the following one among the initial values of the
    hypergeometric family, cf. Eq. \eqref{eq:init-H11-H12-H02}:
\begin{align*}
    m H_{1,1}+2 z H_{1,2} -2 H_{0,2}[0,z] = 0.
\end{align*}
For $m=2$, Eq. \eqref{eq:tobechecked-KDelvs-TDel} becomes
$(1+\pmb\sigma_{-2})(T_{\Delta_k}) =0$, which is reduced to 
Eq. \eqref{eq:H21-as-H11}:
\begin{align*}
    H_{0,2} - H_{1,1} - z H_{1,2} =0.
\end{align*}
\end{prop}
\begin{proof}
    According to Prop. \ref{prop:sigma0-acts-on-Habc}, 
    $\pmb\sigma_{-m/2-1}(T_{\Delta_k})$ is given by:
        \begin{align}
            \label{eq:TDelsigma-m-defn}
            \pmb\sigma_{-m/2-1}( T_{\Delta_k}) =\frac{4-m}{12} H_{0,2}[0,z]
   +   \frac{2 (1-z) \left((m+2) H_{1,3}
   +4 (z-1) H_{1,4}-2 H_{2,3}\right)}{m}      
   ,
        \end{align}
    here we have also used the fact that 
    $\pmb\sigma_{-m/2-1}(H_{0,2}[0,z]) = H_{0,2}[0,z]$. By repeating the
    substitutions given in Lemma \ref{lem:reduction-needed-one-var}, one can
    replace all the $H_{a,b}$ with $a,b>0$ Eqs. \eqref{eq:recall-KDel-TDel} and
    \eqref{eq:TDelsigma-m-defn} by $H_{1,1}$ and $H_{1,2}$, the result reads
    as follows:
    \begin{align*}
 (1+\pmb\sigma_{-m/2-1})(T_{\Delta_k}) - (m-2)K_{\Delta_k}/2   =
\frac{m-4}{12}  \left(
    m H_{1,1}+2 z H_{1,2} -2 H_{0,2}[0,z]
\right) =0,
    \end{align*}
 where the vanishing of right hand side was checked in 
 \eqref{eq:init-H11-H12-H02} (with the assumption $m\neq 2$). When $m=2$, Eq.
 \eqref{eq:tobechecked-KDelvs-TDel} reduces to 
 $(1+\pmb\sigma_{-2})( T_{\Delta_k})=0$. Indeed, again using the reductions
 relations in Lemma \ref{lem:reduction-needed-one-var}, we obtain
 \begin{align*}
     (1+\pmb\sigma_{-2})( T_{\Delta_k}) = \frac13 \left( 
         H_{0,2} - H_{1,1} - z H_{1,2}
     \right) =0,  
 \end{align*}
 where the vanishing is due to Eq. \eqref{eq:H21-as-H11}.
\end{proof}

In dimension $m=2$, the corresponding functional of 
\eqref{eq:tobechecked-KDelvs-TDel} is derived by the variation of the
OPS-functional and  the  function $T_{\Delta_k}$ 
is replaced by $T_{\mathrm{OPS}}$, which we recall from 
 \eqref{eq:Tzeta-varprop}:
\begin{align*}
  T_{\mathrm{OPS}} (y) &= \mathrm{I}_{T_{\zeta'}}(y)
+ \mathrm{II}_{T_{\zeta'}}(y)\\
&=
 \pmb\eta(K)
\ln[1,y] \int_0^1 \pmb\sigma_{-1}\left( 
    z^s[1,y]
\right) ds
+\frac12 \int_0^1 (z^{s}[1,y]_z)^{2}
T_{\Delta_k}(y^s;2) \ln y ds .
\end{align*}

\begin{lem}
    \label{lem:1+sigmaTOPS}
For the given function $T_{\mathrm{OPS}}$  as above, we have
\begin{align}
    \label{eq:1+sigmaTOPS}
        (1+\pmb\sigma_{-2}) (T_{\mathrm{OPS}}) (y)  =
        \pmb\eta(K) y^{-1} +
        (y-1)^{-2} \int_1^y (u-1)^2 T(u) u^{-1} du,
    \end{align}
    where $T(u) = T_{\Delta_k}(u;2)$.    
\end{lem}
\begin{proof}
   One checks that:
   \begin{align}
       \pmb\sigma_{-1}( \ln[1,y] ) = \ln[1,y], \,\,\,
       \pmb\sigma_{s}( z^s[1,y]_z ) = z^s[1,y]_z .  \,\,\,
       \label{eq:symms--+}
   \end{align}
   Also, since $\pmb\sigma_{-2}(T) = - T$, we see that
  \begin{align}
      \pmb\sigma_{2s} ( T(y^s) ) = - T(y^s), \,\,\, 
      \pmb\sigma_0( \ln(y) ) = - \ln (y) .
      \label{eq:symms---}
  \end{align}
  Therefore, $\mathrm{II}_T$ is invariant under $\pmb\sigma_{-2}$:
  \begin{align*}
      \pmb\sigma_{-2}\left( 
          ( z^s[1,y] )^2 T(y^s) \ln y
      \right) & = 
      (\pmb\sigma_s (z^s[1,y]) )^2
      \pmb\sigma_{-2s}( T(y^s) )
      \pmb\sigma_0( \ln y)  \\
     & =  
     ( z^s[1,y] )^2 T(y^s) \ln y.
  \end{align*}
 Thus 
  \begin{align*}
      (1+ \pmb\sigma_{-2}) (\mathrm{II}_T) (y) &=  2 \mathrm{II}_T (y) 
      =
  \int_0^1 (z^s[1,y]) )^2 T(y^s) \ln y  ds \\
  &= (y-1)^{-2} \int_1^y (u-1)^2 T(u) u^{-1} du
  .
  \end{align*}
  Meanwhile, to compute  $\pmb\sigma_{-2}(\mathrm{I}_T)$, 
  note that $\pmb\sigma_{-1}^2 = 1$:
  \begin{align*}
      \pmb\sigma_{-2} \left( 
          \ln[1,y]  \pmb\sigma_{-1}( z^s[1,y] )
      \right) =
      \pmb\sigma_{-1} ( \ln[1,y] )  \pmb\sigma_{-1}^2 ( z^s[1,y] )
      = \ln[1,y]    z^s[1,y] .
  \end{align*}
  It follows that
  \begin{align*}
      (1+\pmb\sigma_{-2}) (\mathrm{I}_T) (y)  &=  
      \pmb\eta(K)    \ln[1,y] \int_0^1 
      \left( 1+\pmb\sigma_{-1} \right) ( z^s[1,y] ) ds 
      = \pmb\eta(K) y^{-1}
      .
  \end{align*}
\end{proof}

\begin{prop}
        \label{prop:KDelk-TOPS-tobechecked}
    Keep notations as above. The functional relations 
    \begin{align}
        \label{eq:KDelk-TOPS-tobechecked}
        K_{\Delta_k} = (1+\pmb\sigma_{-2})( T_{\mathrm{OPS}} )
    \end{align}
    is can be reduced to the initial relation defined in Eq. \eqref{eq:H21-as-H11}:
    \begin{align}
        \label{eq:ini-required-KDelvsTOPS}
        z H_{2,1} = H_{0,2} - H_{1,1} ,\,\,\,
        H_{0,2} = (1-z)^{-1}
        .
    \end{align}
\end{prop}
\begin{proof}
    Set $K
    \defeq K_{\Delta_k}(\cdot;2)$ and $T\defeq T_{\Delta_k}(\cdot;2)$, we have
    simplified the right hand side of Eq. \eqref{eq:KDelk-TOPS-tobechecked}:
    \begin{align*}
        (1+\pmb\sigma_{-2}) (T_{\mathrm{OPS}}) (y)  =
        \pmb\eta(K) y^{-1} +
        (y-1)^{-2} \int_1^y (u-1)^2 T(u) u^{-1} du.
    \end{align*}
In particular, we see that 
$(1+\pmb\sigma_{-2}) (T_{\mathrm{OPS}}) (1)  = \pmb\eta(K) = K(1)$, thus it
suffices to check their derivatives are equal. 
We group Eq. \eqref{eq:KDelk-TOPS-tobechecked} in the following way:
\begin{align*}
 (y-1)^2    (K  - \pmb\eta(K) y^{-1}) =
 \int_1^y (u-1)^2 T(u) u^{-1} du 
\end{align*}
and then apply $d/dy$ to both sides:
\begin{align}
    \left( 
    2+ (y-1)  d/dy 
\right) 
\left( 
    K(y) - \pmb\eta(K) y^{-1}
\right) 
& = y^{-1} (y-1) T (z).
\nonumber \\
-(1-z) z^{-1}  \left( 
    2+  z d/dz
\right) 
\left( 
    K(1-z) - \pmb\eta(K) (1-z)^{-1}
\right) 
& =  T (z),
    \label{eq:D-of-KDelk-TOPS}
\end{align}
where $z= 1-y$. 
With the help of $d/dz (H_{a,b}) = b H_{a,b+1}$, we write the difference of 
two sides 
\begin{align*}
    &\,\,\, T(z) + (1-z) z^{-1}  \left( 2+  z d/dz \right) 
\left( 
    K(1-z) - \pmb\eta(K) (1-z)^{-1} 
\right) 
    \\
    = &\,\,\,
    \frac{-6 (z-1) H_{2,1}-3 (z-1) z H_{2,2}+12 z H_{4,1}-12 H_{3,1}+1}{3 z}
    \\
    =&\,\, 
    \frac{(z-1) H_{1,1}+(z-1) z H_{1,2}+1}{3 z},
\end{align*}
where the last line is obtained by applying the relations in Lemma
\ref{lem:reduction-needed-one-var}. Moreover, the vanishing of the last line is
equivalent to the conditions in \eqref{eq:ini-required-KDelvsTOPS}.
\end{proof}

\subsection{
Verification of Theorem \ref{thm:CM-KDel-HDel}
}
\label{subsec:2nd-proof-CM-H-K}
\begin{lem}
 \label{lem:reduction-needed-two-var}
    Keep the abbreviation: $H_{a,b,c} \defeq H_{a,b,c}(z_1,z_2;m)$.
    The recurrence relations will be needed in later computation are listed as
    below:
    \begin{align}
        \begin{split}
        H_{3,1,1} &= 
        \frac{1}{4} \left((m+4) H_{2,1,1}+2 \left(z_2-1\right) H_{2,1,2}+2
        \left(z_1-1\right) H_{2,2,1}\right),
        \\
        H_{2,1,2} & =
        \frac{1}{2} (m+4) H_{1,1,2}+2 \left(z_2-1\right)
        H_{1,1,3}+\left(z_1-1\right) H_{1,2,2}
        , \\
        H_{2,1,1} & =
        \frac{1}{2} (m+2) H_{1,1,1}+\left(z_2-1\right)
        H_{1,1,2}+\left(z_1-1\right) H_{1,2,1},
        \end{split}
        \label{eq:remove-H311-H212-H113}
    \end{align}
and
    \begin{align}
        \begin{split}
            H_{1,1,3}&= 
            -\frac{ \left((m+6) z_2-6\right) H_{1,1,2} +(m+2) H_{1,1,1}+2 z_1
            \left(\left(z_2-1\right) H_{1,2,2}+H_{1,2,1}\right)}{4
            \left(z_2-1\right) z_2},
     \\       
    H_{1,3,1}&= 
-\frac{ \left((m+6) z_1-6\right) H_{1,2,1} +(m+2) H_{1,1,1}+2 z_2
\left(\left(z_1-1\right) H_{1,2,2}+H_{1,1,2}\right)}{4 \left(z_1-1\right) z_1}
,
        \end{split}
        \label{eq:remove-H131-H113}
    \end{align}
    and
    \begin{align}
        H_{1,2,2}= \frac{H_{1,2,1}-H_{1,1,2}}{z_1-z_2}.
        \label{eq:remove-H122}
    \end{align}
\end{lem}

\begin{prop}
    \label{prop:varification-CM}
    Let $H_{\Delta_k}$ and $K_{\Delta_k}$ be the spectral functions given in
    \eqref{eq:b2cal-HDeltak} and \eqref{eq:b2cal-KDeltak} respectively.
    The verification of functional relation 
    \begin{align*}
        H_{\Delta_k} = (1+ \pmb\sigma_{-m/2-2} +\pmb\sigma_{-m/2-2}^2 )
      \cdot  \blacktriangle^+  (K_{\Delta_k})
    \end{align*}
    can be reduced to the following initial relations: for $m>2$, we use
    Eq. \eqref{eq:H21-as-H11} and \eqref{eq:init-H11-H12-H02}:
    \begin{align}
    \label{eq:m>2-init-relation-prop:varification-CM}
        z H_{1,2} = H_{0,2} - H_{1,1}, \,\,\,
        m H_{1,1} + 2 zH_{1,2} -2 H_{0,2}[0,z] =0, \,\,\,
    \end{align}
    While  in dimension $m=2$, we need \eqref{eq:H21-as-H11} and
    \eqref{eq:H02-defn}:
    \begin{align}
    \label{eq:m=2-init-relation-prop:varification-CM}
        z H_{1,2} = H_{0,2} - H_{1,1}, \,\,\,
         \,\,\, H_{0,2} = (1-z)^{-1}.
    \end{align}
\end{prop}
\begin{proof}
    Due to the cyclic property $\pmb\sigma_{-m/2-2}^3 =1$, we shall prove the
    equivalent functional relation 
    \begin{align*}
        \frac12 ( 1 + \pmb\sigma_{-m/2-2}) (H_{\Delta_k})
        = \blacktriangle^+(K_{\Delta_k} ).
    \end{align*}
    Start with the left hand side, 
    \begin{align*}
        \pmb\sigma_{-m/2-1} (H_{\Delta_k}) =
        \frac{2 (m+2) H_{1,1,2}+8 \left(z_2-1\right) H_{1,1,3}-4 H_{2,1,2}}{m}
    \end{align*}
 and then
 \begin{align*}
      (1+\pmb\sigma_{-m/2-1} ) ( H_{\Delta_k} ) 
     & = 
     \frac{2 }{m} 
     \left[ 
       (m+2) H_{1,2,1}+(m+2) H_{2,1,1}+2 \left(z_2-1\right)
   H_{1,2,2} \right. \\
   &  + \left.
   4 \left(z_1-1\right) H_{1,3,1}+2 \left(z_1-1\right) H_{2,2,1}-4
     H_{3,1,1}
 \right],
 \end{align*}
 which admits further simplification due to the recurrence relations listed in
 Lemma \ref{lem:reduction-needed-two-var}: 
 \begin{align*}
     -\frac{2 \left((m+2) \left(z_1+z_2\right) H_{1,1,1}+2 \left(z_2-1\right)
     \left(2 z_1+z_2\right) H_{1,1,2}+2 \left(z_1-1\right) \left(z_1+2
     z_2\right) H_{1,2,1}\right)}{m z_1 z_2}.
 \end{align*}
 Notice that
 \begin{align*}
     H_{1,1,1} = (z H_{2,1} ) [z_1, z_2]_z, \,\,\,
     H_{1,2,1} = \partial_{z_1} H_{1,1,1} , \,\,\,
     H_{1,1,2} = \partial_{z_2} H_{1,1,1} , 
 \end{align*}
 and apply again the recurrence relations in Lemma
 \ref{lem:reduction-needed-one-var} for one-variable functions, we can express the
 difference, which we would like to prove to be zero, in terms of $H_{1,1}$ and
 $H_{1,2}$: 
 \begin{align*}
     &\,\, \frac12 (1+\pmb\sigma_{-m/2-2}) (H_{\Delta_k} ) -
     \blacktriangle^+(K_{\Delta_k})  \\
     =&\,\, 
     -\frac{z_1+z_2}{2 z_1  z_2} 
     \left( 
     m z H_{1,1}+2 z^2 H_{1,2}-2 z H_{1,2}-2 H_{1,1} 
     \right) 
     [z_1,z_2]_z.
 \end{align*}
 It remains to show that  
 $m z H_{1,1}+2 z^2 H_{1,2}-2 z H_{1,2}-2 H_{1,1} $ is a constant function (in $z$),
 therefore becomes zero after applying the divided difference. Indeed, with the
 initial
 relations in Eq. \eqref{eq:m>2-init-relation-prop:varification-CM} for
 dimension $m>2$, one  sees that 
 \begin{align*}
     (m z H_{1,1}+2 z^2 H_{1,2})  -(2 z H_{1,2} + 2 H_{1,1}) 
         =2 H_{0,2}[0,z] - 2 H_{0,2}
         =-2 H_{0,2}(0;m) ,
 \end{align*}
 and when $m=2$, the initial relations in Eq.
 \eqref{eq:m=2-init-relation-prop:varification-CM} implies that
\begin{align*}
     m z H_{1,1}+2 z^2 H_{1,2}-2 z H_{1,2}-2 H_{1,1} = 
     2 \left(z_1-1\right) \left(z_1 H_{1,2}+H_{1,1}\right) = 
     2 \left(z_1-1\right) H_{0,2} =-2 .
 \end{align*}
\end{proof}